\documentclass{CSML}

\def\dOi{13(2:1)2017}
\lmcsheading%
{\dOi}
{1--57}
{}
{}
{Jan.~15, 2016}
{Apr.~\phantom07, 2017}
{}

\usepackage[utf8]{inputenc}

\usepackage{amsmath}
\usepackage{amssymb}
\usepackage{thmtools}
\usepackage{thm-restate}
\usepackage{proof}
\usepackage[usenames,dvipsnames]{xcolor}
\usepackage{stmaryrd}
\usepackage{graphicx}
\usepackage{xparse}
\usepackage{listings}
\usepackage{color}
\usepackage{xspace}
\usepackage{url}
\usepackage{float}
\usepackage{refcount}
\usepackage{inconsolata}

\let\theoremstyle\relax

\usepackage{amsthm}
\usepackage{hyperref}
\usepackage{cleveref}

\newtheorem{theorem}{Theorem}[section]
\newtheorem{corollary}[theorem]{Corollary}
\newtheorem{lemma}[theorem]{Lemma}

\newtheorem{definition}[theorem]{Definition}
\newtheorem{remark}[theorem]{Remark}
\newtheorem{example}[theorem]{Example}

\def\buyer{{\role{Buyer}}}
\def\seller{{\role{Seller}}}
\def\bank{{\role{Bank}}}

\def\AdaptRule{Up}
\def\NoAdaptRule{NoUp}
\def\NoAdaptLabel{\mbox{\normalfont\textbf{no-up}}}

\def\net{\mathcal{N}}
\def\event{\func{events}}
\def\devent{\func{events}}

\def\rules{\mathbf{I}}
\def\APOC{DPOC\xspace}
\def\AIOC{DIOC\xspace}
\def\AIOCJ{AIOCJ\xspace}

\def \mid {\; | \;}

\newcommand{\one}{\mathbf{1}}
\newcommand{\zero}{\mathbf{0}}
\newcommand{\tick}{\surd}

\newcommand{\fromC}{\mathbf{from}}
\newcommand{\toC}{\mathbf{to}}
\newcommand{\ifC}{\mathbf{if}}
\newcommand{\elseC}{\mathbf{else}}
\newcommand{\whileC}{\mathbf{while}}
\newcommand{\scopeC}{\mathbf{scope}}
\newcommand{\rolesC}{\mathbf{roles}}
\newcommand{\extFunc}[1]{\mathtt{#1}}
\newcommand{\trueC}{\mathbf{true}}
\newcommand{\falseC}{\mathbf{false}}
\newcommand{\andC}{\mathbf{and}}
\newcommand{\okC}{\mathbf{ok}}
\newcommand{\noC}{\mathbf{no}}

\newcommand{\state}{\Sigma}

\newcommand{\substLabel}[3]{[#1 / #2,#3]}
\newcommand{\substAPOC}[2]{[#1 / #2]}
\newcommand{\comm}[5]{#1 : #2(#3) \rightarrow #4(#5)}

\newcommand{\commLabel}[5]{#1 : #2(#3) \rightarrow #4(#5)}
\newcommand{\co}[2]{\overline{#1} \at {#2}}
\newcommand{\cout}[3]{#1 : #2\; \toC\; {#3}}
\newcommand{\coutLabel}[3]{\overline{#1}\langle#2\rangle \at {#3}}
\newcommand{\ci}[2]{#1 \at #2}
\newcommand{\cinp}[3]{#1 : #2 \;\fromC\; #3}
\newcommand{\cinpLabel}[4]{#1(#2 \leftarrow #3) \at #4 }

\newcommand{\upd}{\func{upd}}
\newcommand{\Prop}{\func{compl}}
\newcommand{\ssim}{\func{clean}}

\newcommand{\gram}{::=}

\newcommand{\proj}{\func{proj}}

\newcommand{\sset}[2]{\left\{~#1  \left |
                    \begin{array}{l}#2\end{array} \right.     \right\}}
\newcommand{\ruleName}[1]{[{\sc #1}]}
\newcommand{\rulebox}[1]{\!\!\fcolorbox{gray}{white}{\mbox{$#1$}}}
\newcommand{\auxiliaryCode}[1]{\color{gray!99}#1}
\newcommand{\arule}[3]{#3}
\newcommand{\cnd}{\mathcal{C}}

\NewDocumentCommand\scope{mmmmg}{\scopeC \ \at #2 \ \{ #3 \}}
\NewDocumentCommand\pscope{mmmmmg}{\idx{#1} \scopeC \ \at #3 \ \{ #4 \}\IfNoValueTF{#6}{}{\ \mathbf{roles} \ \{ #6 \}}}
\newcommand{\psscope}[4]{\idx{#1} \scopeC \ \at #3 \ \{ #4 \}}
\def\grass#1{[\![#1]\!]}
\def\ev{\varepsilon}
\def\leqapoc{\leq_{\APOC}}
\def\leqaioc{\leq_{\AIOC}}

\newcommand{\code}[1]{\mbox{$\mathtt{#1}$}}
\newcommand{\func}[1]{\mbox{$\mathsf{#1}$}}
\newcommand{\did}[2]{	\left\lfloor^\text{\raisebox{-1pt}{#1}}|_\textsc{#2}\right\rceil}
\newcommand{\Div}{|}

\newcommand{\seqOp}{;}
\newcommand{\parOpI}{|}
\newcommand{\parOpP}{|}
\newcommand{\assign}[3]{#1 \at #2 = #3 }
\newcommand{\at}{\mathtt{@}}
\newcommand{\rolesFuncPair}[2]{#1 \rightarrow #2}

\newcommand{\adapt}[4]{\scopeC \; \at #4 \; \{ #1 \}}

\newcommand{\ifthen}[3]{\ifC \; #1  \; \{ #2 \} \; \elseC \; \{ #3 
\}}
\newcommand{\ifthenKey}[4]{\idx{#4} \ifC \; #1  \; \{ #2 \} \; \elseC \; \{ #3 
\}}
\newcommand{\while}[2]{\whileC \; #1 \;  \{ #2 \}}
\newcommand{\whileKey}[3]{\idx{#3} \whileC \; #1 \;  \{ #2 \}}

\newcommand{\arro}[1]{\xrightarrow[]{#1}}

\newcommand{\tuple}[1]{\left\langle #1\right\rangle}

\DeclareMathOperator{\transI}{\func{transI}}

\DeclareMathOperator{\transF}{\func{transF}}

\DeclareMathOperator{\rolesFunc}{\func{roles}}

\DeclareMathOperator{\bisim}{\sim}

\newcommand{\roleExec}[2]{\left(#1\right)_{\role{#2}}}
\newcommand{\role}[1]{\mathbf{#1}}

\definecolor{color:keyword}{RGB}{150,0,150}
\definecolor{color:comment}{RGB}{0,128,0}
\definecolor{color:string}{RGB}{0,0,255}

\lstdefinelanguage{MyChor}{
morekeywords={
  while, if, else, scope, true, false, not, getInput, and, 
  prop, from, to, via, on, do, and, roles, null, rule, N, 
  aioc, preamble, starter, include, or, E, with
},
sensitive=true,
morecomment=[l]{//},
morecomment=[s]{/*}{*/},
morestring=[b]",
otherkeywords={;,|,@,!}
}

\newcommand{\setMyChor}{
\lstset{
  language=MyChor,
  basicstyle=\small\ttfamily
  }
}

\newcommand{\toolplugin}{\func{AIOCJ-ecl}}
\newcommand{\toolmid}{\func{AIOCJ-mid}}
\newcommand{\ambient}{\Sigma,\rules}
\newcommand{\ambientN}{\rules}

\newcommand{\idxSign}[1]{\mathbf{#1}}
\newcommand{\idx}[1]{\idxSign{#1}\hspace{-3pt}:}
\newcommand{\idxTrueSign}[1]{\idxSign {#1}_{\mathtt{T}}}
\newcommand{\idxFalseSign}[1]{\idxSign {#1}_{\mathtt{F}}}
\newcommand{\idxIfRecvSign}[1]{\idxSign {#1}_?}
\newcommand{\idxCloseSign}[1]{\idxSign {#1}_{\mathtt{C}}}
\newcommand{\idxTrue}[1]{\idx{\idxTrueSign{#1}}}
\newcommand{\idxFalse}[1]{\idx{\idxFalseSign{#1}}}
\newcommand{\idxIfRecv}[1]{\idx{\idxIfRecvSign{#1}}}
\newcommand{\idxClose}[1]{\idx{\idxCloseSign{#1}}}
\newcommand{\pnum}[1]{\raisebox{-1pt}{\large \textcircled{\raisebox{1.1pt} {\scriptsize #1}}}}

\newcommand{\AuxSb}{sb}
\newcommand{\AuxSe}{se}
\newcommand{\AuxIf}{cnd}

\newcommand{\AuxWb}{wb}
\newcommand{\AuxWe}{we}
\newcounter{NotationCounter}

\begin{document}

\title[Dynamic Choreographies]{Dynamic Choreographies: Theory and Implementation}

\author[M.~Dalla Preda]{Mila Dalla Preda\rsuper a}
\address{{\lsuper a}Department of Computer Science, University of Verona}
\email{mila.dallapreda@univr.it}

\author[M. Gabbrielli]{Maurizio Gabbrielli\rsuper b}
\address{{\lsuper{b,c,d}}Department of Computer Science and Engineering,
 University of Bologna/INRIA}
\email{\{gabbri, sgiallor, lanese\}@cs.unibo.it}

\author[S.~Giallorenzo]{Saverio Giallorenzo\rsuper c}
\address{\vspace{-18 pt}}
\thanks{{\lsuper c}Supported by the EU EIT Digital project \emph{SMAll}}
\author[I.~Lanese]{Ivan Lanese\rsuper d}
\address{\vspace{-18 pt}}
\thanks{{\lsuper d}Supported by the GNCS group of INdAM via project 
   \emph{Logica, Automi e Giochi per Sistemi Auto-adattivi}}

\author[J.~Mauro]{Jacopo Mauro\rsuper e}
\address{{\lsuper e}Department of Informatics, University of Oslo}
\email{jacopom@ifi.uio.no}
\thanks{{\lsuper e}Supported by the EU project FP7-644298
   \emph{HyVar: Scalable Hybrid Variability for Distributed, Evolving Software
Systems}}


\begin{abstract}
Programming distributed applications free from communication deadlocks and
race conditions is complex. Preserving these properties when applications are
updated at runtime is even harder.
We present a choreographic approach for programming updatable, distributed
applications. We define a choreography language, called Dynamic
Interaction-Oriented Choreography (\AIOC{}), that allows the programmer to
specify, from a global viewpoint, which parts of the application can be
updated. At runtime, these parts may be replaced by new \AIOC{} fragments from
outside the application. \AIOC{} programs are compiled, generating code for
each participant in a process-level language called Dynamic Process-Oriented
Choreographies (\APOC{}). We prove that \APOC{} distributed applications
generated from \AIOC{} specifications are deadlock free and race free and that
these properties hold also after any runtime update.
We instantiate the theoretical model above into a programming framework called
Adaptable Interaction-Oriented Choreographies in Jolie (AIOCJ)
that comprises an integrated development environment, a compiler from an
extension of \AIOC{}s to distributed Jolie programs, and a runtime environment
to support their execution.
\end{abstract}

\maketitle

\section{Introduction}\label{sec:intro}

Programming distributed applications is an error-prone activity. Participants
send and receive messages but, if the application is badly programmed, they may
get stuck waiting for messages that never arrive (communication deadlock), or
they may receive messages in an unexpected order, depending on the speed of the
other participants and of the network (races).

Recently, language-based approaches have been proposed to tackle the
complexity of programming concurrent and distributed applications. These
approaches provide high-level primitives that avoid by construction some of
the risks of concurrent programming. Some renowned examples are the
ownership/borrowing mechanisms of Rust~\cite{rust} and the \texttt{separate}
keyword of SCOOP~\cite{scoop}. In these settings most of the work
needed to ensure a correct behaviour is done by the language compiler and
runtime support. The use of these languages requires a conceptual shift from
traditional ones, but reduces times and costs of development, testing, and
maintenance by avoiding some of the most common programming errors.

We propose an approach based on \emph{choreographic
programming}~\cite{hondaESOPext,poplmontesi,SEFM08,scribble,Montesi15}
following a similar philosophy. A choreography is a global description of a
distributed application expressed as the composition of the expected
interactions between its components. Based on such an abstraction, we present
high-level primitives that avoid some of the main errors linked to programming
communication-centred concurrent and distributed applications that can be
updated at run time.

For an informal introduction of our approach, let us consider a simple
application where a $\buyer$ asks the price of some product to a
$\seller$.
\vspace{5pt}
\begin{lstlisting}[mathescape=true,numbers=left,xleftmargin=120pt]
$priceReq : \buyer( prod ) \rightarrow \seller( order );$
$order\_price@\seller = \extFunc{getPrice}( order );$
$\mathit{offer} : \seller( order\_price ) \rightarrow \buyer(prod\_price)$
\end{lstlisting}
\vspace{5pt}
A unique choreographic program describes the behaviour of multiple
participants, here the $\buyer$ and the $\seller$. 
The first line specifies that the $\role{Buyer}$ sends, along channel
$priceReq$, the name of the desired product $prod$ to the $\role{Seller}$,
which stores it in its local variable $order$. At Line 2, the $\role{Seller}$
computes the price of the product with function $\extFunc{getPrice}$. Finally,
at Line 3 the $\role{Seller}$ sends the computed price to the $\role{Buyer}$
on channel $\mathit{offer}$. The $\role{Buyer}$ stores the  value in its local
variable $prod\_price$.

Choreographic languages focus on describing \emph{message-based interactions}
between distributed participants. As shown by the example, the distinctive
feature of choreographies is that communications are atomic entities, i.e.,
not split into send and receive actions. This makes impossible to write
programs with common errors of concurrent and distributed applications like
deadlocks or race conditions.
However, a choreographic description is not directly executable. Indeed,
choreographies describe atomic interactions from a global point of view whilst
executable programs are written in lower level languages (like Java or C) that
support only local behaviours and communication/synchronisation primitives
such as message send and message receive.
Hence, to run a choreographic description, we have to compile it into a set of
lower level programs. The correct compilation of choreographic descriptions is
one of the main challenges of choreographic languages, yet the ease of
development and the strong guarantees of deadlock and race freedom make it a
worth endeavour.

In this work, we take this challenge one step further: we consider
\emph{updatable} distributed applications whose code can change dynamically,
i.e., while the application is running. In particular, the applications we
consider can integrate external code at runtime. Such a feature, tricky in a
sequential setting and even more in a distributed one, has countless uses: to
deal with emergency requirements, to cope with rules and requirements
depending on contextual properties or to improve and specialise the
application to user preferences.
We propose a general mechanism to structure application updates. Inside
applications, we delimit blocks of code, called \emph{scope}s, that may be
dynamically replaced by new blocks of code, called \emph{update}s. Updates are
loosely related to scopes: it is not necessary to know the details of the
behaviour of updates when writing scopes, and updates may be written while
applications are running.

To illustrate our approach, let us consider the previous example and let us
suppose that we would like to introduce new commercial offers at runtime, e.g.,
to provide a discount on the computed prices. Since we want to be able to
change how prices are computed, we enclose Lines 2--3 of the example within a
scope, as shown below.
\vspace{5pt}
\begin{lstlisting}[mathescape=true,numbers=left,xleftmargin=100pt]
$priceReq : \buyer( prod ) \rightarrow \seller( order );$
$\scopeC\ @\seller \{$
  $order\_price@\seller = \extFunc{getPrice}( order );$
  $\mathit{offer} : \seller( order\_price ) \rightarrow \buyer(prod\_price)$
$\}$
\end{lstlisting}
\vspace{5pt}
In essence, a scope is a delimiter that defines which part of the application
can be updated. Each scope identifies a \emph{coordinator} of the update,
i.e., the participant entitled to ensure that either none of the participants
updates, or they all apply the same update. In the example above, the
coordinator of the update is the $\seller$ (Line 2).

Since now the code includes a scope, we can introduce runtime updates. Let us
assume that the $\seller$ issued a fidelity card to some frequent customers and
(s)he wants to update the application to let $\buyer$s insert their fidelity
card to get a discount. The update in Figure~\ref{fig:rule_price_inquiry}
answers this business need.
\begin{figure}[b]
\begin{center}
\begin{lstlisting}[mathescape=true,numbers=left,xleftmargin=100pt]
$cardReq : \seller( null ) \rightarrow \buyer( \_ );$
$card\_id@\buyer = \extFunc{getInput()};$
$cardRes : \buyer( card\_id ) \rightarrow \seller( buyer\_id );$
$\mathtt{if} \ \extFunc{isValid}( buyer\_id )@\seller \ \{$
  $order\_price@\seller = \extFunc{getPrice}( order ) * 0.9 $
$\} \ \mathtt{else} \ \{$
  $order\_price@\seller = \extFunc{getPrice}( order )$
$\};$
$\mathit{offer} : \seller( order\_price ) \rightarrow \buyer( prod\_price )$
\end{lstlisting}
\end{center}
\caption{Fidelity Card Update. \label{fig:rule_price_inquiry}}
\end{figure}
At runtime, if the $\seller$ (which is the coordinator of the update) applies
the update in Figure~\ref{fig:rule_price_inquiry}, the code of the update
replaces the scope. When this new code executes, the $\buyer$ sends his/her
$card\_id$ to the $\seller$. If the $card\_id$ is valid, the $\seller$ issues
a 10\% discount on the price of the selected good, otherwise it reports the
standard price to the $\buyer$.

We remark that expressing the behaviour described above using lower level
languages that lack a dedicated support for distributed runtime updates (e.g.,
Java, C, etc.) is extremely error prone. For instance, in the example above,
let us consider the case in which the $\buyer$ is updated first and it starts
the execution whilst the $\seller$ has not been updated yet. The $\buyer$
waits for a message from the $\seller$ on channel $cardReq$, whilst the
$\seller$ sends a message on channel $\mathit{offer}$. Thus, the application
is deadlocked. Furthermore, in our setting the available updates may change at
any time, posing an additional challenge.
To avoid errors, participants must select the same update and, since updates
appear and disappear at runtime, they must also be able to retrieve the same
selected update.

Since at choreographic level updates are applied atomically to all the
involved participants, these problems cannot arise if both the original
application and the updates are described as choreographies.
However, to execute our choreographic specifications we need to provide a
correct compilation from choreographies to lower level, executable programs.
Such task is very challenging and, in particular, we need to make sure that
our compilation generates correct behaviours that avoid inconsistencies on
updates.

\paragraph{\emph{Contributions}.}

In this paper, we present a correctness-by-construction approach to solve the
problem of dynamic updates of distributed applications. We provide a general
solution that comprises:
\begin{itemize}
\item the definition of a \emph{Dynamic Interaction-Oriented Choreography}
language, called \AIOC{}, to program distributed applications and supporting
runtime code update (\S~\ref{sec:aioc});
\item the definition of a \emph{Dynamic Process-Oriented Choreography}
language, called \APOC{}. \APOC{}s are based on standard send and receive
primitives but they are tailored to program updatable systems.
We introduce
\APOC{}s to describe implementations corresponding to \AIOC{}s.
(\S~\ref{sec:apoc});
\item the definition of a behaviour-preserving \emph{projection function} to
compile \AIOC{}s into \APOC{}s (\S~\ref{sec:endpoint_projection});
\item the \emph{proof of the correctness} of the projection function
  (\S~\ref{sec:corr}). Correctness is guaranteed even in a scenario where the
  set of available updates dynamically changes, at any moment and without
  notice;
 \item one instantiation of our theory into a \emph{development framework} for
  adaptive distributed applications called AIOCJ (\S~\ref{sec:aiocj}). In
  AIOCJ updates are embodied into adaptation rules, whose application is not
  purely non-deterministic (as in \AIOC{}s), but depends on the state of the
  system and of its environment. AIOCJ comprises an Integrated Development
  Environment, a compiler from choreographies to executable programs, and a
  runtime environment to support their execution and update.
\end{itemize}

\noindent
This paper integrates and extends material from~\cite{coordination_paper},
which outlines the theoretical aspects, and \cite{sle_paper}, which describes
the AIOCJ framework. Main extensions include the full semantics of \AIOC{} and
\APOC{}, detailed proofs and technical definitions, and a thorough description
of the example. Furthermore, both the presentation and the technical
development have been deeply revised and refined.

\section{Dynamic Interaction-Oriented Choreographies}\label{sec:aioc}

In this section we introduce  the syntax of \AIOC{}s, we illustrate the
constructs of the \AIOC{} language with a comprehensive example, and we finally
present the semantics of \AIOC{}s.

\subsection{\AIOC{} Syntax}\label{sub:aioc_syntax}

\AIOC{}s rely on a set of $\mathit{Roles}$, ranged over by
$\role{R},\role{S},\dots$, to identify the participants in the choreography. We
call them roles to highlight that they have a specific duty in the
choreography. Each role has its local state.

Roles exchange messages over public channels, also called
\emph{operations}, ranged over by $o$.  We denote with $\mathit{Expr}$
the set of expressions, ranged over by $e$. We deliberately do not give a
formal definition of expressions and of their typing, since our results do not
depend on it. We only require that expressions include at least values,
belonging to a set $\mathit{Val}$ ranged over by $v$, and variables, belonging
to a set $\mathit{Var}$ ranged over by $x,y,\dots$. We also assume a set of
boolean expressions ranged over by $b$.

The syntax of \emph{\AIOC{} processes}, ranged over by ${\mathcal I}, {\mathcal
I}', \ldots$, is defined as follows:
$$
\begin{array}{rll}
\mathcal I \gram 
  & \comm{o}{\role R}{e}{\role S}{x}  & \textit{(interaction)}      \\[2pt]
\Div & \mathcal I \seqOp{\mathcal I}'   & \textit{(sequence)}       \\[2pt]
\Div & \mathcal I \parOpI{\mathcal I}'  & \textit{(parallel)}       \\[2pt]
\Div & \assign{x}{\role R}{e}           & \textit{(assignment)}     \\[2pt]
\Div & \one                             & \textit{(inaction)}       \\[2pt]
\end{array}
\begin{array}{lll}
\Div & \zero                            & \textit{(end)}            \\[2pt]
\Div & \ifthen{b \at \role{R}}{\mathcal{I}}{\mathcal{I}'} 
                                        & \textit{(conditional)}    \\[2pt]
\Div & \while{b \at \role{R}}{\mathcal{I}}
                                        & \textit{(while)}          \\[2pt]
\Div & \scope{l}{\role{R}}{\mathcal{I}}{}
                                        & \textit{(scope)}
\vspace{1em}                                      
\end{array}
$$
Interaction $\comm{o}{\role R}{e}{\role S}{x}$ means that role $\role{R}$ sends
a message on operation $o$ to role $\role{S}$ (we require $\role{R}
\neq \role{S}$). The sent value is obtained by evaluating expression $e$ in the
local state of $\role{R}$ and it is then stored in the local variable $x$ of $\role{S}$.
Processes ${\mathcal I}\seqOp{\mathcal I}'$ and ${\mathcal I}\parOpI{
  \mathcal I}'$ denote sequential and parallel composition, respectively.
Assignment $\assign{x}{\role R}{e}$ assigns the evaluation of expression $e$ in
the local state of $\role R$ to its local variable $x$. 
The empty process
$\one$ defines a \AIOC{} that can only terminate. $\zero$ represents a
terminated \AIOC{}. It is needed for the definition of the operational
semantics and it is not intended to be used by the programmer.
We call \emph{initial} a \AIOC{} process where $\zero$ never occurs.
The conditional $\ifthen{b \at \role{R}}{\mathcal{I}}{\mathcal{I}'}$ and the
iteration $\while{b \at \role{R}}{\mathcal{I}}$ are guarded by the evaluation
of the boolean expression $b$ in the local state of $\role{R}$.The construct $\scope{l}{\role R}{\mathcal{I}}{\Delta}$ delimits a subterm
$\mathcal{I}$ of the \AIOC{} process that may be updated in the future.
In $\scope{l}{\role R}{\mathcal{I}}{\Delta}$, role $\role R$ is the
coordinator of the update: it ensures that either none of the participants
update, or they all apply the same update.

\vspace{1em}
\paragraph{\emph{A Running Example.}}
\label{sub:example}
We report in Figure~\ref{fig:purchase_scenario} a running example of a \AIOC{}
process that extends the one presented in the Introduction: the example
features a $\buyer$ that orders a product from a $\seller$, and a $\bank$ that
supports the payment from the $\buyer$ to the $\seller$.
The \AIOC{} process describes the behaviour of all of them. In this sense, the
\AIOC{} implements a protocol they agreed upon to integrate their business.
The \AIOC{} protocol also interacts with functionalities that are only
available at the premises of some of the participants (for instance, the
function $\extFunc{getPrice}$ is provided by the $\seller$ IT system and may
be only locally available). For this reason, it is important that the
execution of the \AIOC{} is distributed across the participants.

\begin{figure}
\begin{center}
\begin{lstlisting}[mathescape=true,numbers=left,xleftmargin=80pt]
$price\_ok@\buyer = \falseC;$
$continue@\buyer = \trueC;$
$\whileC (\ !price\_ok \ \andC \ continue \ )@\buyer \{$
  $prod@\buyer = \extFunc{getInput()};$
  $priceReq : \buyer( prod ) \rightarrow \seller( order );$
  $\scopeC\ @\seller \{$
    $order\_price@\seller = \extFunc{getPrice}( order );$
    $\mathit{offer} : \seller( order\_price ) \rightarrow \buyer(prod\_price)$   
  $\};$
  $price\_ok@\buyer = \extFunc{getInput()};$
  $\ifC (\ !price\_ok\ )@\buyer \{$
    $\ continue@\buyer = \extFunc{getInput()}$
  $\} $
$\};$
$\ifC (\ price\_ok \ )@\buyer \{$
  $payReq : \seller(\ \extFunc{payDesc}(order\_price) \ ) \rightarrow \bank( desc );$
  $\scopeC\ @\bank \{$
    $pay : \buyer(\ \extFunc{payAuth}(prod\_price)\ ) \rightarrow \bank( auth )$
  $\};$
  $payment\_ok@\bank = \extFunc{makePayment}( desc, auth );$
  $\ifC (\ payment\_ok \ )@\bank \{$
    $confirm : \bank( null ) \rightarrow \seller( \_ )$
    $|$
    $confirm : \bank( null ) \rightarrow \buyer( \_ )$
  $\} \ \elseC \ \{$
    $abort : \bank( null ) \rightarrow \buyer( \_ )$
  $\}$
$\}$
\end{lstlisting}
\end{center}
\caption{\AIOC{} process for Purchase Scenario.\label{fig:purchase_scenario}}
\end{figure}

At Lines 1--2 the $\buyer$ initialises its local variables $price\_ok$ and
$continue$. These variables control the while loop used by the $\buyer$ to ask
the price of some product to the $\seller$. 
The loop execution is controlled by the $\buyer$, but it impacts also the
behaviour of the $\seller$. We will see in \S~\ref{sec:endpoint_projection}
that this is done using automatically generated auxiliary communications. A
similar approach is used to manage conditionals.
At Line 4, the $\buyer$ takes the name of the $prod$uct from the user with
function $\extFunc{getInput}$, which models interaction with the user, and
proceeds to send it to the $\seller$ on operation $priceReq$ (Line 5).
The $\seller$ computes the price of the product calling the function
$\extFunc{getPrice}$ (Line 7) and, via operation $\mathit{offer}$, it sends
the price to the $\buyer$ (Line 8), that stores it in a local variable
$prod\_price$. These last two operations are performed within a scope and
therefore they can be updated at runtime to implement some new business
policies (e.g., discounts).
At Lines 10--12 the $\buyer$ checks whether the user is willing to buy the
product, and, if (s)he is not interested, whether (s)he wants to ask prices
for other products. If the $\buyer$ accepts the offer of the $\seller$, the
$\seller$ sends to the $\bank$ the payment details (Line 16), computed using
function $\extFunc{payDesc}$. Next, the $\buyer$ authorises the payment via
operation $pay$, computing the payment authorisation form using function
$\extFunc{payAuth}$. Since the payment may be critical for security reasons,
the related communication is enclosed in a scope (Lines 17--19), thus allowing
the introduction of more refined procedures later on. After the scope
successfully terminates, at Line 20 the $\bank$ locally executes the actual
payment by calling function $\extFunc{makePayment}$. The function
$\extFunc{makePayment}$ abstracts an invocation to the $\bank$ IT system. We
show in
\S~\ref{sec:aiocj} that, using functions, one can integrate existing
service-oriented software into a choreographic program.\enlargethispage{\baselineskip}

Finally, the $\bank$ acknowledges the payment to the $\seller$ and the
$\buyer$ in parallel (Lines 22--24). If the payment is not successful, the
failure is notified explicitly only to the $\buyer$.
Note that at Lines 1--2 the annotation $@\buyer$ means that the
variables $price\_ok$ and $continue$ belong to the $\buyer$.
Similarly, at Line 3, the annotation $@\buyer$ means that the guard of
the while loop is evaluated by the $\buyer$. The term $@\seller$ at
Line 6 is part of the scope construct and indicates the $\seller$ as
coordinator of the update.

\subsection{Annotated \AIOC{}s and their Semantics}

\label{sub:annotated_aioc_semantics}
In the remainder of the paper, we define our results on an annotated
version of the \AIOC{} syntax.Annotations are numerical indexes $\idxSign{i} \in \mathbb{N}$ assigned to
\AIOC{} constructs. 
We only require indexes to be distinct. Any annotation that satisfies this
requirement provides the same result. Indeed, programmers do not need to
annotate \AIOC{}s: the annotation with indexes is mechanisable and can be
performed by the language compiler\footnote{In fact, the \AIOCJ compiler
implements such a feature.}. Indexes are used both in the proof of our results
and in the projection to avoid interferences between different constructs.
From now on we consider only well-annotated \AIOC{}s, defined as follows.
\begin{definition}[Well-annotated \AIOC{}]\label{def:annAIOC}\emph{
\emph{Annotated \AIOC{} processes} are obtained by indexing every interaction,
assignment, conditional, while loop, and scope in a \AIOC{} process with a
positive natural number $\idxSign i \in \mathbb{N}$, resulting in the following
grammar:
$$
\begin{array}{rl}
\mathcal I \gram & \idx{i} \comm{o}{\role R}{e}{\role S}{x}
\\
\Div & \mathcal I \seqOp{\mathcal I}'
\\
\Div & \mathcal I \parOpI{\mathcal I}'
\\
\Div & \idx{i} \assign{x}{\role R}{e}
\\
\Div & \one                                         
\end{array}
\qquad
\begin{array}{rl}
\Div & \zero
\\
\Div & \idx{i} \ifthen{b \at \role{R}}{\mathcal{I}}{\mathcal{I}'}
\\
\Div & \idx{i} \while{b \at \role{R}}{\mathcal{I}}
\\
\Div & \idx{i} \scope{l}{\role{R}}{\mathcal{I}}{}
\\
\vphantom{C}
\end{array}
$$
A \AIOC{} process is \emph{well annotated} if all its indexes are distinct.
}\end{definition}

\vspace{5pt}
\AIOC{} processes do not execute in isolation: they are equipped with a
\emph{global state} $\Sigma$ and a set of available updates $\rules$, i.e.,
a set of \AIOC{}s that may replace scopes. Set $\rules$ may change at runtime.
A global state $\Sigma$ is a map that defines the value $v$ of each variable
$x$ in a given role $\role R$, namely $\Sigma: \mathit{Roles} \times
\mathit{Var} \rightarrow \mathit{Val}$.
The local state of role $\role R$ is denoted as $\Sigma_{\role R}:\mathit{Var}
\rightarrow \mathit{Val}$ and it verifies that $\forall x \in \mathit{Var}: \;
\Sigma(\role R,x) = \Sigma_{\role R}(x)$. Expressions are always evaluated by
a given role $\role R$: we denote the evaluation of expression $e$ in local
state $\Sigma_{\role R}$ as $\grass{e}_{\Sigma_{\role R}}$. We assume that
$\grass{e}_{\Sigma_{\role R}}$ is always defined (e.g., an error value is
given as a result if evaluation is not possible) and that for each boolean
expression $b$, $\grass{b}_{\Sigma_{\role R}}$ is either $\trueC$ or
$\falseC$.

\begin{remark}\label{remark:assumption_theory}\emph{
  The above assumption on expressions is needed for our results. To satisfy
  it, when programming, one should prevent runtime errors and notify
  abnormal events to partners using normal constructs to guarantee error
  management or smooth termination (see e.g., Lines 21--26 in
  Figure~\ref{fig:purchase_scenario}). A more elegant way to deal with errors
  would be to include in the language well-known constructs, such as
  try-catch, which are however left as future work. This could be done by
  adapting the ideas presented in~\cite{giachinoescape}.
}\end{remark}

\begin{definition}[\AIOC{} systems]\emph{
A \AIOC{} system is a triple $\tuple{\Sigma, \rules,\mathcal{I}}$ denoting a
\AIOC{} process $\mathcal{I}$ equipped with a global state $\Sigma$ and a set
of updates $\rules$.
}\end{definition}

\begin{definition}[\AIOC{} systems semantics]\emph{
The semantics of \AIOC{} systems is defined as the smallest labelled
transition system (LTS) closed under the rules in Figure~\ref{fig:ioclts},
where symmetric rules for parallel composition have been omitted.
}\end{definition}
\begin{figure}
$$
\begin{array}{c}
\infer[\did{\AIOC}{Interaction}]
{
	\tuple{\ambient, \idx{i}\comm{o}{\role R}{e}{\role S}{x}} 
		\arro{\commLabel{o}{\role R}{v}{\role S}{x}} 
	\tuple{\ambient,\idx{i} \assign{x}{\role S}{v}}}
{\grass{e}_{\Sigma_{\role R}} = v}
\\[8pt]
\infer[\did{\AIOC}{Assign}]
{
	\tuple{\state,\rules,{\idx{i} \assign{x}{\role R}{e}}}
	\arro{\tau}
	\tuple{\state\substLabel{v}{x}{\role R},\rules, \one}
}
{
\grass{e}_{\Sigma_{\role R}} = v
}
\\[5pt]
\infer[\did{\AIOC}{Sequence}]
{
	\tuple{\ambient,{\mathcal I}\seqOp{\mathcal J}} 
		\arro{\mu} 
	\tuple{\ambient',{\mathcal I}'\seqOp{\mathcal J}}
}
{
	\tuple{\ambient,{\mathcal I}} 
		\arro{\mu} 
	\tuple{\ambient',{\mathcal I}'} &
	\mu \neq \tick
}
\\[8pt]
\infer[\did{\AIOC}{Seq-end}]
{
	\tuple{\ambient,{\mathcal I}\seqOp{\mathcal J}} 
		\arro{\mu} 
	\tuple{\ambient,{\mathcal J}'}
}
{
	\tuple{\ambient,{\mathcal I}} 
		\arro{\tick} 
	\tuple{\ambient,{\mathcal I}'}
	&
	\tuple{\ambient,{\mathcal J}} 
		\arro{\mu} 
	\tuple{\ambient,{\mathcal J}'}
}
\\[5pt]
\infer[\did{\AIOC}{Parallel}]
{
	\tuple{\ambient,{\mathcal I}\parallel{\mathcal J}} 
		\arro{\mu}  
	\tuple{\ambient',{\mathcal I}'\parallel{\mathcal J}}
}
{
	\tuple{\ambient,{\mathcal I}} \arro{\mu} \tuple{\ambient',{\mathcal I}'}
	&
	\mu \neq \tick
}
\\[8pt]
\infer[\did{\AIOC}{Par-end}]
{
	\tuple{\ambient,{\mathcal I}\parallel{\mathcal J}} 
		\arro{\tick} 
	\tuple{\ambient,{\mathcal I}'\parallel{\mathcal J}'}
}
{
	\tuple{\ambient,{\mathcal I}} \arro{\tick} 
	\tuple{\ambient,{\mathcal I}'} 
	&
	\tuple{\ambient,{\mathcal J}} \arro{\tick} 
	\tuple{\ambient,{\mathcal J}'}
}
\\[5pt]
\infer[\did{\AIOC}{If-then}]
{
	\tuple{\ambient,\idx{i} \ifthen{b \at \role R}{\mathcal{I}}{\mathcal{I}'}} 
		\arro{\tau} 
	\tuple{\ambient,{\mathcal I}}
}
{
	\grass{b}_{\Sigma_{\role R}} = \trueC
}
\\[8pt]
\infer[\did{\AIOC}{If-else}]
{
	\tuple{\ambient, \idx{i} \ifthen{b \at \role R}{\mathcal{I}}{\mathcal{I}'}}
	\arro{\tau} 
	\tuple{\ambient,{\mathcal I}'}
}
{
	\grass{b}_{\Sigma_{\role R}} = \falseC
}
\\[8pt]
\infer[\did{\AIOC}{While-unfold}]
{
	\tuple{\ambient, \idx{i} \while{b \at \role R}{\mathcal{I}}} 
		\arro{\tau} 
	\tuple{\ambient,{\mathcal I} \seqOp 
		\idx{i}\while{b \at \role R}{\mathcal{I}}}
}
{
	\grass{b}_{\Sigma_{\role R}} = \trueC
}
\\[8pt]
\infer[\did{\AIOC}{While-exit}]
{
	\tuple{\ambient,\idx{i} \while{b \at \role R}{\mathcal{I}}} 
		\arro{\tau} 
	\tuple{\ambient,\one}
}
{
	\grass{b}_{\Sigma_{\role R}} = \falseC
}
\\[8pt]
\infer[\did{\AIOC}{\AdaptRule}]
{
	\tuple{\ambient, \idx{i} \scope{l}{\role R}{\mathcal I}{\Delta}} 
		\arro{{\mathcal I}'} 	
	\tuple{\ambient, \mathcal{I}'}
}
{ 
	\rolesFunc({\mathcal{I}'}) \subseteq \rolesFunc({\mathcal{I})} 
	& 
	\mathcal{I}' \in \rules
	&
	\func{connected}(\mathcal{I}')
	&
	\func{freshIndexes}(\mathcal{I}')
}
\\[8pt]
\infer[\did{\AIOC}{\NoAdaptRule}]
{
	\tuple{\ambient,\idx{i} \scope{l}{\role R}{\mathcal I}{\Delta}}
		\arro{\texttt{\NoAdaptLabel}}
	\tuple{\ambient,{\mathcal I}}
}{}
\\[8pt]
\infer[\did{\AIOC}{End}]
{\tuple{\ambient,\one} \arro{\tick} \tuple{\ambient,\zero}}{}
\quad
\infer[\did{\AIOC}{Change-Updates}]
{
	\tuple{\state,\rules,{\mathcal I}} 
		\arro{\rules'} 
	\tuple{\state,\rules',{\mathcal I}}
}{}
\end{array}
$$
\caption{Annotated \AIOC{} system semantics.}\label{fig:ioclts} 
\end{figure}\enlargethispage{\baselineskip}
The rules in Figure~\ref{fig:ioclts} describe the behaviour of a \AIOC{}
system by induction on the structure of its \AIOC{} process, with a case
analysis on its topmost operator.  We use $\mu$ to range over labels. The possible values for $\mu$
are described below.
\[\begin{array}{rllcll}
\mu \gram & \;
				 \commLabel{o}{\role R}{v}{\role S}{x}  & \textit{(interaction)}   	&
	\; \Div \; & \tau 											      & \textit{(silent)}		
\\
\Div  &  \;
					\mathcal{I} 													& \textit{(update)} &
	\; \Div \; & 
					\texttt{\NoAdaptLabel}								& \textit{(no update)}
\\
\Div	& \;
	\rules 																 				& \textit{(change updates)}     &
	\; \Div \; & \tick 										        & \textit{(termination)} 
\end{array}
\]
Rule $\did{\AIOC}{Interaction}$ executes a communication from $\role R$ to
$\role S$ on operation $o$, where $\role R$ sends to $\role S$ the value $v$
of an expression $e$. The communication reduces to an assignment that inherits
the index $\idxSign i$ of the interaction. The assignment stores value $v$ in
variable $x$ of role $\role S$.
Rule $\did{\AIOC}{Assign}$ evaluates the expression $e$ in the local state
$\Sigma_{\role R}$ and stores the resulting value $v$ in the local variable
$x$ in role $\role R$ ($\substLabel{v}{x}{\role R}$ represents the
substitution).
Rule $\did{\AIOC}{Sequence}$ executes a step in the first process of a
sequential composition, while Rule $\did{\AIOC}{Seq-end}$ acknowledges the
termination of the first process, starting the second one.
Rule $\did{\AIOC}{Parallel}$ allows a process in a parallel composition to
compute, while Rule $\did{\AIOC}{Par-end}$ synchronises the termination of
two parallel processes.
Rules $\did{\AIOC}{If-then}$ and $\did{\AIOC}{If-else}$ evaluate the boolean
guard of a conditional, selecting the ``then'' and the ``else'' branch,
respectively.
Rules $\did{\AIOC}{While-unfold}$ and $\did{\AIOC}{While-exit}$ correspond
respectively to the unfolding of a while loop when its condition is satisfied
and to its termination otherwise.
Rule $\did{\AIOC}{\AdaptRule}$ and Rule $\did{\AIOC}{\NoAdaptRule}$ deal with
updates: the former applies an update, while the latter allows the body of the
scope to be executed without updating it.
More precisely, Rule $\did{\AIOC}{\AdaptRule}$ models the application of the
update $\arule{l}{\cnd}{\mathcal{I}'}$ to the scope $\scope{l}{\role
R}{\mathcal I}{\Delta}$ which, as a result, is replaced by the \AIOC{} process
${\mathcal{I}'}$.
In the conditions of the rule, we use the function $\rolesFunc$ and the
predicates $\func{connected}$ and $\func{freshIndexes}$. Function
$\rolesFunc({\mathcal I})$, defined in Figure~\ref{fig:roles}, computes the
roles of a \AIOC{} process ${\mathcal I}$. The condition of the rule requires
that the roles of the update are a subset of the roles of the body of the
scope.
\begin{figure}[t]
$$
\begin{array}l
\rolesFunc(\idx{i}\comm{o}{\role R}{e}{\role S}{x}) =  \{\role R,\role S\} 
\\
\rolesFunc(\one) = \rolesFunc(\zero) = \emptyset
\\
\rolesFunc(\idx{i} x \at \role R = e) = \{\role R\}
\\
\rolesFunc({\mathcal I}\seqOp{\mathcal I}') = 
  \rolesFunc({\mathcal I} | {\mathcal I}') = \rolesFunc({\mathcal I}) \cup 
    \rolesFunc({\mathcal I}')
\\
\rolesFunc(\idx{i} \ifthen{b \at \role R}{\mathcal{I}}{\mathcal{I}'}) = 
  \{\role R\} \cup \rolesFunc({\mathcal{I}}) \cup \rolesFunc(\mathcal{I}')
\\
\rolesFunc(\idx{i} \while{b \at \role R}{\mathcal{I}}) = 
\rolesFunc(\idx{i} \scope{l}{\role R}{\mathcal I}{\Delta}) = 
  \{\role R\} \cup \rolesFunc(\mathcal{I})
\end{array}
$$
\caption{Auxiliary function $\rolesFunc$.\label{fig:roles}}
\end{figure}
Predicate $\func{connected}(\mathcal{I}')$ holds if $\mathcal{I}'$ is
connected. Connectedness is a well-formedness property of \AIOC{}s and is
detailed in \S~\ref{sub:connectedness}. Roughly, it ensures that roles
involved in a sequential composition have enough information to enforce the
correct sequence of actions.
Predicate $\func{freshIndexes}(\mathcal{I}')$ holds if all indexes in
$\mathcal{I}'$ are fresh with respect to all indexes already present in the
target \AIOC{}\footnote{We do not give a formal definition of
$\func{freshIndexes}(\mathcal{I}')$ to keep the presentation simple. However,
freshness of indexes can be formally ensured using restriction as in
$\pi$-calculus~\cite{sangiorgi2003pi}.}. This is needed to avoid interferences
between communications inside the update and communications in the original
\AIOC{}. This problem is discussed in more details in
\S~\ref{sec:endpoint_projection}, Example~\ref{example:indexes}.
Rule $\did{\AIOC}{\NoAdaptRule}$, used when no update is applied, removes the
scope boundaries and starts the execution of the body of the scope.
Rule $\did{\AIOC}{End}$ terminates the execution of an empty process.  
Rule $\did{\AIOC}{Change-Updates}$ allows the set $\rules$ of available updates
to change. This rule is always enabled and models the fact that the set of
available updates is not controlled by the system, but by the external world:
the set of updates can change at any time, the system cannot forbid or delay
these changes, and the system is not notified when they happen.
Label $\rules$ is needed to make the changes to the set of
available updates observable (cf.~Definition~\ref{def:bisim}).

\begin{remark}\emph{
Whether to update a scope or not, and which update to apply if many are
available, is completely non-deterministic. We have adopted this view to
maximise generality. However, for practical applications it is also possible to
reduce the non-determinism using suitable policies to decide when and whether a
given update applies. One of such policies is defined in AIOCJ (see
\S~\ref{sec:aiocj}).
}\end{remark}

We can finally provide the definition of \emph{\AIOC{} traces} and 
\emph{weak \AIOC{} traces}, which we use to express our results of behavioural
correspondence between \AIOC{}s and \APOC{}s. Intuitively, in \AIOC{} traces
all the performed actions are observed, whilst in weak \AIOC{} traces silent
actions $\tau$ are not visible.
\begin{definition}[\AIOC{} traces]\emph{
A \emph{(strong) trace} of a \AIOC{} system $\tuple{ \Sigma_1,
\rules_1,{\mathcal I}_1}$ is a sequence (finite or infinite) of labels $\mu_1,
\mu_2, \dots$ such that there is a sequence of \AIOC{} system transitions
$\tuple{\Sigma_1,\rules_1,{\mathcal I}_1} \arro{\mu_1}
\tuple{\Sigma_2,\rules_2,{\mathcal I}_2} \arro{\mu_2} \dots$.}

\emph{A \emph{weak trace} of a \AIOC{} system $\tuple{ \Sigma_1,\rules_1,{\mathcal
I}_1}$ is a sequence of labels $\mu_1, \mu_2, \dots$ obtained by removing all
silent labels $\tau$ from a trace of $\tuple{
\Sigma_1,\rules_1,{\mathcal I}_1}$.
}\end{definition}

\section{Dynamic Process-Oriented Choreographies}\label{sec:apoc}

In this section we define the syntax and semantics of \APOC{}s, the target
language of our projection from \AIOC{}s.
We remind that \AIOC{}s are not directly executable since their basic
primitives describe distributed interactions. On the contrary, mainstream
languages like Java and C, used for implementation, describe distributed
computations using local behaviours and communication/synchronisation
primitives, such as message send and message receive. In order to describe
implementations corresponding to \AIOC{}s we introduce the \APOC{} language, a
core language based on this kind of primitives, but tailored to program
updatable systems. Indeed, differently from \AIOC{} constructs, \APOC{}
constructs are locally implementable in any mainstream language. In AIOCJ (see
\S~\ref{sec:aiocj}) we implement the \APOC{} constructs in the
Jolie~\cite{MGZ14} language.

\subsection{DPOC syntax}\label{sub:dpoc_syntax}

\APOC{}s include \emph{processes}, ranged over by $P,P',\ldots$,
describing the behaviour of participants.  $(P,\Gamma)_{\role R}$ denotes a
\emph{\APOC{} role} named $\role R$, executing process $P$ in a local state
$\Gamma$. \emph{Networks}, ranged over by $\net$, $\net'$, $\ldots$, are
parallel compositions of \APOC{} roles with different names. \APOC{} systems,
ranged over by ${\mathcal S}$, are \APOC{} networks equipped with a set of
updates $\rules$, namely pairs $\tuple{\rules, \net}$.

\APOC{}s, like \AIOC{}s, feature operations $o$. Here we call them
\emph{programmer-specified operations} to mark the difference with respect to
\emph{auxiliary operations}, ranged over by $o^*$. We use $o^?$ to range over
both programmer-specified and auxiliary operations.
Differently from communications on programmer-specified operations,
communications on auxiliary operations have no direct correspondent at the
\AIOC{} level. Indeed, we introduce auxiliary operations in \APOC{}s to
implement the synchronisation mechanisms needed to realise the global
constructs of \AIOC{}s (conditionals, while loops, and scopes) at \APOC {}
level.

Like \AIOC{} constructs, also \APOC{} constructs are annotated using indexes.
However, in \APOC{}s we use two kinds of indexes: \emph {normal indexes}
$\idxSign i \in \mathbb{N}$ and \emph{auxiliary indexes} of the forms
$\idxTrueSign{i}$, $\idxFalseSign{i}$, $\idxIfRecvSign{i}$, and
$\idxCloseSign{i}$ where $\idxSign{i} \in \mathbb{N}$.
Auxiliary indexes are introduced by the projection, described in
\S~\ref{sec:endpoint_projection}, and are described in detail there.
We let $\iota$ range over \APOC{} indexes.

In \APOC{}s, normal indexes are also used to prefix the operations in send and
receive primitives\footnote{In principle, one may use just indexes and drop
operations. However, explicit operations are the standard for
communication-based systems, in particular in the area of Web Services, as
they come handy to specify and to debug such systems.}. Thus, a send and a
receive can interact only if they are on the same operation and they are
prefixed by the same normal index. This is needed to avoid interferences
between concurrent communications, in particular when one of them comes from
an update. We describe in greater detail this issue in
\S~\ref{sec:endpoint_projection}, Example~\ref{example:indexes}.

The syntax of \APOC{}s is the following.

$$\begin{array}{rll}
P \gram 
		 & \idx{\iota} \cinp{\idxSign{i}.o^?}{x}{\role R}	& \textit{(receive)} 			\\
\Div & \idx{\iota} \cout{\idxSign{i}.o^?}{e}{\role R}	& \textit{(send)} 				\\
\Div & \idx{i} \cout{\idxSign{i}.o^\ast}{X}{\role R} 	& \textit{(send-update)} 	\\
\Div &	P\seqOp P' 																		& \textit{(sequence)}			\\
\Div &	P \mid P'																			& \textit{(parallel)}			\\
\Div &	\idx{\iota} x = e 														& \textit{(assignment)}		\\
\end{array}
\begin{array}{rll}
\Div &	\one 																					& \textit{(inaction)}			\\
\Div &	\zero 																				& \textit{(end)}					\\
\Div &	\ifthenKey{b}{P}{P'}{i} 											& \textit{(conditional)}		\\
\Div &	\whileKey{b}{P}{i}														& \textit{(while)}				\\
\Div &	\pscope{i}{l}{\role R}{P}{\Delta}{S}  				& \textit{(scope-coord)} 	\\
\Div &	\psscope{i}{l}{\role R}{P}										& \textit{(scope)}				
\end{array}
$$
\vspace{5pt}
$$
X \gram \qquad {\noC} 											\quad \mid \quad  
				P 
\hspace{50pt}
\net \gram \qquad \roleExec{P,\Gamma}{\role R} \quad \mid \quad 
					 \net \parallel \net'
$$
\vspace{5pt}

\APOC{} processes include receive action
$\idx{\iota}\cinp{\idxSign{i}.o^?}{x}{\role R}$ on a specific operation
$\idxSign{i}.o^?$ (either programmer-specified or auxiliary) of a message from
role $\role R$ to be stored in variable $x$, send action $\idx{\iota}
\cout{\idxSign{i}.o^?}{e}{\role R}$ of the value of an expression $e$ to be sent
to role $\role R$, and higher-order send action $\idx{i}
\cout{\idxSign{i}.o^*}{X}{\role R}$ of the higher-order argument $X$ to be sent
to role $\role R$. Here $X$ may be either a \APOC{} process $P$, which is the
new code for a scope in $\role R$, or a token $\tt no$, notifying that no
update is needed. $P \seqOp P'$ and $P \parOpP P'$ denote the sequential and
parallel composition of $P$ and $P'$, respectively. Processes also feature
assignment $\idx{\iota} x = e$ of the value of expression $e$ to variable $x$,
the process $\one$, that can only successfully terminate, and the terminated
process $\zero$. \APOC{} processes also include conditionals
$\ifthenKey{b}{P}{P'}{i}$ and loops $\whileKey{b}{P}{i}$. Finally, there are
two constructs for scopes. Construct $\pscope{i}{l}{\role R}{P}{\Delta}{S}$
defines a scope with body $P$ and set of participants $S$, and may occur only
inside role $\role R$, which acts as coordinator of the update.  The shorter
version $\psscope{i}{l}{\role R}{P}$ is used instead inside the code of some
role $\role R_1$, which is not the coordinator $\role R$ of the update. In
fact, only the coordinator $\role R$ needs to know the set $S$ of involved
roles to be able to send to them their updates.

\subsection{\APOC{} semantics}\label{sec:apoc_semantics}

\APOC{} semantics is defined in two steps: we define the semantics of
\APOC{} roles first, and then we define how roles interact giving rise to the
semantics of \APOC{} systems.

\begin{definition}[\APOC{} roles semantics]\emph{ The semantics of \APOC{} roles
is defined as the smallest LTS closed under the rules in
Figure~\ref{fig:apoc-proc_1}, where we report the rules dealing with
computation, and Figure~\ref{fig:apoc-proc_2}, in which we define the
rules related to updates. Symmetric rules for parallel composition
have been omitted.
}\end{definition}

\begin{figure}
$$
\begin{array}{c}
\infer[\did{\APOC}{One}]
{\roleExec{\one,\Gamma}{R} \arro{\tick}
\roleExec{\zero,\Gamma}{R}
}{}
\qquad
\infer[\did{\APOC}{Assign}]
{\roleExec{\idx{\iota} x = e,\Gamma}{R} \arro{\tau} 
	\roleExec{\one,\Gamma\substAPOC{v}{x}}{R}
}
{\grass{e}_\Gamma = v}
\\[5pt]
\infer[\did{\APOC}{Send}]{
	\roleExec{\idx{\iota} \cout{\idxSign{i}.o^?}{e}{\role S},\Gamma}{R} 
		\arro{\coutLabel{\idxSign{i}.o^?}{v}{\role S}:\role R} 
	\roleExec{\one,\Gamma}{R}}
	{\grass{e}_\Gamma = v}
\\[5pt]
\infer[\did{\APOC}{Recv}]{
	\roleExec{\idx{\iota} \cinp{\idxSign{i}.o^?}{x}{\role S},\Gamma}{R} 
		\arro{\cinpLabel{\idxSign{i}.o^?}{x}{v}{\role S}:\role R}
	\roleExec{\idx{\iota} x = v,\Gamma}{R}}{}\\[5pt]
\infer[\did{\APOC}{Send-\AdaptRule}]
{\roleExec{\idx{i} \cout{\idxSign{i}.o^?}{X}{\role S},\Gamma}{R}
	\arro{\coutLabel{\idxSign{i}.o^?}{X}{\role S}:\role R}
	\roleExec{\one,\Gamma}{R}}{}
\\[5pt]
\infer[\did{\APOC}{Sequence}]
	{\roleExec{P \seqOp Q,\Gamma}{R} \arro{\delta} \roleExec{P' \seqOp
	Q,\Gamma'}{R}} {\roleExec{P,\Gamma}{R} \arro{\delta} \roleExec{P',\Gamma'}{R}
	&
		\delta \neq \tick}
\\[5pt]
\infer[\did{\APOC}{Seq-end}]
{\roleExec{P \seqOp Q,\Gamma}{R} \arro{\delta} \roleExec{Q',\Gamma'}{R}}
{\roleExec{P,\Gamma}{R} \arro{\tick} \roleExec{P',\Gamma}{R} \quad
\roleExec{Q,\Gamma}{R}\arro{\delta} \roleExec{Q',\Gamma'}{R}}
\\[5pt]
\infer[\did{\APOC}{Parallel}]
{\roleExec{P \mid Q,\Gamma}{R} \arro{\delta} \roleExec{P' \mid Q,\Gamma'}{R}}
{\roleExec{P,\Gamma}{R} \arro{\delta} \roleExec{P',\Gamma'}{R} & 
	\delta \neq \tick
}
\\[5pt]
\infer[\did{\APOC}{Par-end}]
{\roleExec{P \mid Q,\Gamma}{R} \arro{\tick}\roleExec{P' \mid Q',\Gamma}{R}}
{\roleExec{P,\Gamma}{R} \arro{\tick} \roleExec{P',\Gamma}{R} &
\roleExec{Q,\Gamma}{R}
\arro{\tick} \roleExec{Q',\Gamma}{R}}
\\[5pt]
\infer[\did{\APOC}{If-Then}]
{\roleExec{\ifthenKey{b}{P}{P'}{i},\Gamma}{R}\arro{\tau} 
	\roleExec{P,\Gamma}{R}
}
{\grass{b}_\Gamma = \trueC}
\\[5pt]
\infer[\did{\APOC}{If-Else}]
{\roleExec{\ifthenKey{b}{P}{P'}{i},\Gamma}{R}\arro{\tau}
\roleExec{P',\Gamma}{R}} {\grass{b}_\Gamma = \falseC}
\\[5pt]
\infer[\did{\APOC}{While-unfold}]
{\roleExec{\whileKey{b}{P}{i},\Gamma}{R} \arro{\tau}
\roleExec{P \seqOp \whileKey{e}{P}{i},\Gamma}{R}
} 
{\grass{b}_\Gamma = \trueC}
\\[5pt]
\infer[\did{\APOC}{While-exit}]
{\roleExec{\whileKey{b}{P}{i},\Gamma}{R} \arro{\tau}
	\roleExec{\one,\Gamma}{R}
}
{\grass{b}_\Gamma = \falseC}
\end{array}
$$
\caption{\APOC{} role semantics. Computation rules. (Update rules in
Figure~\ref{fig:apoc-proc_2})\label{fig:apoc-proc_1}}
\end{figure}

\begin{figure}
{\small
$$
\begin{array}c
\infer[\did{\APOC}{Lead-\AdaptRule}]
{
\begin{array}{l}
\roleExec{\pscope{i}{l}{\role R}{P}{\Delta}{S},\Gamma}{R} \arro{{\mathcal I}} 
\\ \quad \quad
\roleExec{\prod\limits_{\role R_j \in S \setminus \{\role R\}}
\idx{i} \cout{\idxSign{i}.\AuxSb^*_{\idxSign{i}}}{\pi(\mathcal{I},\role R_j)}{\role
R_j}\seqOp \pi(\mathcal{I},\role R)\seqOp
\prod\limits_{\role R_j \in S \setminus \{\role R\}} \idx{i}
\cinp{\idxSign{i}.\AuxSe^*_{\idxSign{i}}}{\_}{\role R_j},\Gamma}{R}\\
\end{array}
}
{
\rolesFunc(\mathcal{I}) \subseteq S
&
\func{freshIndexes}(\mathcal{I})
&
\func{connected}(\mathcal{I})
}
\\[35pt]
\infer[\did{\APOC}{Lead-\NoAdaptRule}]
{
\begin{array}{l}
\roleExec{\pscope{i}{l}{\role R}{P}{\Delta}{S},\Gamma}{R}
\arro{\texttt{\NoAdaptLabel}} 
\\ \quad
\roleExec{
\prod\limits_{\role R_j \in S \setminus \{\role R\}} 
\idx{i} \cout{\idxSign{i}.\AuxSb^*_{\idxSign{i}}}{\noC}{\role R_j}\seqOp
P \seqOp \prod\limits_{\role R_j \in S \setminus \{\role R\}}
\idx{i} \cinp{\idxSign{i}.\AuxSe^*_{\idxSign{i}}}{\_}{\role R_j},\Gamma}{R}
\end{array}
}{}
\\[35pt]
\infer[\did{\APOC}{\AdaptRule}]
{
	\roleExec{\psscope{i}{l}{\role S}{P},\Gamma}{R}
		\arro{\cinpLabel{\idxSign{i}.\AuxSb^*_{\idxSign{i}}}{}{P'}{\role
		S}:\role{R}}
	\roleExec{P' \seqOp \idx{i}
\cout{\idxSign{i}.\AuxSe^*_{\idxSign{i}}}{\okC}{\role S},\Gamma}{R} }{}
\\[15pt]
\infer[\did{\APOC}{\NoAdaptRule}] {\roleExec{\psscope{i}{l}{\role
S}{P},\Gamma}{R}
\arro{\cinpLabel{\idxSign{i}.\AuxSb^*_{\idxSign{i}}}{}{\noC}{\role
S}:\role{R}} \roleExec{P \seqOp \idx{i}
\cout{\idxSign{i}.\AuxSe^*_{\idxSign{i}}}{\okC}{\role S},\Gamma}{R} }{}
\end{array}
$$}
\caption{\APOC{} role semantics. Update rules. (Computation rules in
Figure~\ref{fig:apoc-proc_1})}\label{fig:apoc-proc_2}
\end{figure}

\paragraph{\emph{\APOC{} role semantics.}}
We use $\delta$ to range over labels. The possible values for $\delta$ are as follows: 

\begin{align*}
\delta \gram \; &
	\coutLabel{\idxSign{i}.o^?}{v}{\role S}:\role R 		& \textit{(send)} &
\; \Div \;
	\cinpLabel{\idxSign{i}.o^?}{x}{v}{\role S}:\role R 	& \textit{(receive)} &
\\
\Div \;	& 
	\coutLabel{\idxSign{i}.o^*}{X}{\role S}:\role R			& \textit{(send-update)} &
	\; \Div \;
	\cinpLabel{\idxSign{i}.o^*}{}{X}{\role S}:\role R 	& \textit{(receive-update)}	&
\\
\Div \; &
	\mathcal{I}																					& \textit{(update)} &
	\; \Div \;
	\texttt{\NoAdaptLabel}															& \textit{(no-update)} 	&
\\
\Div \; &
\tau 																									& \textit{(silent)} &
	\; \Div \;
\tick 																								& \textit{(termination)} &	
\end{align*}

The semantics is in the early style.
Rule $\did{\APOC}{Recv}$ receives a value $v$ from role $\role S$ and assigns it
to local variable $x$ of $\role R$. Similarly to Rule $\did{\AIOC}{Interact}$
(see Figure~\ref{fig:ioclts}), the reception reduces to an assignment that
inherits the index $\idxSign i$ from the receive primitive.

Rules $\did{\APOC}{Send}$ and $\did{\APOC}{Send-\AdaptRule}$ execute send and
higher-order send actions, respectively. Send actions evaluate expression
$e$ in the local state $\Gamma$.
Rule $\did{\APOC}{One}$ terminates an empty process.
Rule $\did{\APOC}{Assign}$ executes an assignment ($\substAPOC{v}{x}$
represents the substitution of value $v$ for variable $x$).
Rule $\did{\APOC}{Sequence}$ executes a step in the first process of a
sequential composition, while Rule $\did{\APOC}{Seq-end}$ acknowledges the
termination of the first process, starting the second one.
Rule $\did{\APOC}{Parallel}$ allows a process in a parallel composition to
compute, while Rule $\did{\APOC}{Par-end}$ synchronises the termination of two
parallel processes.
Rules $\did{\APOC}{If-then}$ and $\did{\APOC}{If-else}$ select the ``then'' and
the ``else'' branch in a conditional, respectively.
Rules $\did{\APOC}{While-unfold}$ and $\did{\APOC}{While-exit}$ model
respectively the unfolding and the termination of a while loop.

The rules reported in Figure~\ref{fig:apoc-proc_2} deal with code updates.
Rules $\did{\APOC}{Lead-\AdaptRule}$ and $\did{\APOC}{Lead-\NoAdaptRule}$
specify the behaviour of the coordinator $\role R$ of the update, respectively
when an update is performed and when no update is performed.
In particular, $\role R$ non-deterministically selects whether to update or
not and, in the first case, which update to apply.
The coordinator communicates the selection to the other roles in the scope
using operations $\AuxSb^*_\idxSign{i}$. The content of the message is either the new
code that the other roles need to execute, if the update is performed, or a
token $\noC$, if no update is applied.
Communications on operations $\AuxSb^*_\idxSign{i}$ also ensure that no role starts
executing the scope before the coordinator has selected whether to update or
not. The communication is received by the other roles using Rule
$\did{\APOC}{\AdaptRule}$ if an update is selected and Rule
$\did{\APOC}{\NoAdaptRule}$ if no update is selected.
Similarly, communications on $\AuxSe^*_\idxSign{i}$ ensure that the updated
code has been completely executed before any role can proceed to execute the
code that follows the scope in its original $\APOC{}$ process. Communications
on $\AuxSe^*_\idxSign{i}$ carry no relevant data: they are used for synchronisation
purposes only.

As already discussed, Rule $\did{\APOC}{Lead-\AdaptRule}$ models the fact that
the coordinator $\role R$ of the update non-deterministically selects an
update $\mathcal{I}$. The premises of Rule $\did{\APOC}{Lead-\AdaptRule}$ are
similar to those of Rule $\did{\AIOC}{\AdaptRule}$ (see
Figure~\ref{fig:ioclts}). Function $\rolesFunc$ is used to check that the
roles in $\mathcal{I}$ are included in the roles of the scope. Freshness of
indexes is checked by predicate $\func{freshIndexes}$, and well formedness of
$\mathcal{I}$ by predicate $\func{connected}$ (formally defined later on, in
Definition~\ref{def:connectedness} in
\S~\ref{sub:connectedness}).
In particular, the coordinator $\role R$ generates the processes to be
executed by the roles in $S$ using the process-projection function $\pi$
(detailed in \S~\ref{sec:endpoint_projection}). More precisely,
$\pi(\mathcal{I},\role R_i)$ generates the code for role $\role R_i$. The
processes $\pi(\mathcal{I},\role R_i)$ are sent via auxiliary higher-order
communications on $\AuxSb_\idxSign{i}^*$ to the roles that have to execute
them. These communications also notify the other roles that they can start
executing the new code.  Here, and in the remainder of the paper, we define
$\prod_{\role R_i \in S} P_i$ as the parallel composition of \APOC{} processes
$P_i$ for $R_i \in S$. We assume that $\prod$ binds more tightly than
sequential composition.
After the communication of the updated code to the other participants, $\role R$
starts its own updated code $\pi(\mathcal{I},\role R)$. Finally, auxiliary
communications $\AuxSe^*_{\idxSign{i}}$ are used to synchronise the end of the
execution of the update (here $\_$ denotes a fresh variable to store the
synchronisation message $\texttt{ok}$). 
Rule $\did{\APOC}{Lead-\NoAdaptRule}$ defines the behaviour of the coordinator
$\role R$ when no update is applied. In this case, $\role R$ sends a token
$\noC$ to the other involved roles, notifying them that no update is applied
and that they can start executing their original code. End of scope
synchronisation is the same as that of Rule $\did{\APOC}{Lead-\AdaptRule}$.
Rules $\did{\APOC}{\AdaptRule}$ and $\did{\APOC}{\NoAdaptRule}$ define the
behaviour of the other roles involved in the scope. The scope waits for a
message from the coordinator. If the content of the message is $\noC$,
the body of the scope is executed. Otherwise, the content of the message is a
process $P'$ which is executed instead of the body of the scope.

We highlight the importance of the coordinator $\role R$. Since the set of
updates may change at any moment, we need to be careful to avoid that the
participants in the scope get code projected from different updates.
Given that only role $\role R$ obtains and delivers the new code, one is
guaranteed that all the participants receive their projection of the same
update.

\paragraph{\emph{\APOC{} system semantics.}}
\begin{definition}[\APOC{} systems semantics]\emph{ The semantics of \APOC{} systems
is defined as the smallest LTS closed under the rules in
Figure~\ref{fig:apoc-sys}. Symmetric rules for parallel composition
have been omitted.
}\end{definition}

\begin{figure}
$$
\begin{array}{c}

\infer[\did{\APOC}{Lift}]
{\tuple{\rules, \net} \arro{\delta} \tuple{\rules, \net'}}
{\net \arro{\delta} \net' & \delta \neq {\mathcal I}}

\qquad

\infer[\did{\APOC}{Lift-\AdaptRule}]
{\tuple{\rules, \net} \arro{\mathcal{I}} \tuple{\rules,\net'}}
{\net \arro{\mathcal{I}} \net' 
		& \mathcal{I} \in \rules}

\\[10pt]

\infer[\did{\APOC}{Synch}]
{
	\tuple{\rules, \net \parallel \net''}
		\arro{\commLabel{o^?}{\role R}{v}{\role S}{x} } 
	\tuple{\rules, \net' \parallel \net'''}
}
{	
	\tuple{\rules, \net} 
		\arro{\coutLabel{\idxSign{i}.o^?}{v}{\role S}:\role R} 
	\tuple{\rules, \net'} &
	\tuple{\rules,\net''} 
		\arro{\cinpLabel{\idxSign{i}.o^?}{x}{v}{\role R}:\role S} 
	\tuple{\rules,\net'''}
}
\\[10pt]

\infer[\did{\APOC}{Synch-\AdaptRule}]
{
	\tuple{\rules,\net \parallel \net''}
		\arro{\commLabel{o^*}{\role R}{X}{\role S}{}}
	\tuple{\rules,\net' \parallel \net'''}
}
{
	\tuple{\rules,\net}
		\arro{\coutLabel{\idxSign{i}.o^*}{X}{\role S}:\role R}
	\tuple{\rules,\net'} 
	& 
	\tuple{\rules,\net''}
		\arro{\cinpLabel{\idxSign{i}.o^*}{}{X}{\role R}:\role S}
	\tuple{\rules,\net'''}
}
\\[10pt]

\infer[\did{\APOC}{Ext-Parallel}]
{
	\tuple{\rules,\net\parallel \net''}
		\arro{\eta} \tuple{\rules,\net' \parallel \net''}
}
{
	\tuple{\rules,\net} 
		\arro{\eta} 
	\tuple{\rules,\net'} 
	& 
	\eta \neq \tick
}

\\[10pt]

\infer[\did{\APOC}{Ext-Par-End}]
{
	\tuple{\rules,\net\parallel \net''} 
		\arro{\tick} 
	\tuple{\rules,\net' \parallel \net'''}
}
{
	\tuple{\rules,\net} \arro{\tick} \tuple{\rules,\net'} 
	& 
	\tuple{\rules,\net''} \arro{\tick} \tuple{\rules,\net'''}
}

\\[10pt]

\infer[\did{\APOC}{Change-Updates}]
{\tuple{\rules,\net} \arro{\rules'} \tuple{\rules',\net}}{}
\end{array}
$$
\caption{\APOC{} system semantics.}
\label{fig:apoc-sys}
\label{fig:apoc-roles}
\label{fig:apoc-net}
\end{figure}

We use $\eta$ to range over \APOC{} systems labels. The possible values of
$\eta$ are as follows:

$$\begin{array}{rl@{\qquad}l}
\eta \gram & 
	\commLabel{o^?}{\role R}{v}{\role S}{x} 		& \textit{(interaction)} 
\\[2pt]
\Div & \commLabel{o^*}{\role R}{X}{\role S}{} & \textit{(interaction-update)}
\\[2pt]
\Div & \delta 																& \textit{(role label)}
\end{array}$$

Rules $\did{\APOC}{Lift}$ and $\did{\APOC}{Lift-\AdaptRule}$ lift role
transitions to the system level. Rule $\did{\APOC}{Lift-\AdaptRule}$ also
checks that the update ${\mathcal I}$ belongs to the set of currently
available updates $\rules$.
Rule $\did{\APOC}{Synch}$ synchronises a send with the corresponding
receive, producing an interaction. Rule
$\did{\APOC}{Synch-\AdaptRule}$ is similar, but it deals with
higher-order interactions.  Note that Rules $\did{\APOC}{Synch}$ and
$\did{\APOC}{Synch-\AdaptRule}$ remove the prefixes from DPOC
operations in transition labels.  The labels of these transitions
store the information on the occurred communication: label
$\commLabel{o^?}{\role R_1}{v}{\role R_2}{x}$ denotes an interaction
on operation $o^?$ from role $\role R_1$ to role $\role R_2$ where the
value $v$ is sent by $\role R_1$ and then stored by $\role R_2$ in
variable $x$. Label $\commLabel{o^*}{\role R_1}{X}{\role R_2}{}$
denotes a similar interaction, but concerning a higher-order value
$X$, which can be either the code used in the update or a token
$\noC$ if no update is performed. No receiver variable is
specified, since the received value becomes part of the code of the
receiving process.
Rule $\did{\APOC}{Ext-Par}$ allows a network inside a parallel
composition to compute. Rule $\did{\APOC}{Ext-Par-End}$ synchronises the
termination of parallel networks. Finally, Rule $\did{\APOC}{Change-Updates}$
allows the set of updates to change arbitrarily.

We now define \emph{\APOC{} traces} and \emph{weak \APOC{} traces}, which we
later use, along with \AIOC{} traces and weak \AIOC{} traces, to define our
result of correctness.
\begin{definition}[\APOC{} traces]\emph{A \emph{(strong) trace} of a \APOC{} system $\tuple{\rules_1, \net_1}$ is a sequence
(finite or infinite) of labels $\eta_1, \eta_2, \dots$ with $$\eta_i \in
\{\tau, \commLabel{o^?}{\role R_1}{v}{\role R_2}{x}, \commLabel{o^*}{\role
R_1}{X}{\role R_2}{}, \tick, {\mathcal I},\texttt{\NoAdaptLabel},
\rules \}$$ such that there is a sequence of transitions
$\tuple{\rules_1,\net_1} \arro{\eta_1} \
\tuple{\rules_{2},\net_{2}} \arro{\eta_2} \dots$.\\
A \emph{weak trace} of a \APOC{} system $\tuple{\rules_1,\net_1}$ is a sequence of
labels $\eta_1,\eta_2, \ldots$ obtained by removing all the labels
corresponding to auxiliary communications, i.e., of the form
$\commLabel{o^*}{\role R_1}{v}{\role R_2}{x}$ or $\commLabel{o^*}{\role
R_1}{X}{\role R_2}{}$, and the silent labels $\tau$, from a trace of
$\tuple{\rules_1,\net_1}$. 
}\end{definition}

\APOC{} traces do not allow send and receive actions. Indeed these
actions represent incomplete interactions, thus they are needed for
compositionality reasons, but they do not represent relevant
behaviours of complete systems. Note also that these actions have no
correspondence at the \AIOC{} level, where only whole interactions are
allowed.

\begin{remark}\emph{
Contrarily to \AIOC{}s, \APOC{}s in general can deadlock. For instance,
$$\roleExec{\idx{i} \cinp{\idxSign{i}.o}{x}{\role R'},\Gamma}{R}$$ is a
deadlocked \APOC{} network: the process $\idx{i} \cinp{\idxSign{i}.o}{x}{\role
R'}$ is not terminated, and the only enabled actions are changes of the set of
updates (i.e., transitions with label $\rules$), which are not actual system
activities, but are taken by the environment. Notably, the \APOC{} above
cannot be obtained by projecting a
\AIOC{}.
In fact, \APOC{}s generated from \AIOC{}s are guaranteed to be deadlock free.
}\end{remark}

\section{Projection Function}\label{sec:endpoint_projection}

We now introduce the projection function $\proj$. Given a \AIOC{}
specification, $\proj$ returns a network of \APOC{} programs that
enact the behaviour defined by the originating \AIOC{}.

We write the projection of a \AIOC{} $\mathcal{I}$ as
$\proj(\mathcal{I},\Sigma)$, where $\Sigma$ is a global state. Informally, the
projection of a \AIOC{} is a parallel composition of terms, one for each role of
the \AIOC{}. The body of these roles is computed by the
\emph{process-projection} function $\pi$ (defined below). Given a \AIOC{} and a
role name $\role R$, the process-projection returns the process corresponding to
the local behaviour of role $\role R$. Since the roles executing the
process-projections are composed in parallel, the projection of a \AIOC{}
program results into the \APOC{} network of the projected roles.

To give the formal definition of projection, we first define $\parallel_{i \in
I} \net_i$ as the parallel composition of networks $\net_i$ for $i \in I$.
\begin{definition}[Projection]\emph{ The projection of a \AIOC{} process ${\mathcal
I}$ with global state $\Sigma$ is the \APOC{} network defined by:
$$\proj({\mathcal I},\Sigma)=\parallel_{\role S \in \rolesFunc({\mathcal I})}
\roleExec{\pi({\mathcal I},\role S), \Sigma_{\role S}}{\role S}$$
}\end{definition}

The process-projection function that derives \APOC{} processes from
\AIOC{} processes is defined as follows.

\begin{definition}[Process-projection]\emph{
Given an annotated \AIOC{} process $\mathcal{I}$ and a role $\role R$ the
projected \APOC{} process $\pi(\mathcal{I},\role R)$ is defined as in
Figure~\ref{fig:pi}.
}\end{definition}

\begin{figure}
{\footnotesize
$$
\hspace{-25pt}\begin{array}{c}
\rulebox{\pi(\one,\role S)} = \one
\qquad
\rulebox{\pi(\zero,\role S)} = \zero 
\\[5pt]
\rulebox{\pi({\mathcal I} \seqOp {\mathcal I}',\role S)} =
	\pi({\mathcal I},\role S) \seqOp \pi({\mathcal I}',\role S)
\qquad
\rulebox{\pi({\mathcal I} \parOpI {\mathcal I}',\role S)} =
	\pi({\mathcal I},\role S) \parOpI{\pi({\mathcal I}',\role S)}
\end{array}
$$
$$
\setlength{\arraycolsep}{2pt}\hspace{-20pt}\begin{array}{lcl}
\rulebox{\pi(\idx{i} \assign{x}{\role R}{e},\role R)} & = & \idx{i} x = e
\\[5pt]
\rulebox{\pi(\idx{i} \assign{x}{\role R}{e},\role S)
	\mbox{ and } \role S \neq \role R} & = & \one 
\\[5pt]
\rulebox{\pi(\idx{i} \comm{o}{\role R_1}{e}{\role R_2}{x} ,\role R_1)}
& = &
	\idx{i} \idxSign{i}.\cout{o}{e}{\role R_2}
\\[5pt]
\rulebox{\pi(\idx{i} \comm{o}{\role R_1}{e}{\role R_2}{x} ,\role R_2)} & = &
	\idx{i} \idxSign{i}.\cinp{o}{x}{\role R_1}
\\[5pt]
\rulebox{
	\begin{array}l
	\pi(\idx{i} \comm{o}{\role R_1}{e}{\role R_2}{x} ,\role S)
	\\ \quad \mbox{and } \role S \not \in \{\role R_1,\role R_2\}
\end{array}} & = & \one
\\
\rulebox{\pi(\ifthenKey{b \at \role R}{\mathcal{I}}{\mathcal{I}'}{i},\role R)} & = & \left\{
	\begin{array}l
		\idx{i} \code{if} \; b \; 
		\left\{ 
		\begin{array}l
		\left(\prod\limits_{\role R' \in \rolesFunc(\mathcal{I} ,
		\mathcal{I}')\setminus\{\role R\}} \; \idxTrue{i}
		\cout{\idxSign{i}.\AuxIf^*_\idxSign{i}}{\trueC}{\role R'}\right)
		\seqOp \; \\[15pt] \pi(\mathcal{I},\role R)
		\end{array}
		\right\}
		\\[10pt]
		\quad \code{else} \;  \left\{
		\begin{array}l
		\left(\prod\limits_{\role R' \in 
			\rolesFunc(\mathcal{I}, \mathcal{I}')\setminus \{\role R\}} \; 
		\idxFalse{i} \cout{\idxSign{i}.\AuxIf^*_\idxSign{i}}{\falseC}{\role R'}\right)
		\seqOp \; \\[15pt] \pi(\mathcal{I}',\role R) 	
		\end{array}
		\right\}
	\end{array}\right.
\\
\rulebox{
\begin{array}l
\pi(\ifthenKey{b \at \role R}{\mathcal{I}}{\mathcal{I}'}{i},\role S)
\\
\quad \mbox{and } \role S \in \rolesFunc(\mathcal{I} , \mathcal{I}')
\setminus\{\role R\}
\end{array}}
& = & 
	\idxIfRecv{i} \cinp{\idxSign{i}.\AuxIf^*_\idxSign{i}}{x_\idxSign{i}}{\role R} \seqOp \; 
	\idx{i} \ifthen{x_\idxSign{i}}{\pi(\mathcal{I},\role S)}{\pi(\mathcal{I}',\role S)}
\\[15pt]
\rulebox{
\begin{array}l
\pi(\ifthenKey{b \at \role R}{\mathcal{I}}{\mathcal{I}'}{i},\role S)\\
\quad \mbox{and } \role S \not \in \rolesFunc(\mathcal{I} , \mathcal{I}')
\cup \{\role R\}
\end{array}} & = &
\one
\\
\rulebox{\pi( \whileKey{b \at \role R}{\mathcal{I}}{i},\role R)} & = & 
	\left\{\begin{array}l
		\idx{i} \code{while} \; b \; \Big\{
				\\[10pt]
				\quad\left(\prod\limits_{\role R' \in \rolesFunc(\mathcal{I}) \setminus
				\{\role R\}} \idxTrue{i}
				\cout{\idxSign{i}.\AuxWb^*_\idxSign{i}}{\trueC}{\role R'}\right)
					\seqOp \quad \pi(\mathcal{I},\role R) \seqOp \;
					\\[15pt]
				\quad\prod\limits_{\role R' \in \rolesFunc(\mathcal{I}) \setminus
				\{\role R\}} \; \idxClose{i}
				\cinp{\idxSign{i}.\AuxWe^*_\idxSign{i}}{\_}{\role R'}
			\\[25pt]
			\Big\}  \seqOp
					\prod\limits_{\role R'\in \rolesFunc(\mathcal{I}) \setminus \{\role
					R\}} \; \idxFalse{i}
					\cout{\idxSign{i}.\AuxWb^*_\idxSign{i}}{\falseC}{\role R'}
	\end{array}\right.
\\
	\rulebox{
	\begin{array}l
	\pi(\whileKey{b \at \role R}{\mathcal{I}}{i},\role S)
	\\ \quad \mbox{and } \role S \in\rolesFunc(\mathcal{I}) \setminus \{\role R\}
\end{array}} & = &
	\left\{\begin{array}l
		\idxIfRecv{i} \cinp{\idxSign{i}.\AuxWb^*_\idxSign{i}}{x_\idxSign{i}}{\role
		R} \seqOp \\
		\quad \idx{i} \whileC 
		\ x_\idxSign{i}
		\left\{
		\begin{array}l
		\pi(\mathcal{I},\role S) \seqOp \;	\\
		\idxClose{i} \cout{\idxSign{i}.\AuxWe^*_\idxSign{i}}{\okC}{\role R} \seqOp \; \\
		\idxIfRecv{i} \cinp{\idxSign{i}.\AuxWb^*_\idxSign{i}}{x_\idxSign{i}}{\role R}\ 	
		\end{array}
		\right\}
	\end{array}\right.
\\[12pt]
	\rulebox{
	\begin{array}l
	\pi(\whileKey{b \at \role R}{\mathcal{I}}{i},\role S)
	\\ \quad \mbox{and } \role S \not\in\rolesFunc(\mathcal{I}) \cup \{\role R\}
\end{array}} & = & \one
\\[12pt]
\rulebox{\pi(\idx{i} \scope{l}{\role R}{\mathcal I}{\Delta},\role R)} & = &
	\pscope{i}{l}{\role R}{\pi({\mathcal I},\role R)}{\Delta}{\rolesFunc({\mathcal I})}
\\[12pt]
\rulebox{\begin{array}l
	\pi(\idx{i} \scope{l}{\role R}{\mathcal I}{\Delta},\role S)
	\\ \quad \mbox{and } \role S \in \rolesFunc(\mathcal{I}) \setminus \{\role
	R\}
\end{array}} & = &
	\psscope{i}{l}{\role R}{\pi({\mathcal I},\role S)}
\\[12pt]
\rulebox{\begin{array}l
\pi(\idx{i} \scope{l}{\role R}{\mathcal I}{\Delta},\role S)
	\\ \quad \mbox{and } \role S \not \in \rolesFunc(\mathcal{I}) \cup \{\role
	R\}
\end{array}}
 & = & \one
\end{array}
$$}
\vspace{-1em}\caption{process-projection function $\pi$.}\label{fig:pi}
\end{figure}

With a little abuse of notation, we write
$\rolesFunc(\mathcal{I},\mathcal{I}')$ for
$\rolesFunc(\mathcal{I})\cup\rolesFunc(\mathcal{I}')$. We assume that
variables $x_\idxSign{i}$ are never used in the \AIOC{} to be projected and we
use them for auxiliary synchronisations.

The projection is homomorphic for sequential and parallel composition, $\one$
and $\zero$. The projection of an assignment is the assignment on the role
performing it and $\one$ on other roles. The projection of an interaction is a
send on the sender role, a receive on the receiver, and $\one$ on any other
role.
The projection of a scope is a scope on all its participants. On its
coordinator it also features a clause that records the roles of the involved
participants. On the roles not involved in the scope the projection is $\one$.
Projections of conditional and while loop are a bit more complex, since they
need to coordinate a distributed computation. To this end they exploit
communications on auxiliary operations. In particular, $\AuxIf^*_\idxSign{i}$
coordinates the branching of conditionals, carrying information on whether the
``then'' or the ``else'' branch needs to be taken. Similarly,
$\AuxWb^*_\idxSign{i}$ coordinates the beginning of a while loop, carrying
information on whether to loop or to exit. Finally, $\AuxWe^*_\idxSign{i}$
coordinates the end of the body of the while loop. This closing operation
carries no relevant information and it is just used for synchronisation
purposes.  In order to execute a conditional $\ifthenKey{b \at \role
R}{\mathcal{I}}{\mathcal{I}'}{i}$, the coordinator $\role R$ of the
conditional locally evaluates the guard and tells the other roles which branch
to choose using auxiliary communications on $\AuxIf^*_\idxSign{i}$. Finally,
all the roles involved in the conditional execute their code corresponding to
the chosen branch.
Execution of a loop $\whileKey{b \at \role R}{\mathcal{I}}{i}$ is similar,
with two differences. First, end of loop synchronisation on operations
$\AuxWe^*_\idxSign{i}$ is used to notify the coordinator that an iteration is
terminated, and a new one can start. Second, communication of whether to loop
or to exit is more tricky than communication on the branch to choose in a
conditional. Indeed, there are two points in the projected code where the
coordinator $\role R$ sends the decision: the first is inside the body of the
loop and it is used if the decision is to loop; the second is after the loop
and it is used if the decision is to exit. Also, there are two points where
these communications are received by the other roles: before their loop at the
first iteration, at the end of the previous iteration of the body of the loop
in the others.

One has to keep attention since, by splitting an interaction into a send and a
receive primitive, primitives corresponding to different interactions, but on
the same operation, may interfere.

\begin{example}\label{example:indexes}\emph{
We illustrate the issue of interferences using the two  \APOC {} processes
below, identified by their respective roles, $\role R_1$ (right) and $\role
R_2$ (left), assuming that operations are not prefixed by indexes. We describe
only $\role R_1$ as $\role R_2$ is its dual. At Line 1, $\role R_1$ sends a
message to $\role R_2$ on operation $o$. In parallel with the send, $\role
R_1$ had a scope (Lines 3--5) that performed an update. The new code (Line 4)
contains a send on operation $o$ to role $\role R_2$. Since the two sends and
the two receives share the same operation $o$ and run in parallel, they can
interfere with each other.
$$
\begin{array}{c|c}
  \begin{array}l
    \mbox{\underline{process $\role R_1$}}										\\
    1. \quad \idx{1} o: e_1 \ \toC \ \role R_2 				\\
    2. \quad |																								\\
    3. \quad \text{\auxiliaryCode{// update auxiliary code}}\\
    4. \quad \idx{2} o: e_2 \ \toC \ \role R_2					\\
    5. \quad \text{\auxiliaryCode{// update auxiliary code}}
  \end{array} \quad
  &
  \quad \begin{array}l
    \mbox{\underline{process $\role R_2$}}\\
    1. \quad \idx{1} o: x_1 \ \fromC \ \role R_2         \\
    2. \quad |																									\\
   	3. \quad \text{\auxiliaryCode{// update auxiliary code}}  \\
    4. \quad \idx{2} o: x_2 \ \fromC \ \role R_2 				\\
    5. \quad \text{\auxiliaryCode{// update auxiliary code}}
  \end{array}
\end{array}
$$
\vspace{1em}
}\end{example}

Note that, since updates come from outside and one cannot know in advance
which operations they use, this interference cannot be statically avoided.

For this reason, in \S~\ref{sec:apoc} we introduced indexes to prefix \APOC{}
operations.

A similar problem may occur also for auxiliary communications. In particular,
imagine to have two parallel conditionals executed by the same role. We need
to avoid that, e.g., the decision to take the ``else'' branch on the first
conditional is wrongly taken by some role as a decision concerning the second
conditional. To avoid this problem, we prefix auxiliary operations using the
index $\idxSign{i}$ of the conditional. In this way, communications involving
distinct conditionals cannot interact. Note that communications concerning the
same conditional (or while loop) may share the same operation name and prefix.
However, since all auxiliary communications are from the coordinator of the
construct to the other roles involved in it, or vice versa, interferences are
avoided.

We now describe how to generate indexes for statements in the projection. As a
general rule, all the \APOC{} constructs obtained by projecting a \AIOC{}
construct with index $\idxSign{i}$ have index $\idxSign{i}$. The only
exceptions are the indexes of the auxiliary communications of the projection
of conditionals and while loops.

Provided $\idxSign{i}$ is the index of the conditional: \emph{i}) in the
projection of the coordinator we index the auxiliary communications for
selecting the ``then'' branch with index $\idxTrueSign i$, the ones for
selecting the ``else'' branch with index $\idxFalseSign i$; \emph{ii}) in the
projection of the other roles involved in the conditional we assign the index
$\idxIfRecvSign i$ to the auxiliary receive communications.
To communicate the evaluation of the guard of a while loop we use the same indexing
scheme ($\idxTrueSign i$, $\idxFalseSign i$, and $\idxIfRecvSign i$) used in the
projection of conditional. Moreover, all the auxiliary communications for end of
loop synchronisation are indexed with $\idxCloseSign i$.

\section{Running Example: Projection and Execution}\label{sec:a_running_example}

In this section we use our running example (see Figure \ref{fig:purchase_scenario}) to illustrate the projection and execution of \AIOC{} programs.

\subsection{Projection}\label{sub:demo_projection}
Given the code in Figure~\ref{fig:purchase_scenario}, we need to annotate it to
be able to project it (we remind that in \S~\ref{sec:endpoint_projection} we
defined our projection function on well-annotated \AIOC{}s). Since we wrote one
instruction per line in Figure~\ref{fig:purchase_scenario}, we annotate every
instruction using its line number as index. This results in a well-annotated \AIOC{}.

From the annotated \AIOC, the projection generates three \APOC{} processes for
the $\seller$, the $\buyer$, and the $\bank$, respectively reported in
Figures~\ref{fig:proj_seller},~\ref{fig:proj_buyer}, and~\ref{fig:proj_bank}.
To improve readability, we omit some $\one$ processes. In the projection of
the program, we also omit to write the index that prefixes the operations
since it is always equal to the numeric part of the index of their
correspondent construct. Finally, we write auxiliary communications in grey.
\begin{figure}
\begin{minipage}{.45\textwidth}
\begin{lstlisting}[mathescape=true,numbers=left,xleftmargin=-20pt]
$\auxiliaryCode{\idxIfRecv 3 \AuxWb^*_{3} : x_{3} \ \fromC \ \buyer;}$
$\idx 3 \whileC ( x_{3} ) \{$
 $\idx 5 priceReq : order \ \fromC \ \buyer;$
 $\idx 6 \scopeC \ @\seller \{$
  $\idx 7 order\_price = \extFunc{getPrice}( order );$
  $\idx 8 \mathit{offer} : order\_price \ \toC \ \buyer$
 $\} \ \rolesC \ \{ \ \seller,\ \buyer \ \};$
 $\auxiliaryCode{\idxClose{3} \AuxWe^*_{3} : \okC \ \toC \ \buyer;}$
 $\auxiliaryCode{\idxIfRecv 3 \AuxWb^*_{3} : x_{3} \ \fromC \ \buyer}$
$\};$
$\auxiliaryCode{\idxIfRecv {15} \AuxIf^*_{15} : x_{15} \ \fromC \ \buyer;}$
$\idx {15} \ifC ( x_{15} ) \{$
 $\idx {16} payReq : \extFunc{payDesc}( order\_price ) \ \toC \ \bank;$
 $\auxiliaryCode{\idxIfRecv {21} \AuxIf^*_{21} : x_{21} \ \fromC \ \bank;}$
 $\idx {21} \ifC ( x_{21} ) \{$
  $\idx {22} confirm : \_ \ \fromC \ \bank \ \}$ 
$\}$
\end{lstlisting}
\vspace{-15pt}
\caption{$\seller$ \APOC{} Process.\label{fig:proj_seller}}
\vspace{12pt}
\begin{lstlisting}[mathescape=true,numbers=left,xleftmargin=-20pt]
$\auxiliaryCode{\idxIfRecv{15} \AuxIf^*_{15} : x_{15} \ \fromC \ \buyer;}$
$\idx{15} \ifC ( x_{15} ) \{$
 $\idx{16} payReq : desc \ \fromC \ \seller;$
 $\idx{17} \scopeC \ @\bank \{$
  $\idx{18} pay : auth \ \fromC \ \buyer$
 $\} \ \rolesC \ \{ \buyer,\bank \};$
 $\idx{20} payment\_ok = \extFunc{makePayment}(desc, auth);$
 $\idx{21} \ifC ( payment\_ok ) \{$
  $\{$
   $\auxiliaryCode{\idxTrue{21} \AuxIf^*_{21} : \trueC \ \toC \ \seller}$
   $\auxiliaryCode{| \ \idxTrue{21} \AuxIf^*_{21} : \trueC \ \toC \ \buyer }$
  $\}\ ; \; \{$
   $\idx{22} confirm : null \ \toC \ \seller$
   $| \ \idx{24} confirm : null \ \toC \ \buyer$
  $\}$
 $\} \ \elseC \ \{$
  $\{$
   $\auxiliaryCode{\idxFalse{21} \AuxIf^*_{21} : \falseC \ \toC \ \seller}$
   $\auxiliaryCode{|\ \idxFalse{21} \AuxIf^*_{21} : \falseC \ \toC \ \buyer }$
  $\};$ 
  $\idx{26} abort : null \ \toC \ \buyer$
 $\}$ 
$\}$
\end{lstlisting}
\vspace{-15pt}
\caption{$\bank$ \APOC{} Process.\label{fig:proj_bank}}
\end{minipage}
\begin{minipage}{.45\textwidth}
\begin{lstlisting}[mathescape=true,numbers=left,xleftmargin=35pt]
$\idx{1} price\_ok = \falseC;$
$\idx{2} continue = \trueC;$
$\idx{3} \whileC ( !price\_ok \ \andC \ \ continue ) \{$
 $\auxiliaryCode{\idxTrue 3 \AuxWb^*_{3}: \trueC \ \toC \ \seller;}$
 $\idx{4} prod = \extFunc{getInput}();$
 $\idx{5} priceReq : prod \ \toC \ \seller;$
 $\idx{6} \scopeC \ @\seller \{$
  $\idx{7} \mathit{offer} : prod\_price \ \fromC \ \seller$
 $\};$
 $\idx{10} price\_ok = \extFunc{getInput}();$
 $\idx{11} \ifC ( !price\_ok ) \{ $
  $\idx{12} continue = \extFunc{getInput}() $
 $\};$
 $\auxiliaryCode{\idxClose{3} \AuxWe^*_{3} : \_ \ \fromC \ \seller}$
$\};$
$\auxiliaryCode{\idxFalse{3} \AuxWb^*_{3} : \falseC \ \toC \ \seller;}$
$\idx{15} \ifC ( price\_ok ) \{$
 $\{$
  $\auxiliaryCode{\idxTrue{15} \AuxIf^*_{15} : \trueC \ \toC \ \seller}$
  $\auxiliaryCode{| \idxTrue{15} \AuxIf^*_{15} : \trueC \ \toC \ \bank}$
 $\};$
 $\idx{17} \scopeC \ @\bank \{$
  $\idx{19} pay : \extFunc{payAuth}( prod\_price ) \ \toC \ \bank $
 $\};$
 $\auxiliaryCode{\idxIfRecv{21} \AuxIf^*_{21} : x_{21} \ \fromC \ \bank;}$
 $\idx{21} \ifC ( x_{21} ) \{$
  $\idx{24} confirm : \_ \ \fromC \ \bank$
 $\} \ \elseC \ \{$
  $\idx{26} abort : \_ \ \fromC \ \bank \}$
 $\}$
$\}$
\end{lstlisting}
\caption{$\buyer$ \APOC{} Process.\label{fig:proj_buyer}}
\end{minipage}
\end{figure}

\subsection{Runtime Execution}\label{sub:demo_execution}
\label{app:running_example}

We now focus on an excerpt of the code 
to exemplify how updates are performed at runtime. We consider the code
of the scope at Lines 6--9 of Figure~\ref{fig:purchase_scenario}.
In this execution scenario we assume to introduce in the set of available
updates the update presented in Figure~\ref{fig:rule_price_inquiry}, which
enables the use of a fidelity card to provide a price discount.
Below we consider both the \AIOC{} and the \APOC{} level, dropping some $\one$s
to improve readability.

Since we describe a runtime execution, we assume that the $\buyer$ has just
sent the name of the product (s)he is interested in to the $\seller$ (Line 5 of
Figure~\ref{fig:purchase_scenario}). The annotated \AIOC{}s we execute is the
following.
\begin{lstlisting}[mathescape=true,numbers=left,xleftmargin=120pt]
$\idx{6} \scopeC \ @\seller \{$
 $\idx{7} order\_price@\seller = \extFunc{getPrice}( order );$
 $\idx{8} \mathit{offer} : \seller( order\_price ) \rightarrow \buyer( prod\_price )$
$\}$
\end{lstlisting}
At runtime we apply Rule $\did{\AIOC}{\AdaptRule}$ that substitutes the scope
with the new code. The replacement is atomic. Below we assume that the instructions of the update are annotated with indexes corresponding to their line number plus 30.
\vspace{5pt}
\begin{lstlisting}[mathescape=true,numbers=left,xleftmargin=120pt]
$\idx{31} cardReq : \seller( null ) \rightarrow \buyer( \_ );$
$\idx{32} card\_id@\buyer = \extFunc{getInput}();$
$\idx{33} card : \buyer( card\_id ) \rightarrow \seller( buyer\_id );$
$\idx{34} \ifC \ \extFunc{isValid}( buyer\_id )@\seller \{$
  $\idx{35} order\_price@\seller = \extFunc{getPrice}( order ) * 0.9$
$\} \ \elseC \ \{$
  $\idx{37} order\_price@\seller = \extFunc{getPrice}( order )$
$\};$
$\idx{39} \mathit{offer} : \seller( order\_price ) \rightarrow \buyer( 
prod\_price )$
\end{lstlisting}
\vspace{5pt}

Let us now focus on the execution at \APOC{} level, where the application of 
updates is not atomic. The scope is distributed between two 
participants.
The first step of the update protocol is performed by the $\seller$, since
(s)he is the coordinator of the update. The \APOC{} description of the $\seller$ before the update is:

\vspace{5pt}
\begin{lstlisting}[mathescape=true,xleftmargin=140pt]
$\idx{6} \scopeC\ @\seller \{$
  $\idx 7 order\_price = \extFunc{getPrice}( order );$
  $\idx 8 \mathit{offer} : order\_price \ \toC \ \buyer$
$\} \ \rolesC \ \{ \seller, \buyer \}$
\end{lstlisting}
\vspace{5pt}

When the scope construct is enabled, the $\seller$ non-deterministically
selects whether to update or not and, in the first case, which update to
apply. Here, we assume that the update using the code in
Figure~\ref{fig:rule_price_inquiry} is selected. Below we report on the left
the reductum of the projected code of the $\seller$ after the application of
Rule $\did{\APOC}{Lead-\AdaptRule}$. The $\seller$ sends to the $\buyer$ the
code
--- denoted as $P_{\role B}$ and reported below on the right --- obtained
    projecting the update on role $\buyer$.

\noindent\;\begin{minipage}{.53\textwidth}
\begin{lstlisting}[mathescape=true,xleftmargin=1.5em,numbers=left]
$\auxiliaryCode{\idx{6} \AuxSb^*_{6} : P_{\role B} \ \toC \ \buyer;}$
$\idx{31} cardReq : null \ \toC \ \buyer;$
$\idx{33} card : buyer\_id \ \fromC \ \buyer;$
$\idx{34} \ifC \ \extFunc{isValid}( buyer\_id ) \{$
  $\idx{35} order\_price = \extFunc{getPrice}( order ) * 0.9$
$\} \ \elseC \ \{$
  $\idx{37} order\_price = \extFunc{getPrice}( order )$
$\};$
$\idx{39} \mathit{offer} : order\_price \ \toC \ \buyer;$
$\auxiliaryCode{\idx{6} \AuxSe^*_{6} : \_ \ \fromC \ \buyer;}$
\end{lstlisting}  
\end{minipage}
\begin{minipage}{.5\textwidth}
\begin{lstlisting}[mathescape=true]
$P_{\role B} :=\,$ $\idx{31} cardReq : null \ \fromC \ \seller;$
      $\!\idx{32} card\_id = \extFunc{getInput}();$
      $\!\idx{33} card : card\_id \ \toC \ \seller;$
      $\!\idx{39} \mathit{offer} : prod\_price \ \fromC \ \seller$
\end{lstlisting}  
\end{minipage}
\vspace{10pt}

Above, at Line 1 the $\seller$ requires the $\buyer$ to update, sending to him
the new \APOC{} fragment to execute. Then, the $\seller$ starts to execute its
own updated \APOC{}. At the end of the execution of the new \APOC{} code (Line
10) the $\seller$ waits for the notification of termination of the \APOC{}
fragment executed by the $\buyer$.

Let us now consider the process-projection of the $\buyer$, reported below.
\vspace{5pt}
\begin{lstlisting}[mathescape=true,xleftmargin=140pt]
$\idx{6} \scopeC \ @\seller \{$
  $\idx 8 \mathit{offer} : order\_price \ \fromC \ \seller$
$\}$
\end{lstlisting}
\vspace{5pt}
At runtime, the scope waits for the arrival of a message from the coordinator
of the update.
In our case, since we assumed that the update is applied, the $\buyer$
receives using Rule $\did{\APOC}{\AdaptRule}$ the \APOC{} fragment $P_{\role
B}$ sent by the coordinator.
In the reductum, $P_{\role B}$ replaces the scope, followed by the
notification of termination to the $\seller$.

\vspace{5pt}
\begin{lstlisting}[mathescape=true,xleftmargin=140pt]
$\idx{31} cardReq : null \ \fromC \ \seller;$
$\idx{32} card\_id = \extFunc{getInput}();$
$\idx{33} card : card\_id \ \toC \ \seller;$
$\idx{39} \mathit{offer} : prod\_price \ \fromC \ \seller$
$\auxiliaryCode{\idx 6 \AuxSe^*_{6} : \okC \ \toC \ \seller}$
\end{lstlisting}
\vspace{5pt}

Consider now what happens if no update is applied. At \AIOC{} level the
$\seller$ applies Rule $\did{\AIOC}{No\AdaptRule}$, which removes the scope
and runs its body.
At \APOC{} level, the update is not atomic. The code of the $\seller$ is the
following one.

\vspace{5pt}
\begin{lstlisting}[mathescape=true,xleftmargin=140pt,numbers=left]
$\auxiliaryCode{\idx 6 \AuxSb^*_{6} : \noC \ \toC \ \buyer;}$
$\idx 7 order\_price = \extFunc{getPrice}( order );$
$\idx 8 \mathit{offer} : order\_price \ \toC \ \buyer;$
$\auxiliaryCode{\idx 6 \AuxSe^*_{6} : \_ \ \fromC \ \buyer;}$
\end{lstlisting}
\vspace{5pt}
Before executing the code inside the scope, the $\seller$ notifies the
$\buyer$ that (s)he can proceed with her execution (Line 1). Like in the case
of update, the $\seller$ also waits for the notification of the end of
execution from the $\buyer$ (Line 4).

Finally, we report the \APOC{} code of the $\buyer$ after the reception of
the message that no update is needed. Rule $\did{\APOC}{No\AdaptRule}$ removes
the scope and adds the notification of termination (Line 2 below) to the
coordinator at the end.

\vspace{5pt}
\begin{lstlisting}[mathescape=true,xleftmargin=140pt,numbers=left]
$\idx 7 \mathit{offer} : prod\_price \ \fromC \ \seller;$
$\auxiliaryCode{\idx 6 \AuxSe^*_{6} : \okC \ \toC \ \seller;}$
\end{lstlisting}

\section{Connected \AIOC{}s}
\label{sub:process_level_coordination}
\label{sub:connectedness}

We now give a precise definition of the notion of connectedness that we
mentioned in \S~\ref{sub:annotated_aioc_semantics} and
\S~\ref{sec:apoc_semantics}. In both \AIOC{} and \APOC{} semantics we checked
such a property of updates with predicate $\func{connected}$, respectively in
Rule $\did{\AIOC}{\AdaptRule}$ (Figure~\ref{fig:ioclts}) and Rule
$\did{\APOC}{\AdaptRule-Lead}$ (Figure~\ref{fig:apoc-proc_2}).

To give the intuition of why we need to restrict to connected updates, consider
the scenario below of a \AIOC{} (left side) and its projection (right side).
$$
\begin{array}l
op1 : \role A( e_1 ) \rightarrow \role B( x );\\[3pt]
op2 : \role C( e_2 ) \rightarrow \role D( y )
\end{array}
\qquad
\makeatletter
\ext@arrow 0359\Rightarrowfill@{}{projection}
\makeatother
\qquad
\begin{array}{c|c}
  \begin{array}l
    \mbox{\underline{process $\role A$}}\\
    \quad op_1: e_1 \ \toC \ \role B
  \end{array}
  &
  \begin{array}l
    \mbox{\underline{process $\role B$}}\\
    \quad op_1: x \ \fromC \ \role A 
  \end{array}
  \\[10pt]
  \hline
  \\[-10pt]
  \begin{array}l
    \mbox{\underline{process $\role C$}}\\
    \quad op_2: e_2 \ \toC \ \role D 
  \end{array}
  &
  \begin{array}l
    \mbox{\underline{process $\role D$}}\\
    \quad op_2: y \ \fromC \ \role C
  \end{array}
\end{array}
$$
\AIOC{}s can express interactions that, if projected as described in
\S~\ref{sec:endpoint_projection}, can behave differently with respect to the
originating \AIOC{}. Indeed, in our example we have a
\AIOC{} that composes in sequence two interactions: an interaction between
$\role A$ and $\role B$ on operation $op_1$ followed by an interaction between
$\role C$ and $\role D$ on operation $op_2$. The projection of the \AIOC{}
produces four processes (identified by their role): $\role A$ and $\role C$
send a message to $\role B$ and $\role D$, respectively. Dually, $\role B$ and
$\role D$ receive a message form $\role A$ and $\role C$, respectively.
In the example, at the level of processes we lose the global order among the
interactions: each projected process runs its code locally and it is not aware
of the global sequence of interactions. Indeed, both sends and both receives
are enabled at the same time. Hence, the semantics of \APOC{} lets the two
interactions interleave in any order. It can happen that the interaction
between $\role C$ and $\role D$ occurs before the one between $\role A$ and
$\role B$, violating the order of interactions prescribed by the originating
\AIOC{}.

Restricting to connected \AIOC{}s avoids this kind of behaviours. We formalise
\emph{connectedness} as an efficient (see Theorem~\ref{teo:compl}) syntactic check.
We highlight that our definition of connectedness does not hamper
programmability and it naturally holds in most real-world scenarios (the
interested reader can find in the website of the AIOCJ project~\cite{AIOCJ}
several such scenarios).

\begin{remark}\emph{
There exists a trade-off between efficiency and ease of programming with
respect to the guarantee that all the roles are aware of the evolution of the
global computation.
This is a common element of choreographic approaches, which has been handled
in different ways, e.g., \emph{i}) by restricting the set of well-formed
choreographies to only those on which the projection preserves the order of
actions~\cite{hondaesop}; \emph{ii}) by mimicking the non-deterministic
behaviour of process-level networks at choreography level~\cite{poplmontesi};
or \emph{iii}) by enforcing the order of actions with additional auxiliary
messages between roles~\cite{wwvlanese}.}

\emph{Our choice of preserving the order of interactions defined at \AIOC{} level
follows the same philosophy of~\cite{hondaesop}, whilst for scopes,
conditionals, and while loops we enforce connectedness with auxiliary messages
as done in~\cite{wwvlanese}.
We remind that we introduced auxiliary messages for coordination both in the
semantics of scopes at \APOC{} level (\S~\ref{sec:apoc_semantics}) and in the
projection (\S~\ref{sec:endpoint_projection}). We choose to add such auxiliary
messages to avoid to impose strong constraints on the form of scopes,
conditionals, and while loops, which in the end would pose strong limitations
to the programmers of \AIOC{}s.
On the other hand, for sequential composition we choose to restrict the set of
allowed $\AIOC{}$s by requiring connectedness, which ensures that the order of
interactions defined at \AIOC{} level is preserved by projection.}

\emph{As discussed above, the execution of conditionals and while loops rely on
auxiliary communications to coordinate the different roles. Some of these
communications may be redundant. For instance, in
Figure~\ref{fig:proj_seller}, Line 8, the $\seller$ notifies to the $\buyer$
that (s)he has completed her part of the while loop. However, since her last
contribution to the while loop is the auxiliary communication for end of scope
synchronisation, the $\buyer$ already has this information. Hence, Line 8 in
Figure~\ref{fig:proj_seller}, where the notification is sent, and Line 14 in
Figure~\ref{fig:proj_buyer}, where the notification is received, can be safely
dropped. Removing redundant auxiliary communications can be automatised using
a suitable static analysis. We leave this topic for future work.
}\end{remark}

To formalise connectedness we introduce, in Figure~\ref{fig:tri}, the auxiliary
functions $\transI$ and $\transF$ that, given a \AIOC{} process, compute sets
of pairs representing senders and receivers of possible initial and final
interactions in its execution. We represent one such pair as
$\rolesFuncPair{\role R}{\role S}$. Actions located at $\role R$ are
represented as $\rolesFuncPair{\role R}{\role R}$. For instance, given an
interaction $\idx{i}\comm{o}{\role R}{e}{\role S}{x}$ both its $\transI$ and
$\transF$ are $\{\rolesFuncPair{\role R}{\role S}\}$. For conditional,
$\transI( \idx{i}\ifthen{b \at \role R}{\mathcal{I}}{\mathcal{I}'}) =
\{\rolesFuncPair{\role R}{\role R}\}$ since the first action executed is the
evaluation of the guard by role $\role R$. The set $\transF(\idx{i}\ifthen{b
\at \role R}{\mathcal{I}}{\mathcal{I}'})$ is normally $\transF({\mathcal I})  \
\cup \ \transF({\mathcal I}')$, since the execution terminates with an action
from one of the branches. If instead the branches are both empty then $\transF$
is $\{ \rolesFuncPair{\role R}{\role R} \}$, representing guard evaluation.
\begin{figure}[t]
$$
\begin{array}{lcl}	
\transI(\idx{i}\comm{o}{\role R}{e}{\role S}{x}) & = & 
\transF(\idx{i}\comm{o}{\role R}{e}{\role S}{x}) = \{\rolesFuncPair{\role R}{\role S}\}
\\[5pt]
\transI(\idx{i}\assign{x}{\role R}{e})  & = & 
	\transF(\idx{i}\assign{x}{\role R}{e}) = \{\rolesFuncPair{\role R}{\role R}\}
\\[5pt]
\transI(\one) & = & \transI(\zero) = \transF(\one) = \transF(\zero) = \emptyset
\\[5pt]
\transI({\mathcal I} \parOpI {\mathcal I}') & = & 
	\transI({\mathcal I}) \cup \transI({\mathcal I}')
\\[5pt]
\transF({\mathcal I} \parOpI {\mathcal I}') & = & 
	\transF({\mathcal I}) \cup \transF({\mathcal I}')
\\[5pt]
\transI({\mathcal I}\seqOp{\mathcal I}')  & = & 
\begin{cases}
	\transI({\mathcal I}') & \mbox{if } \transI({\mathcal I}) = \emptyset\\
	\transI({\mathcal I}) & \mbox{otherwise}\\
\end{cases}
\\[15pt]
\transF({\mathcal I}\seqOp{\mathcal I}')  & = & 
\begin{cases}
	\transF({\mathcal I}) & \mbox{if } \transF({\mathcal I}') = \emptyset\\
	\transF({\mathcal I}') & \mbox{otherwise}\\
\end{cases}
\\[15pt]
\transI( \idx{i}\ifthen{b \at \role R}{\mathcal{I}}{\mathcal{I}'}) & = & \transI(\idx{i}\while{b \at \role R}{\mathcal{I}} ) =
\{\rolesFuncPair{\role R}{\role R}\}
\\[5pt]
\transF( \idx{i}\ifthen{b \at \role R}{\mathcal{I}}{\mathcal{I}'}) & = &
\begin{cases}
	\{ \rolesFuncPair{\role R}{\role R} \} & \mbox{if } \transF({\mathcal I}) \cup \transF({\mathcal I}') = \emptyset\\
	\transF({\mathcal I}) \cup \transF({\mathcal I}') & \mbox{otherwise}\\
\end{cases}
\\[15pt]
\transF(\idx{i}\while{b \at \role R}{\mathcal{I}} ) & = & 
\begin{cases}	
	\{ \rolesFuncPair{\role R}{\role R} \} & 
		\mbox{if } \transF({\mathcal I}) = \emptyset\\
	\bigcup\limits_{\role R' \in \rolesFunc(\mathcal{I}) \smallsetminus \{\role R\}} \{ \rolesFuncPair{\role R}{\role R'} \} & \mbox{otherwise}\\
\end{cases}
\\[15pt]
\transI(\idx{i}\scope{l}{\role R}{\mathcal I}{\Delta}) & = & 
	\{ \rolesFuncPair{\role R}{\role R}\}
\\[5pt]
\transF(\idx{i}\scope{l}{\role R}{\mathcal I}{\Delta}) & = &
\begin{cases}	
	\{ \rolesFuncPair{\role R}{\role R} \} & 
		\mbox{if } \rolesFunc({\mathcal I}) \subseteq \{ \role R \}\\
	\bigcup\limits_{\role R' \in \rolesFunc(\mathcal{I}) \smallsetminus \{\role R\}} \{ \rolesFuncPair{\role R'}{\role R} \} & \mbox{otherwise}
\end{cases}
\end{array}
$$
\caption{Auxiliary functions $\transI$ and $\transF$.}\label{fig:tri} 
\end{figure}

Finally, we give the formal definition of connectedness.

\begin{definition}[Connectedness]\label{def:connectedness}\emph{
A \AIOC{} process ${\mathcal I}$ is connected 
if each
subterm ${\mathcal I}' \seqOp {\mathcal I}''$ of $\mathcal I$ satisfies
$$\forall\ \rolesFuncPair{\role R_1}{\role R_2} \in \transF({\mathcal I}'), 
\forall\ \rolesFuncPair{\role S_1}{\role S_2} \in \transI({\mathcal I}'') \; . \;
\{\role R_1,\role R_2\} \cap \{\role S_1,\role S_2\} \neq \emptyset$$
}\end{definition}

Connectedness can be checked efficiently.

\begin{restatable}[Connectedness-check complexity]{theorem}{compl}\label{teo:compl}
\mbox{}\\
The connectedness of a \AIOC{} process ${\mathcal I}$ can be checked in
time $O(n^2\log(n))$, where $n$ is the number of nodes in the abstract syntax
tree of ${\mathcal I}$.
\end{restatable}

The proof of the theorem is reported in Appendix~\ref{app:complexity}.

We remind that we allow only connected updates. Indeed, replacing a
scope with a connected update always results in a deadlock- and
race-free \AIOC{}. Thus, one just needs to statically check
connectedness of the starting program and of the updates, and there is
no need to perform expensive runtime checks on the whole application
after updates have been performed.

\section{Correctness}\label{sec:corr}

In the previous sections we have presented \AIOC{}s, \APOC{}s, and described
how to derive a \APOC{} from a given \AIOC{}. This section presents the main
technical result of the paper, namely the correctness of the projection.
Moreover, as a consequence of the correctness, in Section
\ref{subsec:deadlock} we prove that properties like deadlock freedom, 
termination, and race freedom are preserved by the projection.

Correctness here means that a connected \AIOC{} and its projected \APOC{} are
weak system bisimilar. Weak system bisimilarity is formally defined as
follows.

\begin{definition}[Weak System Bisimilarity]\label{def:bisim}\emph{ A \emph{weak
system bisimulation} is a relation $\mathcal R$ between \AIOC{} systems and
\APOC{} systems such that if $(\tuple{\state,\rules,{\mathcal
I}},\tuple{\rules',\net})
\in \mathcal R$ then:
\begin{itemize}
	\item if $\tuple{\state,\rules,{\mathcal I}} \arro{\mu}
	\tuple{\state'',\rules'',{\mathcal I}''}$ then $\tuple{\rules',\net}
	\arro{\eta_1},\ldots,\arro{\eta_k}\arro{\mu} \tuple{\rules''',\net'''}$
	with \\$\forall \ i \in [1 \ldots k], \eta_{i} \in \{
	\commLabel{o^*}{\role	R_1}{v}{\role R_2}{x}, 
	\commLabel{o^*}{\role	R_1}{X}{\role R_2}{},
	\tau \}$
	and\\ $(\tuple{\state'',\rules'',{\mathcal I}''},\tuple{\rules''',\net'''})
	\in
	\mathcal R$;
	\item if $\tuple{\rules',\net} \arro{\eta}
	\tuple{\rules''',\net'''}$ with $\eta \in \{\commLabel{o^?}{\role
	R_1}{v}{\role R_2}{x}, \commLabel{o^*}{\role
	R_1}{X}{\role R_2}{}, \tick, {\mathcal I},
	\\ \NoAdaptLabel, \rules''', \tau \}$ then one of the following two conditions holds:
	\begin{itemize}
	\item $\tuple{\state,\rules,{\mathcal I}} \arro{\eta}
	\tuple{\state'',\rules',{\mathcal I}''}$ and 
	$(\tuple{\state'',\rules'',{\mathcal I}''},$ $\tuple{\rules''',\net'''})
	\in \mathcal R$ or
	\item $\eta \in \{\commLabel{o^*}{\role R_1}{v}{\role
	R_2}{x},\commLabel{o^*}{\role R_1}{X}{\role R_2}{},\tau\}$ and
	$(\tuple{\state,\rules,{\mathcal I}},$ $\tuple{\rules''',\net''')} \in
	\mathcal R$
	\end{itemize}
\end{itemize}
}

\emph{\noindent A \AIOC{} system $\tuple{\state,\rules,{\mathcal I}}$ and a 
\APOC{} system  $\tuple{\rules',\net}$ are \emph{weak system bisimilar} iff
there exists a weak system bisimulation $\mathcal R$ such that
$(\tuple{\state,\rules,{\mathcal I}},\tuple{\rules',\net}) \in \mathcal R$.
}\end{definition}

In the proof, we provide a relation $\mathcal R$ which relates each
well-annotated connected \AIOC{} system with its projection and show that it
is a weak system bisimulation. Such a relation is not trivial since events
that are atomic in the \AIOC{}, e.g., the evaluation of the guard of a
conditional, including the removal of the discarded branch, are not atomic at
\APOC{} level. In the case of conditional, the \AIOC{} transition is mimicked
by a conditional performed by the role evaluating the guard, a set of
auxiliary communications sending the value of the guard to the other roles,
and local conditionals based on the received value.  These mismatches are
taken care by function $\upd$ (Definition~\ref{def:upd}). This function needs
also to remove the auxiliary communications used to synchronise the
termination of scopes, which have no counterpart after the \AIOC{} scope has
been consumed. However, we have to record the impact of the mentioned
auxiliary communications on the possible executions. Thus we define an event
structure for \AIOC{} (Definition~\ref{def:ev}) and one for \APOC{}
(Definition~\ref{def:apocev}) and we show that the two are related
(Lemma~\ref{lemma:ev}).

Thanks to the existence of a bisimulation relating each well-annotated
connected \AIOC{} system with its projection we can prove that the
projection is correct.  Formally:

\begin{restatable}[Correctness]{theorem}{final}\label{teo:final}
For each initial, connected \AIOC{} process ${\mathcal I}$, each state $\state$,
each set of updates $\rules$, the \AIOC{} system $\tuple{\state,
\rules,{\mathcal I}}$ and the \APOC{} system $\tuple{\rules,\proj({\mathcal
I},\state)}$ are weak system bisimilar.
\end{restatable}

As a corollary of the result above, a \AIOC{} system and its projection are
also trace equivalent. Trace equivalence is defined as follows.

\begin{definition}[Trace equivalence]\emph{
 A \AIOC{} system $\tuple{\state, \rules,{\mathcal I}}$ and a \APOC{}
 system $\tuple{ \rules,\net}$ are \emph{(weak) trace equivalent} iff
 their sets of (weak) traces coincide.
}\end{definition}

The following lemma shows that, indeed, weak system bisimilarity implies  weak
trace equivalence.

\begin{lemma}\label{lemma:bis2trsynch} Let
$\tuple{\state,\rules,{\mathcal I}}$ be a \AIOC{} system and
$\tuple{\rules',\net}$ a \APOC{} system. \\If $\tuple{\state,\rules,{\mathcal
I}} \bisim \tuple{\rules',\net}$ then the \AIOC{} system
$\tuple{\state,\rules,{\mathcal I}}$ and the \APOC{} system
$\tuple{\rules',\net}$ are weak trace equivalent.
\end{lemma}

\proof The proof is by coinduction.  Take a \AIOC{} trace
$\mu_1,\mu_2,\dots$ of the \AIOC{} system. From bisimilarity, the \APOC{}
system has a sequence of transitions with labels $\eta_1,\ldots,\eta_k,\mu_1$
where $\eta_1,\ldots,\eta_k$ are weak transitions. Hence, the first label in
the weak trace is $\mu_1$.
After the transition with label $\mu_1$, the \AIOC{} system and the \APOC{}
system are again bisimilar. By coinductive hypothesis, the \APOC{} system has
a weak trace $\mu_2,\ldots$. By composition the \APOC{} system has a trace
$\mu_1,\mu_2,\dots$ as desired. The opposite direction is similar.\qed

Hence the following corollary holds.

\begin{corollary}[Trace Equivalence]
For each initial, connected \AIOC{} process ${\mathcal I}$, each state $\state$,
each set of updates $\rules$, the \AIOC{} system $\tuple{\state,
\rules,{\mathcal I}}$ and the \APOC{} system $\tuple{\rules,\proj({\mathcal
I},\state)}$ are weak trace equivalent.
\end{corollary}

\proof
  It follows from Theorem~\ref{teo:final} and Lemma~\ref{lemma:bis2trsynch}.
\qed

The following section presents the details of the proof of
Theorem~\ref{teo:final}. The reader not interested in such details can
safely skip \S~\ref{sec:detailed_proof} and go to \S~\ref{subsec:deadlock}.

\subsection{Detailed Proof of Correctness}\label{sec:detailed_proof}

In our proof strategy we rely on the distinctness of indexes of \AIOC{}
constructs that, unfortunately, is not preserved by transitions due to while
unfolding.

\begin{example}\label{example:while_indexes}\emph{
Consider the \AIOC $\whileKey{b \at \role{R}}{\idx{j}
\assign{x}{\role R}{e}}{i}$. If the condition $b$ evaluates to true, in one
step the application of Rule $\did{\APOC}{While-unfold}$ produces the \AIOC $\mathcal I$
below $$\idx{j}\assign{x}{\role R}{e} \seqOp \whileKey{b \at
\role{R}}{\idx{j}\assign{x}{\role R}{e}}{i}$$ where the index $\idxSign{j}$
occurs twice.	
}\end{example}

To solve this problem, instead of using indexes, we rely on \emph{global
indexes} built on top of indexes. Global indexes can be used both at the
\AIOC{} level and at the \APOC{} level and their distinctness is preserved by
transitions.

\begin{definition}[Global index]\emph{ Given an annotated \AIOC{} process
$\mathcal{I}$, or an annotated \APOC{} network $\net$, for each annotated
construct with index $\iota$ we define its global index $\xi$ as 
follows:
\begin{itemize}
\item if the construct is not in the body of a while loop then $\xi = \iota$;
\item if the innermost while construct that contains the considered construct
has global index $\xi'$ then the considered construct has global index $\xi =
\xi':\iota$.
\end{itemize}
}\end{definition}

\begin{example}\emph{
Consider the \AIOC $\mathcal I$ in Example~\ref{example:while_indexes}. The
first assignment with index $\idxSign j$ also has global index $\idxSign j$,
while the second assignment with index $\idxSign j$ has global index $\idx
i\idxSign j$, since this last assignment is inside a while loop with global
index $\idxSign i$.
}\end{example}

\begin{restatable}[Distinctness of Global Indexes]{lemma}{lemmaDistinct}\label{lemma:distinct}
Given a well-annotated \AIOC{} $\mathcal{I}$, a global state $\state$, and a set
of updates $\rules$, if $\tuple{\state, \rules,{\mathcal I}}
\arro{\eta_1}\ldots\arro{\eta_n} \tuple{\state', \rules',{\mathcal I}'}$ then
all global indexes in $\mathcal I'$ are distinct.
\end{restatable}

\begin{proof}
	The proof is by induction on the number $n$ of transitions. Details are in
	Appendix~\ref{appendix:proof}.
\end{proof}

Using global indexes we can now define event structures corresponding to the
execution of \AIOC{}s and \APOC{}s. We start by defining \AIOC{} events. Some
events correspond to transitions of the \AIOC{}, and we say that they are
enabled when the corresponding transition is enabled, executed when the
corresponding transition is executed.

\begin{definition}[\AIOC{} events]\label{def:ev}\emph{ We use $\ev$ to range over
events, and we write $[\ev]_{\role R}$ to highlight that event $\ev$ is
performed by role $\role R$.
An annotated \AIOC{} ${\mathcal I}$ contains the following events:}

\emph{{\bf Communication events:} a sending event $\xi : \co{o}{\role R_2}$ in role
$\role R_1$ and a receiving event $\xi: \ci{o}{\role R_1}$ in role $\role R_2$
for each interaction $\idx{i} \comm{o}{\role R_1}{e}{\role R_2}{x}$ with global
index $\xi$; we also denote the sending event as $f_\xi$ or $[f_\xi]_{\role
R_1}$ and the receiving event as $t_\xi$ or $[t_\xi]_{\role R_2}$. Sending and
receiving events correspond to the transition executing the interaction.}

\emph{{\bf Assignment events:} an assignment event $\ev_\xi$ in role $\role R$ for
each assignment $\idx{i} \assign{x}{\role R}{e}$ with global index $\xi$; the
event corresponds to the transition executing the assignment.}

\emph{{\bf Scope events:} a scope initialisation event $\uparrow_{\xi}$ and a scope
termination event $\downarrow_{\xi}$ for each scope $\idx{i} \scope{l}{\role
R}{\mathcal{I}}{\Delta}{A}$ with global index $\xi$. Both these events belong
to all the roles in $\rolesFunc(\mathcal{I})$. The scope initialisation event
corresponds to the transition performing or not performing an update on the
given scope. The scope termination event is just an auxiliary event (related to
the auxiliary interactions implementing the scope termination).}

\emph{{\bf If events:} a guard if-event $\ev_\xi$ in role $\role R$ for each
construct $\ifthenKey{b \at \role R}{\mathcal{I}}{\mathcal{I}'}{i}$ with global
index $\xi$; the guard-if event corresponds to the transition evaluating the
guard of the condition.}

\emph{{\bf While events:} a guard while-event $\ev_\xi$ in role $\role R$ for each
construct $\whileKey{b \at \role R}{\mathcal{I}}{i}$ with global index $\xi$;
the guard-while event corresponds to the transition evaluating the guard of the
while loop.}

\emph{Function $\event({\mathcal I})$ denotes the set of events of the annotated
\AIOC{} ${\mathcal I}$. A sending and a receiving event with the same global
index $\xi$ are called matching events. We denote with $\overline{\ev}$ an
event matching event $\ev$. A communication event is either a sending event or
a receiving event. A communication event is unmatched if there is no event
matching it.
}\end{definition}

As a corollary of Lemma~\ref{lemma:distinct} events have distinct names.
Note also that, for each while loop, there are events corresponding to the
execution of just one iteration of the loop. If unfolding is performed, new
events are created.

Similarly to what we have done for \AIOC, we can define events for \APOC as 
follows.

\begin{definition}[\APOC{} events]\label{def:apocev}\emph{
An annotated \APOC{} network $\net$ contains the following events:
\begin{description}
\item[Communication events]
a sending event $\xi: \co{o^?}{\role R_2}$ in role $\role R_1$ for each send
$\idx{\iota} \cout{\idxSign{i}.o^?}{e}{\role R_2}$  with global index $\xi$ in
role $\role R_1$; and a receiving event $\xi: \ci{o^?}{\role R_1}$ in role
$\role R_2$ for each receive $\idx{\iota} \cinp{\idxSign{i}.o^?}{x}{\role R_1}$
with global index $\xi$ in role $\role R_2$; we also denote the sending event
as $f_\xi$ or $[f_\xi]_{\role R_1}$; and the receiving event as $t_\xi$ or
$[t_\xi]_{\role R_2}$. Sending and receiving events correspond to the
transitions executing the corresponding communication.
\item[Assignment events] an assignment event $\ev_\xi$ in role $\role R$ for
each assignment $\idx{i} x = e$ with global index $\xi$; the event corresponds
to the transition executing the assignment.
\item[Scope events] a scope initialisation event $\uparrow_{\xi}$ and a scope
termination event $\downarrow_{\xi}$ for each $\pscope{i}{l}{\role
R}{P}{\Delta}{S}$ or $\psscope{i}{l}{\role R}{P}$ with global index $\xi$.
Scope events with the same global index coincide, and thus the same event may
belong to different roles; the scope initialisation event corresponds to the
transition performing or not performing an update on the given scope for the
role leading the update. The scope termination event is just an auxiliary event
(related to the auxiliary interactions implementing the scope termination).
\item[If events] a guard if-event  $\ev_\xi$ in role $\role R$ for each
construct $\ifthenKey{b}{P}{P'}{i}$ with global index $\xi$; the guard-if
event corresponds to the transition evaluating the guard of the condition.
\item[While events] a guard while-event  $\ev_\xi$ in role $\role R$ for
each construct $\whileKey{b}{P}{i}$ with global index $\xi$; the guard-while
event corresponds to the transition evaluating the guard of the while loop.
\end{description}
Let $\event(\net)$ denote the set of events of the network $\net$. A sending
and a receiving event with either the same global index $\xi$ or with global
indexes differing only for replacing index $\idxSign{i}_?$ with
$\idxTrueSign{i}$ or $\idxFalseSign i$ are called matching events. We denote
with $\overline{\ev}$ an event matching event $\ev$.
}\end{definition} 
With a slight abuse of notation, we write $\event(P)$ to
denote events originated by constructs in process $P$, assuming the network
$\net$ to be understood.
We use the same syntax for events of \AIOC{}s and of \APOC{}s. Indeed, the two
kinds of events are strongly related (cf.\ Lemma~\ref{lemma:ev}).

We define below a causality relation $\leqaioc$ among \AIOC{} events based on
the constraints given by the semantics on the execution of the corresponding
transitions.

\begin{definition}[\AIOC{} causality relation]\label{def:causalAIOC}\emph{ Let us
consider an annotated \AIOC{} ${\mathcal I}$.  A causality relation $\leqaioc
\, \subseteq \, \devent({\mathcal I}) \times \devent({\mathcal I})$ is the
minimum reflexive and transitive relation satisfying:}

\emph{{\bf Sequentiality:} let ${\mathcal I}' \seqOp {\mathcal I}''$ be a subterm
of \AIOC{} ${\mathcal I}$. If $\ev'$ is an event in ${\mathcal I}'$ and
$\ev''$ is an event in ${\mathcal I}''$, then $\ev'\leqaioc \ev''$.}

\emph{{\bf Scope:} let $\idx{i} \adapt{{\mathcal I}'}{l}\Delta{\role R}$ be a subterm
of \AIOC{} ${\mathcal I}$. If $\ev'$ is an event in ${\mathcal I}'$ then
$\uparrow_{\xi} \leqaioc \ev' \leqaioc \downarrow_{\xi}$.}

\emph{{\bf Synchronisation:} for each interaction the sending event precedes the
receiving event.}

\emph{{\bf If:} let $\ifthenKey{b \at \role{R}}{\mathcal{I}''}{\mathcal{I}''}{i}$ be a
subterm of \AIOC{} $\mathcal{I}$, let $\ev_\xi$ be the guard if-event in role
$\role R$, then for every event $\ev'$ in ${\mathcal I}'$ and for every event
$\ev''$ in ${\mathcal I}''$ we have $\ev_\xi \leqaioc \ev'$ and $\ev_\xi
\leqaioc \ev''$.}

\emph{{\bf While:} let $\whileKey{b \at \role R}{\mathcal{I}'}{i}$ be a subterm of
\AIOC{} $\mathcal{I}$, let $\ev_\xi$ be the guard while-event in role $\role
R$, then for every event $\ev'$ in ${\mathcal I}'$ we have $\ev_\xi
\leqaioc \ev'$.
}\end{definition}

As expected, the relation $\leqaioc$ is a partial order.

\begin{lemma}
	Let us consider an annotated $\AIOC$ $\mathcal{I}$. The relation $\leqaioc$
	among events of $\mathcal{I}$ is a partial order.
\end{lemma}

\begin{proof}
	A partial order is a relation which is reflexive, transitive, and
	antisymmetric. Reflexivity and transitivity follow by definition. We show
	antisymmetry by showing that $\leqaioc$ does not contain any cycle. The
	proof is by structural induction on the $\AIOC$ $\mathcal{I}$. 
	The base cases are interaction, assignment, $\one$, and $\zero$, which are
	all trivial.
	In the case of sequence, $\mathcal{I}';\mathcal{I}''$, by inductive
	hypothesis there are no cycles among the events of $\mathcal{I}'$ nor among
	events of $\mathcal I''$. Since all events of $\mathcal{I}'$ precede all
	events of $\mathcal{I}''$, there are no cycles among the events of
	$\mathcal{I}';\mathcal{I}''$.
	In the case of parallel composition, $\mathcal{I}' \parOpI \mathcal{I}''$,
	there is no relation between events in $\mathcal{I}'$ and in
	$\mathcal{I}''$, hence the thesis follows.
	In the case of conditional $\ifthenKey{b \at
	\role{R}}{\mathcal{I}'}{\mathcal{I}''}{i}$ there are no relations between
	events in $\mathcal{I}'$ and in $\mathcal{I}''$ and all the events follow
	the guard-if event. Hence the thesis follows.
	The cases of while and scope are similar.
\end{proof}

We can now define a causality relation $\leqapoc$ among $\APOC$ events.

\begin{definition}[\APOC{} causality relation]\label{def:causalapoc}\emph{ Let us
consider an annotated \APOC{} network $\net$.  A causality relation $\leqapoc
\, \subseteq \, \event(\net) \times \event(\net)$ is the minimum reflexive and
transitive relation satisfying:}

\emph{{\bf Sequentiality:} Let $P'\seqOp P''$ be a subterm of \APOC{} network $\net$.
If $\ev'$ is an event in $P'$ and $\ev''$ is an event in $P''$ then $\ev'
\leqapoc \ev''$.}

\emph{{\bf Scope:} Let $\pscope{i}{l}{\role R}{P}{\Delta}{S}$ or $\psscope{i}{l}{\role 
R}{P}$ be a
subterm of \APOC{} $\net$ with global index $\xi$. If $\ev'$
is an event in $P$ then $\uparrow_{\xi} \leqapoc \ev' \leqapoc
\downarrow_{\xi}$.}

\emph{{\bf Synchronisation:} For each pair of events $\ev$ and $\ev'$, $\ev \leq
\ev'$ implies $\overline{\ev} \leqapoc \ev'$.}

\emph{{\bf If:} Let $\ifthenKey{b}{P}{P'}{i}$ be a subterm of \APOC{} network $\net$
 with global index $\xi$, let $\ev_\xi$ be the guard if-event, then for every
 event $\ev$ in $P$ and for every event $\ev'$ in $P'$ we have $\ev_\xi \leqapoc
 \ev$ and $\ev_\xi \leqapoc \ev'$.}

\emph{{\bf While:} Let $\whileKey{b}{P}{i}$ be a subterm of \APOC{} network $\net$
with global index $\xi$, let $\ev_\xi$ be the guard while-event, then for every
event $\ev$ in $P$ we have $\ev_\xi \leqapoc \ev$.
}\end{definition}

On $\APOC$ networks obtained as projections of well-annotated $\AIOC$s the
relation $\leqapoc$ is a partial order, as expected. However, since this
result is not needed in the remainder of the paper, we do not present its proof.

There is a relation between $\AIOC$ events and causality relation and their
counterparts at the $\APOC$ level. Indeed, the events and causality relation
are preserved by projection.

\begin{restatable}{lemma}{lemmaEv}\label{lemma:ev}
Given a well-annotated connected \AIOC{} process
$\mathcal{I}$ and for each state $\Sigma$ the \APOC{} network
$\proj(\mathcal{I},\Sigma)$ is such that:
\begin{enumerate}
\item $\devent(\mathcal{I}) \subseteq \event(\proj(\mathcal{I},\Sigma))$;

\item $\forall\ \ev_1,\ev_2 \in \devent(\mathcal{I}). \ev_1 \leqaioc \ev_2
\Rightarrow  \ev_1 \leqapoc \ev_2 \vee \ev_1 \leqapoc \overline{\ev_2}$
\end{enumerate}
\end{restatable}

\begin{proof}
The events of a \APOC obtained by projecting a \AIOC $\mathcal I$ are included
in the events of the \AIOC $\mathcal I$ by definition of projection. The
preservation of the causality relation can be proven by a case analysis on the
condition used to derive the dependency (i.e., sequentiality, scope,
synchronisation, if and while). Details are in Appendix~\ref{appendix:proof}.
\end{proof}
\noindent
To complete the definition of our event structure we now define a notion of
conflict between (\AIOC{} and \APOC{}) events, relating events which are in
different branches of the same conditional.

\begin{definition}[Conflicting events]\label{def:conflict}\emph{ Given a \AIOC{}
process $\mathcal{I}$, two events $\ev, \ev' \in
\devent(\mathcal{I})$ are conflicting if they belong to different branches of
the same conditional, i.e., there exists a subprocess $\ifthenKey{b@\role
R}{\mathcal{I}'}{\mathcal{I}''}{i}$ of $\mathcal{I}$ such that $\ev
\in \devent(\mathcal{I}') \wedge \ev'
\in \devent(\mathcal{I}'')$ or $\ev'
\in \devent(\mathcal{I}') \wedge \ev
\in \devent(\mathcal{I}'')$.}

\emph{Similarly, given a \APOC{} network $\net$, we say that two events $\ev, \ev'
\in \event(\net)$ are conflicting if they belong to different branches of the
same conditional, i.e., there exists a subprocess $\ifthenKey{b}{P}{P'}{i}$ of
$\net$ such that $\ev
\in \event(P) \wedge \ev'
\in \event(P')$ or $\ev'
\in \event(P) \wedge \ev
\in \event(P')$.
}\end{definition}

Similarly to what we did for \AIOC{}s, we define below well-annotated
\APOC{}s. Well-annotated \APOC{}s include all \APOC{}s obtained by projecting
well-annotated \AIOC{}s. As stated in the definition below, and proved in
Lemma~\ref{lemma:IOCwell}, well-annotated \APOC{}s enjoy various properties
useful for our proofs.

\begin{definition}[Well-annotated \APOC{}]\emph{
\label{defin:synchwa} An annotated \APOC{} network $\net$ is well annotated
for its causality relation $\leqapoc$ if the following conditions hold:
\begin{description}
\item[{\sc C1}] for each global index $\xi$ there are at most two
communication events on programmer-specified operations with global index $\xi$ and, in this
case, they are matching events;
\item[{\sc C2}] only events which are minimal according to $\leqapoc$ may
correspond to enabled transitions;
\item[{\sc C3}] for each pair of non-conflicting sending events
$[f_\xi]_{\role R_1}$ and $[f_{\xi'}]_{\role R_1}$ on the same operation
$\idxSign{i}.o^?$ with the same target $\role R_2$ such that $\xi \neq \xi'$ we
have $[f_\xi]_{\role R_1} \leqapoc [f_{\xi'}]_{\role R_1}$ or $[f_{\xi'}]_{\role
R_1} \leqapoc [f_\xi]_{\role R_1}$;
\item[{\sc C4}] for each pair of non-conflicting receiving events
$[t_\xi]_{\role R_2}$ and $[t_{\xi'}]_{\role R_2}$ on the same operation
$\idxSign{i}.o^?$ with the same sender $\role R_1$ such that $\xi \neq \xi'$ we
have $[t_\xi]_{\role R_2} \leq [t_{\xi'}]_{\role R_2}$ or $[t_{\xi'}]_{\role
R_2} \leq [t_\xi]_{\role R_2}$;
\item[{\sc C5}] if $\ev$ is an event inside a scope with global index $\xi$
then its matching events $\overline{\ev}$ (if they exist) are inside a scope
with the same global index.
\item[{\sc C6}] if two events have the same index but different global
indexes then one of them, let us call it $\ev_1$, is inside the body of a while
loop with global index $\xi_1$ and the other, $\ev_2$, is not. Furthermore,
$\ev_2 \leqapoc \ev_{\xi_1}$ where $\ev_{\xi_1}$ is the guarding while-event of
the while loop with global index $\xi_1$.
\end{description}
}\end{definition}

\noindent Since scope update, conditional, and iteration at the \AIOC{} level happen in
one step, while they correspond to many steps of the projected \APOC, we
introduce a function, denoted $\upd$, that bridges this gap. More precisely,
function $\upd$ is obtained as the composition of two functions, a function
$\Prop$ that completes the execution of \AIOC{} actions which have already
started, and a function $\ssim$ that eliminates all the auxiliary closing
communications of scopes (scope execution introduces in the \APOC auxiliary
communications which have no correspondence in the \AIOC).

\begin{definition}[$\upd$ function]\label{def:upd}\emph{ Let $\net$ be an annotated
\APOC{} (we drop indexes if not relevant). The $\upd$ function is defined
as the composition of a function $\Prop$ and a function $\ssim$. Thus,
$\upd(\net) = \ssim(\Prop(\net))$. Network $\Prop(\net)$ is obtained from
$\net$ by repeating the following operations while possible.
\begin{enumerate}[widest=10]
\item Performing the reception of the positive evaluation of the guard of a 
while loop, by replacing for every
$\idxSign{i}.\cout{\AuxWb^*_\idxSign{i}}{\trueC}{\role R'}$ enabled, all
the terms $$\cinp{\idxSign{i}.\AuxWb^*_\idxSign{i}}{x_\idxSign{i}}{\role
R}\seqOp  \while{x_\idxSign{i}}{P\seqOp
\cout{\idxSign{i}.\AuxWe^*_\idxSign{i}}{\okC}{\role R}\seqOp
\cinp{\idxSign{i}.\AuxWb^*_\idxSign{i}}{x_\idxSign{i}}{\role R}}$$
not inside another while construct, with 
$$\qquad P\seqOp
\cout{\idxSign{i}.\AuxWe^*_\idxSign{i}}{\okC}{\role R}\seqOp
\cinp{\idxSign{i}.\AuxWb^*_\idxSign{i}}{x_\idxSign{i}}{\role
R}\seqOp\while{x_\idxSign{i}}{P\seqOp
\cout{\idxSign{i}.\AuxWe^*_\idxSign{i}}{\okC}{\role R}\seqOp
\cinp{\idxSign{i}.\AuxWb^*_\idxSign{i}}{x_\idxSign{i}}{\role R}}$$
and replace $\cout{\idxSign{i}.\AuxWb^*_\idxSign{i}}{\trueC}{\role R'}$
with $\one$.
\item Performing the reception of the negative evaluation of the guard of a
while loop by replacing, for every
$\cout{\idxSign{i}.\AuxWb^*_\idxSign{i}}{\falseC}{\role R'}$ enabled, all
the terms $$\cinp{\idxSign{i}.\AuxWb^*_\idxSign{i}}{x_\idxSign{i}}{\role
R}\seqOp  \while{x_\idxSign{i}}{P\seqOp
\cout{\idxSign{i}.\AuxWe^*_\idxSign{i}}{\okC}{\role R}\seqOp
\cinp{\idxSign{i}.\AuxWb^*_\idxSign{i}}{x_\idxSign{i}}{\role R}}$$ not inside
another while construct, with $\one$, and replace
$\cout{\idxSign{i}.\AuxWb^*_\idxSign{i}}{\falseC}{\role R'}$ with $\one$.\item  Performing the unfolding of a while loop by replacing every
$$\while{x_\idxSign{i}}{P\seqOp
\cout{\idxSign{i}.\AuxWe^*_\idxSign{i}}{\okC}{\role R}\seqOp
\cinp{\idxSign{i}.\AuxWb^*_\idxSign{i}}{x_\idxSign{i}}{\role R}}$$ enabled not
inside another while construct, such that $x_\idxSign{i}$ evaluates to
\texttt{true} in the local state, with 
$$\qquad P\seqOp\cout{\idxSign{i}.\AuxWe^*_\idxSign{i}}{\okC}{\role
R}\seqOp
\cinp{\idxSign{i}.\AuxWb^*_\idxSign{i}}{x_\idxSign{i}}{\role
R}\seqOp\while{x_\idxSign{i}}{P\seqOp
\cout{\idxSign{i}.\AuxWe^*_\idxSign{i}}{\okC}{\role R}\seqOp
\cinp{\idxSign{i}.\AuxWe^*_\idxSign{i}}{x_\idxSign{i}}{\role R}}$$\item Performing the termination of while loop by replacing every
$$\while{x_\idxSign {i}}{P\seqOp
\cout{\idxSign{i}.\AuxWe^*_\idxSign{i}}{\okC}{\role R}\seqOp
\cinp{\idxSign{i}.\AuxWb^*_\idxSign{i}}{x_\idxSign{i}}{\role R}}$$ enabled not
inside another while construct, such that $x_\idxSign{i}$ evaluates to
$\falseC$ in the local state, with $\one$.
\item Performing the reception of the positive evaluation of the guard of 
a conditional by replacing, for every
$\cout{\idxSign{i}.\AuxIf^*_\idxSign{i}}{\trueC}{\role R'}$ enabled, all
the terms $$\cinp{\idxSign{i}.\AuxIf^*_\idxSign{i}}{x_\idxSign{i}}{\role
R}\seqOp
\ifthen{x_\idxSign{i}}{P'}{P''}$$ not inside a while construct with $P'$, and
replace $\cout{\idxSign{i}.\AuxIf^*_\idxSign{i}}{\trueC}{\role R'}$ with
$\one$.
\item Performing the reception of the negative evaluation of the guard of a
conditional by replacing, for every
$\cout{\idxSign{i}.\AuxIf^*_\idxSign{i}}{\falseC}{\role R'}$ enabled, all
the terms $$\cinp{\idxSign{i}.\AuxIf^*_\idxSign{i}}{x_\idxSign{i}}{\role
R}\seqOp \ifthen{x_\idxSign{i}}{P'}{P''}$$ not inside a while construct, with
$P''$, and replace
$\cout{\idxSign{i}.\AuxIf^*_\idxSign{i}}{\falseC}{\role R'}$ with $\one$.
\item Performing the selection of the ``then'' branch by replacing every
$$\ifthen{x_\idxSign{i}}{P'}{P''}$$ enabled, such that $x_\idxSign{i}$
evaluates to $\trueC$ in the local state, with $P'$.
\item Performing the selection of the ``else'' branch by replacing every
$$\ifthen{x_\idxSign{i}}{P'}{P''}$$ enabled, such that $x_\idxSign{i}$
evaluates to $\falseC$ in the local state, with $P''$.
\item Performing the communication of the updated code by replacing, for every \\
$\cout{\idxSign{i}.\AuxSb^*_\idxSign{i}}{P}{\role S}$ enabled, all the terms
$$\psscope{i}{l}{\role R}{P'}$$ in role $\role S$ not inside a while construct
with $P$, and replace $\cout{\idxSign{i}.\AuxSb^*_\idxSign{i}}{P}{\role S}$ with
$\one$.
\item Performing the communication that no update is needed by replacing, for 
each \\$\cout{\idxSign{i}.\AuxSb^*_{\idxSign i}}{\noC}{\role S}$
enabled, all the terms $$\psscope{i}{l}{\role R}{P'}$$ in role $\role S$
not inside a while construct with $P'$, and replace
$\cout{\idxSign{i}.\AuxSb^*_{\idxSign{i}}}{P}{\role S}$ with $\one$.
\end{enumerate}\mbox{}\\
\noindent
Network $\ssim(\net)$ is obtained from $\net$ by repeating the following
operations while possible:
\begin{itemize}\item Removing the auxiliary communications for end of scope and end of
while loop synchronisation by replacing each
$$\begin{array}{ll}
\cout{\idxSign{i}.\AuxSe^*_\idxSign{i}}{\okC}{\role R} 
\qquad & \qquad
\cout{\idxClose{i} \idxSign{i}.\AuxWe^*_\idxSign{i}}{\okC}{\role R}
\\
\cinp{\idxSign {i}.\AuxSe^*_\idxSign{i}}{\_}{\role R}
\qquad & \qquad
\cinp{\idxClose{i} \idxSign{i}.\AuxWe^*_\idxSign{i}}{\_}{\role R}	
\end{array}$$
not inside a while
construct with $\one$.
\end{itemize} Furthermore $\ssim$ may apply 0 or more times the following
operation:
\begin{itemize}
\item \label{pruning_def:one}replace a subterm $\one \seqOp P$ by $P$ or a
subterm $\one \mid P$ by $P$.
\end{itemize}
}\end{definition}

Note that function $\Prop$ does not reduce terms inside a while construct.
Assume, for instance, to have an auxiliary send targeting a receive inside the
body of a while loop. These two communications should not interact since they
have different global indexes. This explains why we exclude terms inside the
body of while loops.

We proceed now to prove some of the proprieties of \AIOC and \APOC. The first
result states that in a well-annotated \APOC{} only transitions corresponding
to events minimal with respect to the causality relation $\leqapoc$ may be
enabled.

\begin{lemma}\label{lemma:minimalevent} If $\net$ is a \APOC{}, $\leqapoc$
its causality relation and $\ev$ is an event corresponding to a transition
enabled in $\net$ then $\ev$ is minimal with respect to  $\leqapoc$.
\end{lemma}
\begin{proof} The proof is by contradiction. Suppose $\ev$ is enabled but not
minimal, i.e., there exists $\ev'$ such that $\ev' \leqapoc \ev$. If there is
more than one such $\ev'$ consider the one such that the length of the
derivation of $\ev' \leqapoc \ev$ is minimal. This derivation should have
length one, and following  Definition~\ref{def:causalapoc} it may result from
one of the following cases:
\begin{itemize}
\item Sequentiality: $\ev' \leqapoc \ev$ means that $\ev' \in \event(P')$,
$\ev \in \event(P'')$, and $P' \seqOp P''$ is a subterm of $\net$. Because of
the semantics of sequential composition $\ev$ cannot be enabled.
\item Scope: let $\pscope{i}{l}{\role R}{P}{\Delta}{S}$ or
$\psscope{i}{l}{\role R}{P}$ be a subprocess of $\net$ with global index
$\xi$. We have the following cases:
\begin{itemize}
\item $\ev' = \uparrow_{\xi}$ and $\ev \in \event(P)$, and this implies that
$\ev$ cannot be enabled since if $\ev'$ is enabled then the Rule
$\did{\APOC}{\AdaptRule}$ or Rule $\did{\APOC}{\NoAdaptRule}$ for starting the
execution of the scope have not been applied yet;
\item $\ev' = \uparrow_{\xi}$ and $\ev =
\downarrow_{\xi}$, or $\ev' \in \event(P)$ and $\ev =
\downarrow_{\xi}$: this is trivial, since  $\downarrow_{\xi}$ is an auxiliary
event and no transition corresponds to it;
\end{itemize}
\item If: $\ev' \leqapoc \ev$ means that $\ev'$ is the evaluation of the guard
of a subterm \\$\ifthenKey{x_\idxSign{i}}{P'}{P''}{i}$ and $\ev \in \event(P')
\cup \event(P'')$. Event $\ev$ cannot be enabled because of the semantics of
conditionals.
\item While: $\ev' \leqapoc \ev$ means that $\ev'$ is the evaluation of the
guard of a subterm \\$\whileKey{x_\idxSign{i}}{P}{i}$ and $\ev \in \event(P)$.
Event $\ev$ cannot be enabled because of the semantics of the while loop.\qedhere
\end{itemize}
\end{proof}

\noindent We now prove that all the \APOC{}s obtained as projection of well-annotated
connected \AIOC{}s are well-annotated.

\begin{restatable}{lemma}{lemmaIOCwell}\label{lemma:IOCwell}
Let ${\mathcal I}$ be a well-annotated connected \AIOC{} process and $\Sigma$ 
a state. Then the projection $\net = \proj({\mathcal I},\Sigma)$ is a 
well-annotated \APOC{} network with respect to $\leqapoc$.
\end{restatable}

\begin{proof}
We have to prove that $\proj(\mathcal{I},\Sigma)$ satisfies the conditions of
Definition~\ref{defin:synchwa}. We have a case for each condition. Details in
Appendix~\ref{appendix:proof}.
\end{proof}

The next lemma shows that for every starting set of updates $\rules$ the
\APOC{} $\net$ and $\upd(\net)$ have the same set of weak traces.

\begin{lemma}\label{lemma:uptoupd} Let $\net$ be a \APOC{}. The following
properties hold:
\begin{enumerate}
\item\label{cond:upd} if $\tuple{\ambientN,\upd(\net)} \arro{\eta}
\tuple{\ambientN',\net'}$ with $\eta \in \{\commLabel{o^?}{\role R_1}{v}{\role
R_2}{x},
\commLabel{o^*}{\role R_1}{X}{\role R_2}{}, \rules', \tick, {\mathcal
I},\\\NoAdaptLabel, \tau \}$ then there exists $\net''$
such that $\tuple{\ambientN,\net}\arro{\eta_1}
\dots \arro{\eta_k} \arro{\eta}
\tuple{\ambientN',\net''}$ where $\eta_{i} \in \{\commLabel{o^*}{\role
R_1}{v}{\role R_2}{x},\commLabel{o^*}{\role R_1}{X}{\role R_2}{},\tau\}$ for
each $i \in \{1,\ldots,k\}$ and $\upd(\net'') = \upd(\net')$.
\item\label{cond:no_upd} if $\tuple{\ambientN,\net} \arro{\eta}
\tuple{\ambientN',\net'}$ for $\eta \in \{\commLabel{o^?}{\role R_1}{v}{\role
R_2}{x},\commLabel{o^*}{\role R_1}{X}{\role R_2}{}, \rules', \tick, {\mathcal
I},\NoAdaptLabel,
\tau\}$, then one of the following holds:
\begin{enumerate}
	\item\label{cond:case_private} $\upd(\net) = \upd(\net')$ and $\eta \in
	\{\commLabel{o^*}{\role R_1}{v}{\role R_2}{x},\commLabel{o^*}{\role
	R_1}{X}{\role R_2}{},\tau\}$, or
	\item\label{cond:case_one_step} $\tuple{\ambientN,\upd(\net)} \arro{\eta}
	\tuple{\ambientN',\net''}$ such that $\upd(\net') = \upd(\net'')$.
\end{enumerate}
\end{enumerate}
\end{lemma}
\newpage

\begin{proof}\hfill
\begin{enumerate}

\item Applying the $\upd$ function corresponds to perform weak transitions,
namely transitions with labels in $\{\commLabel{o^*}{\role R_1}{v}{\role
R_2}{x},\commLabel{o^*}{\role R_1}{X}{\role R_2}{},\tau\}$. Some of such
transitions may not be enabled yet.
Hence, $\net$ may perform the subset of the weak transitions above which are
or become enabled, reducing to some $\net'''$. Then, $\eta$ is enabled also in
$\net'''$ and we have $\tuple{\ambientN,\net'''} \arro{\eta}
\tuple{\ambientN',\net''}$. At this point we have that $\net''$ and $\net'$
may differ only for the weak transitions that were never enabled, which can be
executed by $\upd$.

\item There are two cases. In the first case the transition with label $\eta$ is
one of the transitions executed by function $\upd$. In this case the
condition~\ref{cond:case_private} holds. In the second case, the transition with
label $\eta$ is not one of the transitions executed by function $\upd$. In this
case the transition with label $\eta$ is still enabled in $\upd(\net)$ and can
be executed. This leads to a network that differs from $\net'$ only because of
transitions executed by the $\upd$ function and case~\ref{cond:case_one_step}
holds.\qedhere
\end{enumerate}
\end{proof}

\noindent We now prove a property of transitions with label $\tick$.

\begin{lemma}\label{lemma:tick} For each \AIOC{} system
$\tuple{\state,\rules,{\mathcal I}}$ that reduces with a transition labelled
$\tick$ then, for each role $\role R \in
\rolesFunc({\mathcal I})$, the \APOC{} role $\roleExec{\pi({\mathcal I},\role
R),\state_{\role R}}{\role R}$ can reduce with a transition labelled $\tick$ and
vice versa.
\end{lemma}
\begin{proof} Note that a \AIOC{} can perform a transition with label $\tick$
only if it is a term obtained using sequential and/or parallel composition
starting from $\one$ constructs. The projection has the same shape, hence it can
perform the desired transition. The other direction is similar.
\end{proof}

We can now prove our main theorem (Theorem~\ref{teo:final}, restated below)
for which, given a connected well-annotated \AIOC{} process ${\mathcal I}$ and
a state $\Sigma$, the \APOC{} network obtained as its projection has the same
behaviours of ${\mathcal I}$.

\final*

\begin{proof} We prove that the relation $\mathcal R$ below is a weak system
bisimulation. 

$$
\mathcal R=\sset{
	(\tuple{\state,\rules,{\mathcal I}},\tuple{\rules,\net})
}{
	\upd(\net) = \proj({\mathcal I},\Sigma),\\
	\devent(\mathcal{I}) \subseteq \event(\Prop(\net)),\\
	\forall \ \ev_1, \ev_2 \in \devent(\mathcal{I}) \ . \\
	\quad \ev_1 \leqaioc \ev_2 \Rightarrow \ev_1 \leqapoc \ev_2 \vee \ev_1
\leqapoc \overline{\ev_2} 
}
$$
\vspace{1em}

where ${\mathcal I}$ is obtained from a well-annotated connected \AIOC{} via
$0$ or more transitions and $\upd(\net)$ is a well-annotated \APOC{}.

To ensure that proving that the relation above is a weak system bisimulation
implies our thesis, let us show that the pair $(\tuple{\state,\rules,{\mathcal
I}},\tuple{\rules,\proj(\mathcal{I},\Sigma)})$ from the theorem statement
belongs to $\mathcal R$. Note that here $\mathcal{I}$ is well-annotated and
connected, and for each such $\mathcal{I}$ we have
$\upd(\proj(\mathcal{I},\Sigma)) = \proj(\mathcal{I},\Sigma)$. From
Lemma~\ref{lemma:IOCwell} $\proj({\mathcal I},\Sigma)$ is well-annotated, thus
$\upd(\proj(\mathcal{I},\Sigma))$ is well-annotated.
Observe that $\Prop$ is the identity on $\proj({\mathcal I},\Sigma)$, thus from
Lemma~\ref{lemma:ev} we have that the conditions $\devent(\mathcal{I})
\subseteq \event(\Prop(\net))$ and $\forall \ev_1, \ev_2 \in
\devent(\mathcal{I}) \ . \ \ev_1 \leqaioc \ev_2 \Rightarrow \ev_1 \leqapoc
\ev_2 \vee \ev_1 \leqapoc \overline{\ev_2}$ are satisfied. 

We now prove that $\mathcal R$ is a weak system bisimulation. To prove it, we
show below that it is enough to consider only the case in which $\net$ (and not
$\upd(\net)$) is equal to $\proj({\mathcal I},\Sigma)$. Furthermore, in this
case the transition of $\tuple{\state,\rules,{\mathcal I}}$ is matched by the
first transition of $\tuple{\rules,\proj(\mathcal{I},\Sigma)}$.

Formally, for each $(\tuple{\state,\rules,{\mathcal I}},\tuple{\rules,\net})$
where $\net=\proj(\mathcal I,\Sigma)$ we have to prove the following simplified
bisimulation clauses.
\begin{itemize}
\item if $\tuple{\state,\rules,{\mathcal I}} \arro{\mu}
\tuple{\state'',\rules'',{\mathcal I}''}$ then $\tuple{\rules,\net} \arro{\mu}
\tuple{\rules'',\net'''}$ with $(\tuple{\state'',\rules'',{\mathcal
I}''},\tuple{\rules'',\net'''}) \in \mathcal R$;
\item if $\tuple{\rules,\net} \arro{\eta} \tuple{\rules'',\net'''}$ with $\eta
\in \{ \commLabel{o}{\role R_1}{v}{\role R_2}{x}\seqOp
\tick\seqOp {\mathcal I} \seqOp \NoAdaptLabel \seqOp \rules''\seqOp
\tau \}$ then\\ $\tuple{\state,\rules,{\mathcal I}}
\arro{\eta} \tuple{\state'',\rules'',{\mathcal I}''}$ and
$(\tuple{\state'',\rules'',{\mathcal I}''},\tuple{\rules'',\net'''}) \in
\mathcal R$.
\end{itemize}

\noindent In fact, consider a general network $\net_g$ with $\upd(\net_g)=\proj({\mathcal
I},\Sigma)$. If $\tuple{\state,\rules,{\mathcal I}} \arro{\mu}
\tuple{\state'',\rules'',{\mathcal I}''}$, then by hypothesis
$\tuple{\rules,\upd(\net_g)} \arro{\mu} \tuple{\rules'',\net'''}$. From
Lemma~\ref{lemma:uptoupd} case~\ref{cond:upd} there exists $\net''$
such that $\tuple{\rules,\net_g}\arro{\eta_1}
\dots \arro{\eta_k} \arro{\mu}
\tuple{\rules'',\net''}$ where $\eta_{i} \in \{\commLabel{o^*}{\role
R_1}{v}{\role R_2}{x},\commLabel{o^*}{\role R_1}{X}{\role R_2}{},\tau\}$ for
each $i \in \{1,\ldots,k\}$ and $\upd(\net'') = \upd(\net''')$. By hypothesis
$(\tuple{\state'',\rules'',{\mathcal I}''},\tuple{\rules'',\net'''}) \in
\mathcal R$, hence, by definition of $\mathcal R$, $\upd(\net''') =
\proj({\mathcal I''},\Sigma'')$, and therefore also $\upd(\net'') =
\proj({\mathcal I''},\Sigma'')$.
The conditions on events hold by hypothesis since function $\upd$ has no effect
on \APOC{} events corresponding to \AIOC{} events. Furthermore, only enabled
interactions have been executed, hence dependencies between \APOC{} events
corresponding to \AIOC{} events are untouched.

If instead $\tuple{\rules,\net_g} \arro{\eta} \tuple{\rules'',\net'''}$ with
$\eta \in \{ \commLabel{o^?}{\role R_1}{v}{\role R_2}{x},\commLabel{o^*}{\role
R_1}{X}{\role R_2}{},\tick, {\mathcal I},\\\NoAdaptLabel, \rules'',
\tau \}$ then thanks to Lemma~\ref{lemma:uptoupd} we have one of the following:
(\ref{cond:case_private}) $\upd(\net_g) = \upd(\net''')$ and $\eta
\in \{\commLabel{o^*}{\role R_1}{v}{\role R_2}{x},\commLabel{o^*}{\role
R_1}{X}{\role R_2}{},\tau\}$, or
(\ref{cond:case_one_step})
$\tuple{\rules,\upd(\net_g)} \arro{\eta} \tuple{\rules'',\net''}$
such that $\upd(\net''') = \upd(\net'')$.
In case (\ref{cond:case_one_step}) we have $\tuple{\rules,\upd(\net_g)}
\arro{\eta} \tuple{\rules'',\net''}$. Then, by hypothesis, we have
$\tuple{\state,\rules,{\mathcal I}} \arro{\eta}
\tuple{\state'',\rules'',{\mathcal I}''}$ and
$(\tuple{\state'',\rules'',{\mathcal I}''}, $ $\tuple{\rules'',\net''}) \in
\mathcal R$. To deduce that $(\tuple{\state'',\rules'',{\mathcal I}''}, $
$\tuple{\rules'',\net'''}) \in \mathcal R$, one can proceed using the same
strategy as the case of the challenge from the \AIOC{} above. In case
(\ref{cond:case_private}) the step is matched by the \AIOC{} by staying idle,
following the second option in the definition of weak system bisimulation. The
proof is similar to the one above.\\

Thus, we have to prove the two simplified bisimulation clauses above. The proof
is by structural induction on the \AIOC{} ${\mathcal I}$. All the subterms of a
well-annotated connected \AIOC{} are well-annotated and connected, thus the
induction can be performed. We consider both challenges from the \AIOC{}
($\rightarrow$) and from the \APOC{} ($\leftarrow$). The case for label $\tick$
follows from Lemma~\ref{lemma:tick}. The case for labels $\rules$ is
trivial. Let us consider the other labels, namely $\commLabel{o}{\role
R_1}{v}{\role R_2}{x}, {\mathcal I},\NoAdaptLabel$, and $\tau$.

Note that no transition (at the \AIOC{} or at the \APOC{} level) with one of
these labels can change the set of updates $\rules$. Thus, in the following,
we will not write it. Essentially, we will use \AIOC{} processes and \APOC{}
networks instead of \AIOC{} systems and \APOC{} systems, respectively. Note
that \APOC{} networks also include the state, while this is not the case for
\AIOC{} processes. For \AIOC{} processes, we assume to associate to them the
state $\state$, and comment on its changes whenever needed.

\begin{description}
\item[Case $\one$, $\zero$] trivial.
\item[Case $\idx{i} \assign{x}{\role R}{e}$] the assignment changes the global
  state in the \AIOC{}, and its projection on the role $\role R$ changes the
  local state of the role in the \APOC{} in a corresponding way.
\item[Case $\idx{i} \comm{o}{\role R_1}{e}{\role R_2}{x}$] trivial. Just note
that at the \APOC{} level the interaction gives rise to one send and one
receive with the same operation and prefixed by the same index.
Synchronisation between send and receive is performed by Rule
$\did{\APOC}{Synch}$ that also removes the index from the label.
\item[Case ${\mathcal I}\seqOp{\mathcal I}'$] from the definition of the
projection function we have that $\net = \parallel_{\role R \in
\rolesFunc(\mathcal{I}\seqOp\mathcal{I}')} (\pi({\mathcal I},\role R)\seqOp
\pi({\mathcal I}',\role R),\Sigma_{\role R})_{\role R}$.

\begin{description}
\item[$\rightarrow$] Assume that ${\mathcal I}\seqOp{\mathcal I}' \arro{\mu}
{\mathcal I}''$ with $\mu \in \{\commLabel{o}{\role R_1}{v}{\role R_2}{x}\seqOp
\mathcal{I}\seqOp \NoAdaptLabel,\tau \}$.  There are two
possibilities: either (\emph{i}) ${\mathcal I} \arro{\mu} {\mathcal I}'''$ and
${\mathcal I}'' = {\mathcal I}'''\seqOp{\mathcal I}'$ or (\emph{ii}) ${\mathcal
I}$ has a transition with label $\tick$ and ${\mathcal I}' \arro{\mu} {\mathcal
I}''$. 

In case (\emph{i}) by inductive hypothesis $$\qquad\qquad\quad \parallel_{\role
R \in \rolesFunc(\mathcal{I})} \roleExec{\pi({\mathcal I},\role R),\Sigma_{\role
R}}{\role R} \arro{\mu} \net''' \mbox{ and } \upd(\net''') =
\parallel_{\role R \in \rolesFunc(\mathcal{I})} \roleExec{\pi({\mathcal
I}''',\role R),\Sigma_{\role R}'}{\role R}$$

Thus $$\parallel_{\role R \in \rolesFunc(\mathcal{I})}
\roleExec{\pi({\mathcal I},\role R)\seqOp\pi({\mathcal I}',\role
R),\Sigma_{\role R}}{\role R} \arro{\mu} \net \quad \mbox{and}$$
$$\upd(\net)=\parallel_{\role R \in \rolesFunc(\mathcal{I})}
\roleExec{\pi({\mathcal I}''',\role R)\seqOp\pi({\mathcal I}',\role
R),\Sigma_{\role R}'}{\role R}$$
If $\rolesFunc(\mathcal{I}') \subseteq \rolesFunc(\mathcal{I})$ then the thesis
follows. 
If $\rolesFunc(\mathcal{I}') \not \subseteq \rolesFunc(\mathcal{I})$ then at
the \APOC{} level the processes in the roles in $\rolesFunc(\mathcal{I}')
\setminus \rolesFunc(\mathcal{I})$ are not affected by the transition. Note
however that the projection of $\mathcal{I}$ on these roles is a term composed
only by $\one$s, and the ones corresponding to parts of $\mathcal I$ that have
been consumed can be removed by the $\ssim$ part of function $\upd$.

In case (\emph{ii}), ${\mathcal I}$ has a transition with label $\tick$ and
${\mathcal I}' \arro{\mu} {\mathcal I}''$. By inductive hypothesis
$\proj({\mathcal I}',\Sigma) \arro{\mu} \net''$ and
$\upd(\net'')=\proj({\mathcal I}'',\Sigma')$. The thesis follows since, thanks
to Lemma~\ref{lemma:tick}, $\proj({\mathcal I}\seqOp{\mathcal I}',\Sigma)
\arro{\mu} \net$ and $\upd(\net) = \proj({\mathcal I}'',\Sigma')$, possibly
using the $\ssim$ part of function $\upd$ to remove the $\one$s which are no
more needed.
 
Note that, in both the cases, conditions on events follow by inductive
hypothesis.

\item[$\leftarrow$] Assume that
$$
\net = \parallel_{\role R \in \rolesFunc(\mathcal{I}\seqOp\mathcal{I}')}
\roleExec{\pi({\mathcal I},\role R)\seqOp\pi({\mathcal I}',\role
R),\Sigma_{\role R}}{\role R} \arro{\eta} \parallel_{\role R \in
\rolesFunc(\mathcal{I}\seqOp\mathcal{I}')} \roleExec{P_{\role R},\Sigma_{\role
R}'}{\role R}
$$
with $\eta \in  \{\commLabel{o}{\role R_1}{v}{\role R_2}{x}, {\mathcal
I},\NoAdaptLabel, \tau\}$. We have a case analysis on $\eta$.

If $\eta=\commLabel{o}{\role R_1}{v}{\role R_2}{x}$ then
$$\roleExec{\pi({\mathcal I}\seqOp{\mathcal I'},\role R_1),
\Sigma_{\role R_1}}{\role R_1} \arro{\coutLabel{\idxSign{i}.o}{v}{\role
R_2}:\role R_1} \roleExec{P_{\role R_1}, \Sigma_{\role R_1}}{\role R_1} \mbox{
and }$$
$$ \roleExec{\pi({\mathcal I}\seqOp{\mathcal I'},\role R_2),
\Sigma_{\role R_2}}{\role R_2} \arro{\cinpLabel{\idxSign{i}.o}{x}{v}{\role
R_1}:\role R_2} \roleExec{P_{\role R_2},\Sigma_{\role R_2}}{\role R_2}$$
The two events have the same global index since they have the same index
$\idxSign{i}$ (otherwise they could not synchronise) and they are both outside
of any while loop (since they are enabled), hence the global index coincides
with the index. Thus, they are either both from ${\mathcal I}$ or both from
${\mathcal I}'$.

In the first case we have also
$$
\parallel_{\role R \in
\rolesFunc(\mathcal{I}\seqOp\mathcal{I}')}\roleExec{\pi({\mathcal
I},\role R),\Sigma_{\role R}}{\role R}\arro{\commLabel{o}{\role
R_1}{v}{\role R_2}{x}} \parallel_{\role R \in
\rolesFunc(\mathcal{I}\seqOp\mathcal{I}')}\roleExec{P''_{\role R},\Sigma_{\role
R}}{\role R} $$ with $P_{\role R}=P''_{\role R}\seqOp\pi({\mathcal I}',\role
R)$.
Thus, by inductive hypothesis, ${\mathcal I}
\arro{\commLabel{o}{\role R_1}{v}{\role R_2}{x}} {\mathcal I}''$ and
$\upd(\parallel_{\role R \in \rolesFunc{\mathcal{I}\seqOp\mathcal{I}'}}
\roleExec{P''_{\role R},\Sigma_{\role R}}{\role R}) = \proj(\mathcal
I'',\Sigma)$. Hence, we have that $${\mathcal I}\seqOp{\mathcal I'}
\arro{\commLabel{o}{\role R_1}{v}{\role R_2}{x}} {\mathcal I}''\seqOp{\mathcal
I}'$$ and $$\upd(\parallel_{\role R \in \rolesFunc{\mathcal{I}\seqOp\mathcal{I}'}}
\roleExec{P''_{\role R};\pi({\mathcal I}',\role
R),\Sigma_{\role R}}{\role R}) = \proj(\mathcal
I'';\mathcal I',\Sigma)$$ The thesis follows.

In the second case, we need to show that the interaction $\commLabel{o}{\role
R_1}{v}{\role R_2}{x}$ is enabled. Assume that this is not the case. This means
that there is a \AIOC{} event $\ev$ corresponding to some construct in $\mathcal
I$. Because of the definition of $\mathcal R$, $\ev$ is also a \APOC{} event and
$\ev \leqapoc \xi : \co{o}{\role R_2} \lor \ev \leqapoc \xi: \ci{o}{\role R_1}$.
Hence, at least one of the two events is not minimal and the corresponding
transition cannot be enabled, against our hypothesis. Therefore the interaction
$\commLabel{o}{\role R_1}{v}{\role R_2}{x}$ is enabled. Thus, ${\mathcal I}$ has
a transition with label $\tick$ and ${\mathcal I'} \arro{ \commLabel{o}{\role
R_1}{v}{\role R_2}{x}} {\mathcal I}''$.
Thanks to Lemma~\ref{lemma:tick} then both $\roleExec{\pi({\mathcal I},\role
R_1),\Sigma_{\role R_1}}{\role R_1} $ and $ \roleExec{\pi({\mathcal I},\role
R_2), \Sigma_{\role R_2}}{\role R_2}$ have a transition with label ${\tick}$.
Thus, we have 
$$\roleExec{\pi({\mathcal I}',\role R_1),\Sigma_{\role R_1}}{\role
R_1} \arro{\coutLabel{\idxSign{i}.o}{v}{\role R_2}:\role R_1}
\roleExec{P_{\role R_1},\Sigma_{\role R_1}}{\role R_1}$$
$$\roleExec{\pi({\mathcal I}',\role R_2),\Sigma_{\role R_2}}{\role R_2}
\arro{\cinpLabel{\idxSign{i}.o}{x}{v}{\role R_1}:\role R_2} \roleExec{P_{\role
R_2},\Sigma_{\role R_2}}{\role R_2} \mbox{ and thus}$$
$$\proj({\mathcal I}',\Sigma) \arro{ \commLabel{o}{\role R_1}{v}{\role
R_2}{x}} \parallel_{\role R \in \rolesFunc(\mathcal{I}')} \roleExec{P_{\role
R},\Sigma_{\role R}}{\role R}$$ 
The thesis follows by inductive hypothesis.

For the other cases of $\eta$, all the roles but one are unchanged. The proof of
these cases is similar to the one for interaction, but simpler.

Note that in all the above cases, conditions on events follow by inductive
hypothesis.
\end{description}

\item[Case ${\mathcal I} \parOpI {\mathcal I}'$] from the definition of the
 projection function we have $$\net = \parallel_{\role R \in
 \rolesFunc(\mathcal{I}\seqOp\mathcal{I}')} (\pi({\mathcal I},\role R) \mid
 \pi({\mathcal I}',\role R),\Sigma_{\role R})_{\role R}$$
\begin{description}
\item[$\rightarrow$] 
We have a case analysis on the rule used to derive the transition. If the
transition is derived using Rule $\did{\AIOC}{Parallel}$ and ${\mathcal I}
\parOpI {\mathcal I}'$ can perform a transition with label $\mu$ then one of its
two components can perform a transition with the same label $\mu$ and the thesis
follows by inductive hypothesis. Additional roles not occurring in the term
performing the transition are dealt with by the $\ssim$ part of function
$\upd$. If instead the transition is derived using Rule $\did{\AIOC}{Par-end}$
then the thesis follows from Lemma~\ref{lemma:tick}.

\item[$\leftarrow$] We have a case analysis on the label $\eta$ of the
transition.  If $\eta=\commLabel{o^?}{\role R_1}{v}{\role R_2}{x}$ then a send
and a receive on the same operation are enabled. 
The two events have the same global index since they have the same index
$\idxSign{i}$ (otherwise they could not synchronise) and they are both outside
of any while loop (since they are enabled), hence the global index coincides
with the index. Thus, they are either both from ${\mathcal I}$ or both from
${\mathcal I}'$.
The thesis follows by inductive hypothesis. For the other cases of $\eta$, only
the process of one role changes. The thesis follows by inductive hypothesis. In
all the cases, roles not occurring in the term performing the transition are
dealt with by function $\upd$.
\end{description}

\item[Case $\ifthenKey{b \at \role R}{\mathcal{I}}{\mathcal{I}'}{i}$] 
from the definition of projection
\begin{multline*}
\hspace{10pt} \net = \left(
\parallel_{\role S \in 
\rolesFunc(\mathcal{I},\mathcal{I}') \smallsetminus \{\role R\}
} \Big(\cinp{\idxIfRecv{i}
\idxSign{i}.\AuxIf^*_\idxSign{i}}{x_\idxSign{i}}{\role R}\seqOp
\idx{i}
\ifthen{x_\idxSign{i}}{\pi(\mathcal{I},\role S)}{\pi(\mathcal{I}',\role
S)},\Sigma_{\role S}\Big)_{\role S} \right)
\parallel \\ \Bigg(\idx{i} \mathtt{if} \; b \; \left\{
\left(\prod_{\role R' \in \rolesFunc(\mathcal{I},\mathcal{I}')
\smallsetminus \{\role R\}} \idxTrue{i}
\cout{\idxSign{i}.\AuxIf^*_\idxSign{i}}{\trueC}{\role R'}\right)
 \seqOp \pi(\mathcal{I},\role R)\right\}\\ \mathtt{else} \;
 \left\{\left(\prod_{\role R' \in \rolesFunc(\mathcal{I},\mathcal{I}')
 \smallsetminus \{\role R\}}
 \idxFalse{i} \cout{\idxSign{i}.\AuxIf^*_\idxSign{i}}{\falseC}{\role
 R'}\right)\seqOp
\pi(\mathcal{I}',\role R)\right\}, \Sigma_{\role R}\Bigg)_{\role R}
\end{multline*} Let us consider the case when the guard is true (the other
one is analogous).
\begin{description}
\item[$\rightarrow$] The only possible transition from the \AIOC{} is
$\ifthenKey{b \at \role R}{\mathcal{I}}{\mathcal{I}'}{i} \arro{\tau}
\mathcal{I}$. The \APOC{} can match this transition by reducing to
\begin{multline*}
\hspace{50pt}\net' = \left(
\parallel_{\role S \in \rolesFunc(\mathcal{I},\mathcal{I}') \smallsetminus \{\role
R\}} \left(
\begin{array}{l}
\idxIfRecv{i}\cinp{\idxSign{i}.\AuxIf^*_\idxSign{i}}{x_\idxSign{i}}{\role
R}\seqOp\\
\idx{i} \ifthen{x_\idxSign{i}}{\pi(\mathcal{I},\role S)}{\pi(\mathcal{I}',\role
S)}
\end{array}
,\Sigma_{\role S}\right)_{\role S}\right) \parallel\\
\left(\left(\prod_{\role R' \in \rolesFunc(\mathcal{I},\mathcal{I}')
\smallsetminus \{\role R\}} \idxTrue{i}
\cout{\idxSign{i}.\AuxIf^*_\idxSign{i}}{\trueC}{\role R'} \right) \seqOp
\pi(\mathcal{I},\role R),\Sigma_{\role R}\right)_{\role R}
\end{multline*} By applying function $\upd$ we get
$$
\upd(\net') = \left(\parallel_{\role S \in
\rolesFunc(\mathcal{I},\mathcal{I}')\smallsetminus \{\role R\}}
\roleExec{\pi(\mathcal{I},\role S),\Sigma_{\role S}}{\role S} \right)\parallel
\roleExec{\pi(\mathcal{I},\role R),\Sigma_{\role R}}{\role R}
$$ Concerning events, at the \AIOC{} level events corresponding
to the guard and to the discarded branch are removed. The same holds at the
\APOC{} level, thus conditions on the remaining events are inherited. This
concludes the proof.
\item[$\leftarrow$] The only possible transition from the \APOC{} is the
evaluation of the guard from the coordinator. This reduces $\net$ to $\net'$
above and the thesis follows from the same reasoning.
\end{description}

\item[Case $\whileKey{b \at \role R}{\mathcal{I}}{i}$] 
from the definition of projection 
\begin{multline*}
\hspace{40pt}\net = 
\left(\parallel_{\role S \in
\rolesFunc(\mathcal{I})\smallsetminus \{\role R\}} \left(
\begin{array}l
\cinp{\idxIfRecv{i}
\idxSign{i}.\AuxWb^*_\idxSign{i}}{x_\idxSign{i}}{\role R}\seqOp
\idx{i} \while{x_\idxSign{i}}{\pi(\mathcal{I},\role S)\seqOp\\
\cout{\idxClose{i} \idxSign{i}.\AuxWe^*_\idxSign{i}}{\okC}{\role
R}\seqOp\cinp{\idxIfRecv{i}
\idxSign{i}.\AuxWb^*_\idxSign{i}}{x_\idxSign{i}}{\role R}}
\end{array}
,\Sigma_{\role
S}\right)_{\role S} 
\right)
\parallel\\ 
\left(
\begin{array}{l}
\idx{i} \mathtt{while} \; b \;
\left\{
\begin{array}l
\left(\prod\limits_{\role R' \in \rolesFunc(\mathcal{I}) \smallsetminus
\{\role R\}}
\cout{\idxTrue{i} \idxSign{i}.\AuxWb^*_\idxSign{i}}{\trueC}{\role
R'}\right)\seqOp\pi(\mathcal{I},\role R)\seqOp\\[1.5em]
\prod\limits_{\role R' \in \rolesFunc(\mathcal{I}) \smallsetminus \{\role R\}}
\cinp{\idxClose{i} \idxSign{i}.\AuxWe^*_\idxSign{i}}{\_}{\role R'}
\end{array}
\right\} \seqOp
\\[3em]\prod\limits_{\role R'\in \rolesFunc(\mathcal{I}) \smallsetminus
\{\role R\}}
\cout{\idxFalse{i}
\idxSign{i}.\AuxWb^*_\idxSign{i}}{\falseC}{\role R'}, \Sigma_{\role R}
\end{array}
\right)_{\role R}
\end{multline*}
\begin{description}
\item[$\rightarrow$] Let us consider the case when the guard is true. The
only possible transition from the \AIOC{} is $\whileKey{b \at
\role R}{\mathcal{I}}{i} \arro{\tau} \mathcal{I}\seqOp\whileKey{b \at \role
R}{\mathcal{I}}{i}$. The \APOC{} can match this transition by reducing to
\begin{multline*}
\hspace{50pt}\net ' = \parallel_{\role S \in
\rolesFunc(\mathcal{I})\smallsetminus \{\role R\}} \left(
\begin{array}{l}
\cinp{ \idxIfRecv{i}
\idxSign{i}.\AuxWb^*_\idxSign{i}}{x_\idxSign{i}}{\role R}\seqOp
\idx{i} \while{x_\idxSign{i}}{\pi(\mathcal{I},\role S)\seqOp\\
\cout{\idxClose{i} \idxSign{i}.\AuxWe^*_\idxSign{i}}{\okC}{\role R}\seqOp
\cinp{\idxIfRecv{i} \idxSign{i}.\AuxWb^*_\idxSign{i}}{x_\idxSign{i}}{\role
R}}
\end{array}
,\Sigma_{\role S}\right)_{\role S}
\parallel\\
\left(
\begin{array}{l}	
\left(\prod\limits_{\role R' \in \rolesFunc(\mathcal{I}) \smallsetminus \{\role R\}}
\cout{\idxTrue{i} \idxSign{i}.\AuxWb^*_\idxSign{i}}{\trueC}{\role
R'}\right)\seqOp\pi(\mathcal{I},\role R)\seqOp\\
 \left(\prod\limits_{\role R' \in \rolesFunc(\mathcal{I}) \smallsetminus \{\role R\}}
 \cinp{\idxClose{i} \idxSign{i}.\AuxWe^*_\idxSign{i}}{\_}{\role R'}\right)\seqOp\\
 \idx{i} \mathtt{while} \; b \; \left\{
\begin{array}{l}	
 \left(\prod\limits_{\role R' \in
 \rolesFunc(\mathcal{I})
 \smallsetminus \{\role R\}} \cout{\idxTrue{i}
 \idxSign{i}.\AuxWb^*_\idxSign{i}}{\trueC}{\role
 R'}\right)\seqOp\pi(\mathcal {I},\role R)\seqOp\\
 \prod\limits_{\role R' \in \rolesFunc(\mathcal{I}) \smallsetminus
 \{\role R\}}
 \cinp{\idxClose{i} \idxSign{i}.\AuxWe^*_\idxSign{i}}{\_}{\role R'}
\end{array}
 \right\}\seqOp\\[3em]
 \prod\limits_{\role R'\in \rolesFunc(\mathcal{I}) \smallsetminus \{\role R\}}
 \cout{\idxFalse{i} \idxSign{i}.\AuxWb^*_\idxSign{i}}{\falseC}{\role R'},
 \Sigma_{\role R}
\end{array}
\right)_{\role R}
 \end{multline*}
 By applying function $\upd$ we get
\begin{multline*}
\hspace{50pt}\upd(\net') = \left(
\begin{array}{l}
\parallel_{\role S \in
\rolesFunc(\mathcal{I})\smallsetminus \{\role R\}} 
\left(
\begin{array}l	
\pi(\mathcal{I},\role S) \seqOp
\cinp{\idxIfRecv{i}
\idxSign{i}.\AuxWb^*_\idxSign {i}}{x_\idxSign{i}}{\role R}\seqOp\\
\idx{i} \while{x_\idxSign{i}}{\pi(\mathcal{I},\role S)\seqOp\\
\quad\cout{\idxClose{i} \idxSign{i}.\AuxWe^*_\idxSign{i}}{\okC}{\role R}
\seqOp\\
\quad\cinp{\idxIfRecv{i} \idxSign{i}.\AuxWb^*_\idxSign{i}}{x_\idxSign{i}}{\role
R}}
\end{array}
,\Sigma_{\role S}
\right)_{\role S}	
\end{array}
\right)
\parallel\\
\left(
\begin{array}l	
\pi(\mathcal{I},\role R)\seqOp\\
 \idx{i} \mathtt{while} \; b \; \left\{
\begin{array}{l}	
 \left(\prod\limits_{\role R' \in
 \rolesFunc(\mathcal{I})
 \smallsetminus \{\role R\}} \cout{\idxTrue{i}
 \idxSign{i}.\AuxWb^*_\idxSign{i}}{\trueC}{\role
 R'}\right)\seqOp\pi(\mathcal {I},\role R)\seqOp\\
 \prod\limits_{\role R' \in \rolesFunc(\mathcal{I}) \smallsetminus \{\role R\}}
 \cinp{\idxClose{i} \idxSign{i}.\AuxWe^*_\idxSign{i}}{\_}{\role R'}
\end{array}
\right\}\seqOp\\[3em]
 \prod\limits_{\role R'\in \rolesFunc(\mathcal{I}) \smallsetminus \{\role R\}}
 \cout{\idxFalse{i} \idxSign{i}.\AuxWb^*_\idxSign{i}}{\falseC}{\role R'}
\end{array}
 ,\Sigma_{\role R}
\right)_{\role R}
 \end{multline*}
 exactly the projection  of $ \mathcal{I}\seqOp\whileKey{b \at
 \role R}{\mathcal{I}}{i}$.
 
As far as events are concerned, in $\Prop(\net')$ we have all the needed
events since, in particular, we have already done the unfolding of the while
in all the roles. Concerning the ordering, at the \AIOC{} level, we have two
kinds of causal dependencies: (1) events in the unfolded process precede the
guard event; (2) the guard event precedes the events in the body.  The first
kind of causal dependency is matched at the \APOC{} level thanks to the
auxiliary synchronisations that close the unfolded body (which are not
removed by $\Prop$) using synchronisation and sequentiality. The second kind
of causal dependency is matched thanks to the auxiliary synchronisations that
start the following iteration using  synchronisation, sequentiality and
while.\\ The case when the guard evaluates to $\falseC$ is simpler.
\item[$\leftarrow$] The only possible transition from the \APOC{} is the
evaluation of the guard from the coordinator. This reduces $\net$ to $\net'$
above and the thesis follows from the same reasoning.
\end{description}

\item[Case $\idx{i} \scope{l}{\role R}{\mathcal I}{\Delta}$] from the definition
of the projection
\begin{multline*}
\hspace{50pt}\net = \left(
\parallel_{\role R' \in \rolesFunc(\mathcal{I})
\smallsetminus \{\role R\}}
 \roleExec{\psscope{i}{l}{\role R}{\pi({\mathcal I},\role R')},\Sigma_{\role
R'}}{\role R'}\right) \parallel \\
\roleExec{\pscope{i}{l}{\role R}{\pi({\mathcal I},\role R)}{\Delta}{\rolesFunc({\mathcal
I})},\Sigma_{\role R}}{\role R}
\end{multline*}

\begin{description}
\item[$\rightarrow$] 
Let us consider the case when the scope is updated. At the \AIOC{} level all the
possible transitions have label of the form $\mathcal{I}'$ and are obtained by
applying Rule $\did {\AIOC} {\AdaptRule}$. Correspondingly, at the \APOC{} level
one applies Rule $\did{\APOC}{Lead-\AdaptRule}$ to the coordinator of the
update, obtaining
\begin{multline*}
\hspace{50pt}\net' = \left(\parallel_{\role R' \in \rolesFunc(\mathcal{I})
\smallsetminus \{\role R\}}
 \roleExec{\psscope{i}{l}{\role R}{\pi({\mathcal I},\role R')},\Sigma_{\role
R'}}{\role R'}\right) \parallel \\
\roleExec{
\begin{array}{l}
\left(\prod\limits_{\role R' \in \rolesFunc({\mathcal{I}}) 
\smallsetminus
\{\role R\}}
\cout{\idxSign{i}.\AuxSb^*_\idxSign{i}}{\pi(\mathcal{I}',\role R')}{\role
 R'}\right)\seqOp\\[1.5em]
\pi(\mathcal{I}',\role R)\seqOp\\[.5em]
 \prod\limits_{\role R' \in \rolesFunc({\mathcal{I}}) \smallsetminus \{\role R\}}
 \cinp{\idxSign{i}.\AuxSe^*_\idxSign{i}}{\_}{\role R'}
\end{array},
 \Sigma_{\role R}
 }{\role R}
\end{multline*}
By applying the $\upd$ function we get:
$$
\upd(\net') = 
\left(\parallel_{\role R' \in
\rolesFunc(\mathcal{I}) \smallsetminus \{\role R\}} \roleExec{
\pi(\mathcal{I}',\role R'),\Sigma_{\role R'}}{\role R'} \right)
\parallel
\roleExec{\pi(\mathcal{I}',\role R),\Sigma_{\role R}}{\role R}
$$
This is exactly the projection of the \AIOC{} obtained after applying the
rule $\did{\AIOC}{\AdaptRule}$.
The conditions on events are inherited. Observe that the closing event of the
scope is replaced by events corresponding to the auxiliary interactions
closing the scope. This allows us to preserve the causality dependencies also
when the scope is inserted in a context.
  
The case of Rule $\did{\AIOC}{\NoAdaptRule}$ is simpler.

\item[$\leftarrow$] The only possible transitions from the \APOC{} are the ones
of the coordinator of the update checking whether to apply an update or not.
This reduces $\net$ to $\net'$ above and the thesis follows from the same
reasoning.\qedhere
\end{description}
\end{description}
\end{proof}

\subsection{Deadlock freedom, termination, and race freedom}
\label{subsec:deadlock}

Due to the fact that the projection preserves weak traces, we have that 
trace-based properties of the \AIOC{} are inherited by the \APOC{}.
A first example of such properties is \emph{deadlock freedom}.
\begin{definition}[Deadlock freedom]\emph{
An \emph{internal} \AIOC{} (resp.\ \APOC{}) \emph{trace} is obtained by removing
transitions labelled $\rules$ from a \AIOC{}  (resp.\ \APOC{}) trace. A \AIOC{}
(resp.\ \APOC{}) system is \emph{deadlock free} if all its maximal finite internal
traces have $\tick$ as last label.
}\end{definition}
Intuitively, internal traces are needed since labels $\rules$ do not
correspond to activities of the application and may be executed also after
application termination. The fact that after a $\tick$ only changes in the set
of available updates are possible is captured by the following lemma.
\begin{lemma}\label{lemma:tickEnds}
For each initial, connected \AIOC{} ${\mathcal I}$, state $\Sigma$, and
set of updates $\rules$, if $\tuple{\Sigma,\rules,{\mathcal I}}
\arro{\tick} \tuple{\Sigma',\rules',{\mathcal I}'}$ then each
transition of $\tuple{\Sigma',\rules',{\mathcal I}'}$ has label $\rules''$
for some $\rules''$.
\end{lemma}
\begin{proof}
The proof is by case analysis on the rules which can derive a transition with
label $\tick$. All the cases are easy.
\end{proof}

Since by construction initial \AIOC{}s are deadlock free we have that also the 
\APOC obtained by projection is deadlock free.
\begin{corollary}[Deadlock freedom]\label{cor1}
For each initial, connected \AIOC{} ${\mathcal I}$, state $\Sigma$, and set of
updates $\rules$ the \APOC{} system $\tuple{\rules,\proj({\mathcal I},\Sigma)}$
is deadlock free.
\end{corollary}

\begin{proof}
Let us first prove that for each initial, connected \AIOC{} ${\mathcal I}$,
state $\Sigma$, and set of updates $\rules$, the \AIOC{} system $\tuple {\state,
\rules, \mathcal{I}}$ is deadlock free.
This amounts to prove that all its maximal finite internal traces
have $\tick$ as last label. 
For each trace, the proof is by induction on its length, and for each length by
structural induction on ${\mathcal I}$.
The proof is based on the fact that ${\mathcal I}$ is initial.
The induction considers a reinforced hypothesis, saying also that:
\begin{itemize}
	\item $\tick$ may occur \emph{only} as the last label of the internal trace;
	\item all the \AIOC{} systems in the sequence of transitions generating the trace,
	but the last one, are initial.
\end{itemize}
We have a case analysis on the top-level operator in ${\mathcal I}$. Note that
in all the cases, but $\zero$, at least a transition is derivable.
\begin{description}
\item[Case $\zero$] not allowed since we assumed an initial \AIOC{}.
\item[Case $\one$] trivial because by Rule $\did{\AIOC}{End}$ and Lemma
\ref{lemma:tickEnds} its only internal trace is $\tick$.
\item[Case $\assign{x}{\role R}{e}$] the only applicable rule
  is $\did{\AIOC}{Assign}$ that in one step leads to a $\one$ process. The
  thesis follows by inductive hypothesis on the length of the trace.
\item[Case $\comm{o^?}{\role R_1}{e}{\role R_2}{x}$] the only applicable
  rule is $\did{\AIOC}{Interaction}$, which leads to an assignment. Then the
  thesis follows by inductive hypothesis on the length of the trace.
\item[Case ${\mathcal I}\seqOp{\mathcal I}'$] the first transition
  can be derived either by Rule $\did{\AIOC}{Sequence}$ or Rule
  $\did{\AIOC}{Seq-end}$. In the first case the thesis follows by induction
  on the length of the trace. In the second case the trace coincides with a
  trace of ${\mathcal I}'$, and the thesis follows by structural induction.
\item[Case ${\mathcal I} \parOpI {\mathcal I}'$] the first transition
  can be derived either by Rule $\did{\AIOC}{Parallel}$ or by Rule
  \ruleName{Par-End}. In the first case the thesis follows by
  induction on the length of the trace. In the second case the thesis
  follows by Lemma~\ref{lemma:tickEnds}, since the label is
  $\tick$.
\item[Case $\ifthen{b \at \role R}{\mathcal{I}}{\mathcal{I}'}$] the first
  transition can be derived using either Rule\\ $\did{\AIOC}{If-then}$ or Rule
  $\did{\AIOC}{If-else}$. In both the cases the thesis follows by induction
  on the length of the trace.
\item[Case $\while{b \at \role R}{\mathcal{I}}$] the first transition can
  be derived using either Rule\\ $\did{\AIOC}{While-unfold}$ or Rule
  $\did{\AIOC}{While-exit}$. In both the cases the thesis follows by
  induction on the length of the trace.
\item[Case $\scope{l}{\role R}{\mathcal I}{\Delta}$] the first transition can be
	derived using either Rule $\did{\AIOC}{Up}$ or Rule $\did{\AIOC}{NoUp}$. In
	both the cases the thesis follows by induction on the length of the trace.
\end{description}
The weak internal traces of the \AIOC{} coincide with the weak internal
traces of the \APOC{} by Theorem~\ref{teo:final}, thus also the finite weak
internal traces of the \APOC end with $\tick$. The same holds for the finite strong
internal traces, since label $\tick$ is preserved when moving between strong
and weak traces, and no transition can be added after the $\tick$ thanks to
Lemma~\ref{lemma:tickEnds}.
\end{proof}

\APOC{}s also inherit termination from terminating \AIOC{}s. 

\begin{definition}[Termination]\emph{
A \AIOC{} (resp.\ \APOC{}) system terminates if all its internal traces are 
finite. 
}\end{definition}

Note that if arbitrary sets of updates are allowed, then termination of
\AIOC{}s that contain at least a scope is never granted.
Indeed, the scope can always be replaced by a non-terminating update or it
can trigger an infinite chain of updates. Thus, to exploit this result, one
should add constraints on the set of updates ensuring \AIOC{} termination.

\begin{corollary}[Termination]\label{cor2}
If the \AIOC{} system $\tuple{\Sigma,\rules,\mathcal{I}}$ terminates and
$\mathcal{I}$ is connected then the \APOC{} system
$\tuple{\rules,\proj({\mathcal I},\Sigma)}$ terminates.
\end{corollary}
\begin{proof} 
It follows from the fact that only a finite number of auxiliary actions are
added when moving from \AIOC{}s to \APOC{}s.
\end{proof}

Other interesting properties derived from weak trace equivalence are freedom
from races and orphan messages. A race occurs when the same receive (resp.\
send) may interact with different sends (resp.\ receives). In our setting, an
orphan message is an enabled send that is never consumed by a receive. Orphan
messages are more relevant in asynchronous systems, where a message may be
sent, and stay forever in the network, since the corresponding receive
operation may never become enabled. However, even in synchronous systems
orphan messages should be avoided: the message is not communicated since the
receive is not available, hence a desired behaviour of the application never
takes place due to synchronisation problems.
Trivially, \AIOC{}s avoid races and orphan
messages since send and receive are bound together in the same construct.
Differently, at the \APOC{} level, since all receive of the form
$\idx{\iota}\cinp{\idxSign{i}.o^?}{x}{\role R_1}$ in role $\role R_2$ may
interact with the sends of the form
$\idx{\iota}\cout{\idxSign{i}.o^?}{e}{\role R_2}$ in role $\role R_1$,
races may happen. However, thanks to the correctness of the projection, race
freedom holds also for the projected \APOC{}s.

\begin{corollary}[Race freedom]\label{cor:race}
For each initial, connected \AIOC{} ${\mathcal I}$, state $\Sigma$, and
set of updates $\rules$, if $\tuple{\rules,\proj({\mathcal I},\Sigma)}
\arro{\eta_1}\cdots\arro{\eta_n} \tuple{\rules',\net}$, where $\eta_i \in
\{\tau, \commLabel{o^?}{\role R_1}{v}{\role R_2}{x}, \commLabel{o^*}{\role
R_1}{X}{\role R_2}{}, \tick, {\mathcal I},\NoAdaptLabel,
\rules \}$ for each $i \in \{1,\ldots,n\}$, then in $\net$ there are no two sends
(resp.\ receives) which can interact with the same receive (resp.\ send).
\end{corollary}

\begin{proof}
We have two cases, corresponding respectively to programmer-specified and
auxiliary operations.

For programmer-specified operations, thanks to Lemma~\ref{lemma:IOCwell}, case
C1, for each global index $\xi$ there are at most two communication events
with global index $\xi$. The corresponding \APOC{} terms can be enabled only
if they are outside of the body of a while loop. Hence, their index coincides
with their global index. Since the index prefixes the operation, then no
interferences with other sends or receives are possible.

For auxiliary operations, the reasoning is similar. Note, in fact, that sends
or receives with the same global index can be created only by a unique \AIOC{}
construct, but communications between the same pair of roles are never enabled
together. This can be seen by looking at the definition of the projection.
\end{proof}

As far as orphan messages are concerned, they may appear in infinite
\APOC{} computations since a receive may not become enabled due to an
infinite loop. However, as a corollary of trace equivalence, we have
that terminating \APOC{}s are orphan-message free.

\begin{corollary}[Orphan-message freedom]\label{cor:orphan}
 For each initial, connected \AIOC{} ${\mathcal I}$, state $\Sigma$, and set of
updates $\rules$, if 
$\tuple{\rules,\proj({\mathcal I},\Sigma)}
\arro{\eta_1}\cdots\arro{\eta_n} \arro{\tick} \tuple{\rules',\net}$, where
$\eta_i \in \{\tau, \commLabel{o^?}{\role R_1}{v}{\role R_2}{x}, \commLabel{o^*}{\role
R_1}{X}{\role R_2}{}, \tick, {\mathcal I},\NoAdaptLabel,
\rules \}$, then $\net$ contains no sends.
\end{corollary}

\begin{proof}
The proof is by case analysis on the rules which can derive a transition with
label $\tick$. All the cases are easy.
\end{proof}

\section{Adaptable Interaction-Oriented Choreographies in Jolie}\label{sec:aiocj}
\setMyChor{}

In this section we present \AIOCJ (Adaptable Interaction-Oriented
Choreographies in Jolie), a development framework for adaptable distributed
applications~\cite{autonomicComputing}. \AIOCJ is one of the possible
instantiations of the theoretical framework presented in the previous sections
and it gives a tangible proof of the expressiveness and feasibility of our
approach.
We say that \AIOCJ is an instance of our theory because it follows the theory,
but it provides mechanisms to resolve the non-determinism related to the
choice of whether to update or not and on which update to select. Indeed, in
\AIOCJ updates are chosen and applied according to the state of the
application and of its running environment. To this end, updates are embodied
into adaptation rules, which specify when and whether a given update can be
applied, and to which scopes.
\AIOCJ also inherits all the correctness guarantees provided by our theory, in
particular: \emph{i}) applications are free from deadlocks and races by
construction, \emph{ii}) applications remain correct after any step of
adaptation. As in the theory, adaptation rules can be added and removed while
applications are running.
\begin{remark}\label{remark:assumption_practice}\emph{
In order for the correctness guarantees to be provided, one needs to satisfy
the required assumptions, in particular the fact that functions never block
and always return a value (possibly an error notification). If this condition
is not satisfied, at the theoretical level, the behaviour of the source \AIOC
and the behaviour of its projection may differ, as described in
Remark~\ref{remark:assumption_theory}. The behaviour of the corresponding
AIOCJ application can be different from both of them. Usually, the violation
of the assumption above leads the application to crash. This allows the
programmer to realise that the assumption has been violated.
}\end{remark}

Below we give a brief overview of the \AIOCJ framework by introducing its
components: the Integrated Development Environment (IDE), the AIOCJ compiler,
and the Runtime Environment.

\paragraph{\emph{Integrated Development Environment}}

\AIOCJ supports the writing of programs and adaptation rules in the Adaptable
Interaction-Oriented Choreography (AIOC) language, an extension of the \AIOC{}
language. We discuss the main novelties of the AIOC language in
\S~\ref{sec:aiocj_aioc}.
\AIOCJ offers an integrated environment for developing programs and adaptation
rules that supports syntax highlighting and on-the-fly syntax checking. Since
checking for connectedness (see \S~\ref{sub:connectedness}) of programs and
adaptation rules is polynomial (as proven by Theorem~\ref{teo:compl}), the IDE
also performs on-the-fly checks on connectedness of programs and rules.

\paragraph{\emph{Compiler}}

The \AIOCJ IDE also embeds the \AIOCJ compiler, which implements the procedure
for projecting AIOCs and adaptation rules into distributed executable code.
The implementation of the compiler is based on the rules for projecting
\AIOC{}s, described in \S~\ref{sec:endpoint_projection}.
The target language of the \AIOCJ compiler is Jolie~\cite{jolie,joliepaper}, a
Service-Oriented language with primitives similar to those of our \APOC{}
language. Jolie programs are also called \emph{services}.
Given an AIOC program, the \AIOCJ compiler produces one Jolie service for each role in the source AIOC.
The compilation of an AIOC rule produces one Jolie service for each role and
an additional service that describes the applicability condition of the rule.
All these services are enclosed into an \textsf{Adaptation Server}, described
below.

\paragraph{\emph{Runtime Environment}}

The \AIOCJ runtime environment comprises a few Jolie services that support the
execution and adaptation of compiled programs. The main services of the \AIOCJ
runtime environment are the \textsf{Adaptation Manager}, \textsf{Adaptation
Server}s, and the \textsf{Environment}.
The compiled services interact both among themselves and with an
\textsf{Adaptation Manager}, which is in charge of managing the adaptation
protocol. \textsf{Adaptation Server}s contain adaptation rules, and they can
be added or removed dynamically, thus enabling dynamic changes in the set of
rules, as specified by Rule $\did{\AIOC}{Change-Updates}$. When started, an
\textsf{Adaptation Server} registers itself at the \textsf{Adaptation Manager}.
The \textsf{Adaptation Manager} invokes the registered \textsf{Adaptation
Server}s to check whether their adaptation rules are applicable. In order to
check whether an adaptation rule is applicable, the corresponding
\textsf{Adaptation Server} evaluates its applicability condition.
Applicability conditions may refer to the state of the role which coordinates
the update, to properties of the scope, and to properties of the environment
(e.g., time, temperature, etc.), stored in the \textsf{Environment} service.

In the remainder of this Section we detail the grammar of the AIOC
language used by \AIOCJ in \S~\ref{sec:aiocj_aioc}, we illustrate the use of
\AIOCJ{} on a simple example in \S~\ref{sec:aiocj_preview}, we discuss some
relevant implementation aspects of \AIOCJ in \S~\ref{sec:implementation}, and
we present some guidelines on how to use scopes in \AIOCJ in
\S~\ref{sub:guidelines_for_programming_in_aiocj}.

\subsection{DIOC Language Extensions in AIOCJ}\label{sec:aiocj_aioc}
AIOCs extend \AIOC{}s with:

\begin{itemize}
  \item the definition of adaptation rules, instead of updates, that include the information needed to evaluate their applicability condition;
  \item the definition of constructs to express the deployment information
  needed to implement real-world distributed applications.
\end{itemize}
Below we describe in detail, using the Extended Backus-Naur
Form~\cite{Backus59}, the new or refined constructs introduced by the AIOC
language.

\noindent
\textbf{Function inclusions.}
The AIOC language can exploit functionalities provided by external
services via the \lstinline{include} construct. The syntax is as
follows.

\begin{lstlisting}
Include ::= include FName [,FName]* from "URL" [with PROTOCOL]
\end{lstlisting}

This allows one to reuse existing legacy code and to interact with
third-party external applications.  As an example, the $\seller$ of
our running example can exploit an external database to implement the
functionality for price retrieval \lstinline{getPrice}, provided that
such a functionality is exposed as a service.  If the service is
located at \lstinline{"socket://myService:8000"} and accessible via
the \lstinline{"SOAP"} protocol we enable its use with the following
inclusion:

\begin{lstlisting}
include getPrice from "socket://myService:8000" with "SOAP"
\end{lstlisting}

Similarly, the $\bank$ IT system is integrated in the example by invoking the
function \lstinline{makePayment}, also exposed as a service.

This feature enables a high degree of integration since \AIOCJ
supports all protocols provided by the underlying Jolie language,
which include TCP/IP, RMI, SOAP, XML/RPC, and their encrypted
correspondents over SSL.

External services perfectly fit the theory described in previous sections,
since they are seen as functions, and thus introduced in expressions. Notice,
in particular, that Remark~\ref{remark:assumption_practice} applies.

\noindent
\textbf{Adaptation rules.}
Adaptation rules extend \AIOC{} updates, and are a key ingredient of
the \AIOCJ{} framework. The syntax of adaptation rules is as follows:

\begin{lstlisting}[xleftmargin=130pt]
Rule ::= rule {
  [ Include ]*
  on { Condition }
  do { Choreography }
}
\end{lstlisting}

The applicability condition of the adaptation rule is specified using
the keyword \lstinline{on}, while the code to install in case
adaptation is performed (which corresponds to the \AIOC{} update) is
specified using the keyword \lstinline{do}.
Optionally, adaptation rules can include functions they rely on.

The \lstinline{Condition} of an adaptation rule is a propositional
formula which specifies when the rule is applicable. To this end, it
can exploit three sources of information: local variables of the
coordinator of the update, environmental variables, and properties of
the scope to which the adaptation rule is applied. Environmental
variables are meant to capture contextual information that is not
under the control of the application (e.g., temperature, time,
available resources, \dots).  To avoid ambiguities, local variables
of the coordinator are not prefixed, environment variables are
prefixed by \lstinline{E}, and properties of the scope are prefixed
by \lstinline{N}.

For example, if we want to apply an adaptation rule only to those
scopes whose property \lstinline{name} is equal to the
string \lstinline{"myScope"} we can use as applicability condition the
formula \lstinline{N.name == "myScope"}.

\noindent
\textbf{Scopes.}
As described above, scopes in AIOC also feature a set of properties
(possibly empty).  Scope properties describe the current
implementation of the scope, including both functional and
non-functional properties. Such properties are declared by the
programmer, and the system only uses them to evaluate the applicability
condition of adaptation rules, to decide whether a given adaptation
rule can be applied to a given scope. Thus the syntax for scopes in
AIOC is:

\begin{lstlisting}[xleftmargin=130pt]
scope @Role { Choreography }
[ prop { Properties } ]?
\end{lstlisting}

\noindent where clause \lstinline{prop} introduces a list of comma-separated
assignments of the form \\\lstinline{N.ID = Expression}.

For instance, the code 

\begin{lstlisting}[xleftmargin=130pt,basicstyle=\ttfamily]
scope @R {
  // AIOC code
} prop { N.name = "myScope"}
\end{lstlisting}

\noindent specifies that the scope has a property \lstinline{name} set to the
string \lstinline{"myScope"}, thus satisfying the applicability condition
\lstinline{N.name == "myScope"} discussed above.

\noindent
\textbf{Programs.}
AIOC programs have the following structure.

\begin{lstlisting}[xleftmargin=130pt]
Program ::=
  [ Include ]*
  preamble {
    starter: Role
    [ Location ]*
  }
  aioc { Choreography }
\end{lstlisting}

where \lstinline{Include} allows one to include external
functionalities, as discussed above, and keyword \lstinline{aioc}
introduces the behaviour of the program, which is a \AIOC{} apart for
the fact that scopes may define properties, as specified above.

The keyword \lstinline{preamble} introduces deployment information,
i.e., the definition of the \lstinline{starter} of the AIOC and the
\lstinline{Location} of participants.

The definition of a \lstinline{starter} is mandatory and designs which role is
in charge of waiting for all other roles to be up and running before starting
the actual computation. Any role can be chosen as \lstinline{starter}, but the
chosen one needs to be launched first when running the distributed application.

\lstinline{Location}s define where the participants of the AIOC will be
deployed. They are specified using the keyword \lstinline{location}: 

\begin{lstlisting}[xleftmargin=130pt]
Location ::= location@Role:"URL"
\end{lstlisting}

where \lstinline{Role} is the name of a role (e.g., \lstinline{Role1}) and
\lstinline{URL} specifies where the service can be found (e.g.,
\lstinline{"socket://Role1:8001"}). When not explicitly defined, the
projection automatically assigns a distinct local TCP/IP location to each
role.

\subsection{AIOCJ Workflow}\label{sec:aiocj_preview}
Here we present a brief description of how a developer can write an
adaptable distributed system in \AIOCJ, execute it, and change its
behaviour at runtime by means of adaptation rules.
For simplicity, we reuse here the minimal example presented in the
Introduction, featuring a scope that encloses a price offer from the
$\seller$ to the $\buyer$ and an update --- here an adaptation rule
--- that provides a discount for the $\buyer$. We report the AIOC program in
the upper part and the adaptation rule in the lower part of
Figure~\ref{fig:example_AIOCJ}.

\begin{figure}[t]
\begin{lstlisting}[mathescape=true,xleftmargin=15pt,numbers=left,label=es:preview_aioc]
include getPrice from "socket://storage.seller.com:80"
preamble{ starter: Seller }
aioc {
 scope @Seller {
  order_price@Seller = getPrice( order );
  offer: Seller( order_price ) -> Buyer( prod_price )
} prop { N.scope_name = "price_inquiry" } }

\end{lstlisting}

\noindent\hfil\rule{.8\textwidth}{.4pt}\hfil

\begin{lstlisting}[mathescape=true,xleftmargin=15pt,numbers=left,label=es:preview_rule]
rule {
 include getPrice from "socket://storage.seller.com:80"
 include isValid from "socket://discounts.seller.com:80"
 on { N.scope_name == "price_inquiry" and E.season == "Fall" }
 do {
  cardReq: Seller( _ ) -> Buyer( _ );
  card_id@Buyer = getInput( "Insert your customer card ID");
  cardRes: Buyer( card_id ) -> Seller( buyer_id );
  if isValid( buyer_id )@Seller {
   order_price@Seller = getprice( order )*0.9;
  } else {
   order_price@Seller = getPrice( order )
  };
  offer: Seller( order_price ) -> Buyer( prod_price )
  }
}
\end{lstlisting}
\caption{An AIOC program (upper part) and an applicable adaptation rule (lower
part).\label{fig:example_AIOCJ}}
\end{figure}
At Line 1 of the AIOC program we have the inclusion of function
\lstinline{getPrice}, which is provided by a service within the internal
network of the Seller, reachable as a TCP/IP node at URL
\lstinline{"storage.seller.com"} on port \lstinline{"80"}.
At Line 2 of the AIOC program we have the \lstinline{preamble}. The preamble
specifies deployment information, and, in particular, defines the
\lstinline{starter}, i.e., the service that ensures that all the participants
are up and running before starting the actual computation. No locations are
specified, thus default ones are used. The actual code is at Lines 4--7, where
we declare a scope.
At Line 7 we define a property \lstinline{scope_name} of this scope
with value \lstinline{"price_inquiry"}.

Figure~\ref{fig:framework_scheme} depicts the process of compilation
\pnum{1} and execution \pnum{2} of the AIOC.
From left to right, we write the AIOC and we compile it into a set of
executable Jolie services (\textsf{Buyer service} and \textsf{Seller
service}).  To execute the projected system, we first launch the
\textsf{Adaptation Manager} and then the two compiled services, starting from
the \textsf{Seller}, which is the \lstinline{starter}. Since there is no
compiled adaptation rule, the result of the execution is the offering of the
standard price to the \lstinline{Buyer}.

Now, let us suppose that we want to adapt our system to offer discounts in the
Fall season. To do that, we can write the adaptation rule shown in the lower
part of Figure~\ref{fig:example_AIOCJ}.

Since we want to replace the scope for the \lstinline{price_inquiry}, we
define, at Line 4, that this rule applies only to scopes with property
\lstinline{scope_name} set to \lstinline{"price_inquiry"}. Furthermore, since
we want the update to apply only in the Fall season, we also specify that the
environment variable \lstinline{E.season} should match the value
\lstinline{"Fall"}. The value of \lstinline{E.season} is retrieved from the
\textsf{Environment Service}.

The rule uses two functions. At Line 2, it includes function
\lstinline{getPrice}, which is the one used also in the body of the scope. The
other function, \lstinline{isValid}, is included at Line 3 and it is provided
by a different TCP/IP node, located at URL \lstinline{"discounts.seller.com"}
on port \lstinline{"80"} within the internal network of the Seller.

The body of the adaptation rule (Lines 6--14) specifies the same behaviour
described in Figure~\ref{fig:rule_price_inquiry} in the Introduction: the
\lstinline{Seller} asks the \lstinline{Buyer} to provide its
\lstinline{card_id}, which the \lstinline{Seller} uses to provide a discount
on the price of the ordered product. We depict the inclusion of the new
adaptation rule (outlined with dashes) and the execution of the adaptation at point \pnum{3} of
Figure~\ref{fig:framework_scheme}. From right to left,
we write the rule and we compile it. The compilation of a (set of) adaptation
rule(s) in AIOCJ produces a service, called \textsf{Adaptation Server} (also outlined with dashes), that
the \textsf{Adaptation Manager} can query to fetch adaptation rules at
runtime. The compilation of the adaptation rule can be done while the
application is running. After the compilation, the generated
\textsf{Adaptation Server} is started and registers on the \textsf{Adaptation
Manager}. Since the rule relies on the environment to check its
applicability condition, we also need the \textsf{Environment service} to be
running. In order for the adaptation rule to apply we need the environment
variable \texttt{season} to have value \lstinline{"Fall"}.

\begin{figure}[t]
  \centering
  \includegraphics[width=\textwidth]{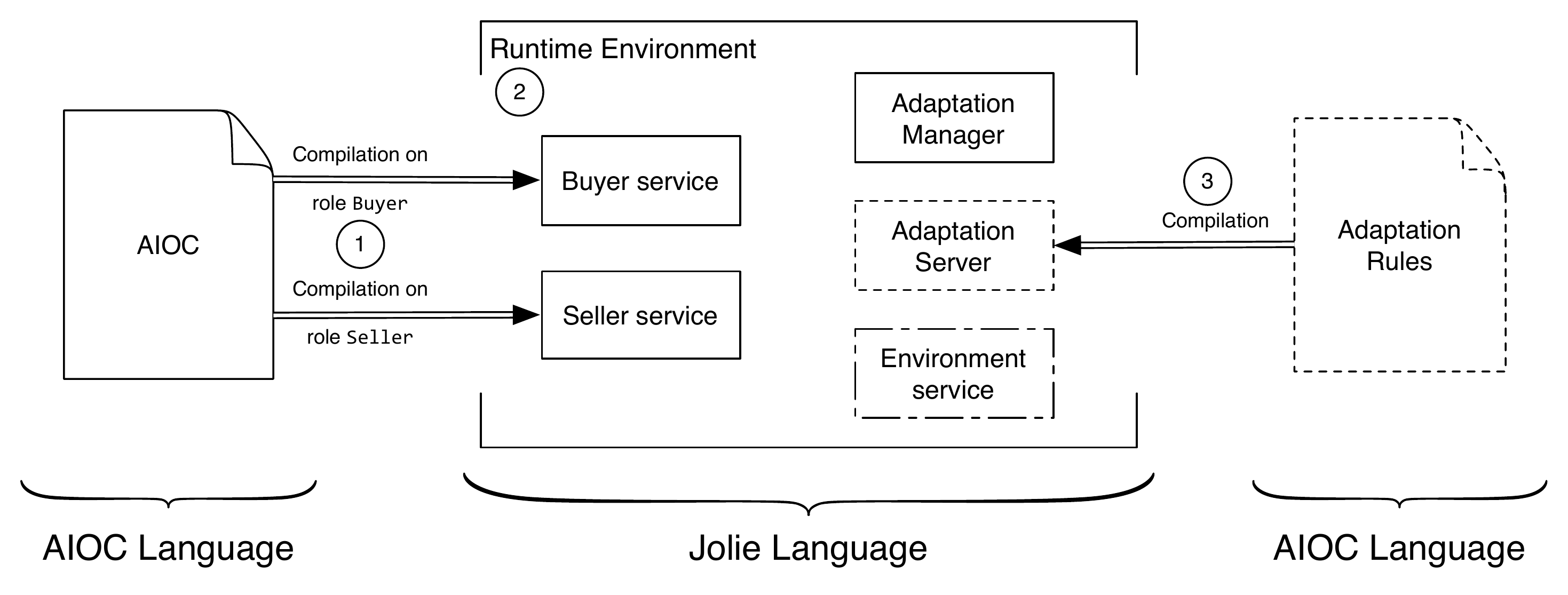}
  \caption{Representation of the AIOCJ framework --- Projection and execution
    of the example in Figure~\ref{fig:example_AIOCJ}.}
  \label{fig:framework_scheme} 
\end{figure}

\subsection{Implementation}\label{sec:implementation}
AIOCJ is composed of two elements: the AIOCJ Integrated Development
Environment (IDE), named
\toolplugin{}, and the adaptation middleware that enables AIOC programs to
adapt, called \toolmid{}.

\toolplugin{} is a plug-in for Eclipse~\cite{eclipse} based on
Xtext~\cite{xtext}. Starting from a grammar, Xtext generates the parser for
programs written in the AIOC language. Result of the parsing is an abstract
syntax tree (AST) we use to implement \emph{i}) the checker of connectedness
for AIOC programs and adaptation rules and \emph{ii}) the generation of Jolie
code for each role. Since the check for connectedness has polynomial
computational complexity (cf.~Theorem~\ref{teo:compl}) it is efficient
enough to be performed while editing the code.
Figure~\ref{fig:con_seq} shows \toolplugin{} notifying the error on the first
non-connected instruction (Line 13).
\begin{figure}[t]
	\begin{center}
		\includegraphics[width=\textwidth]{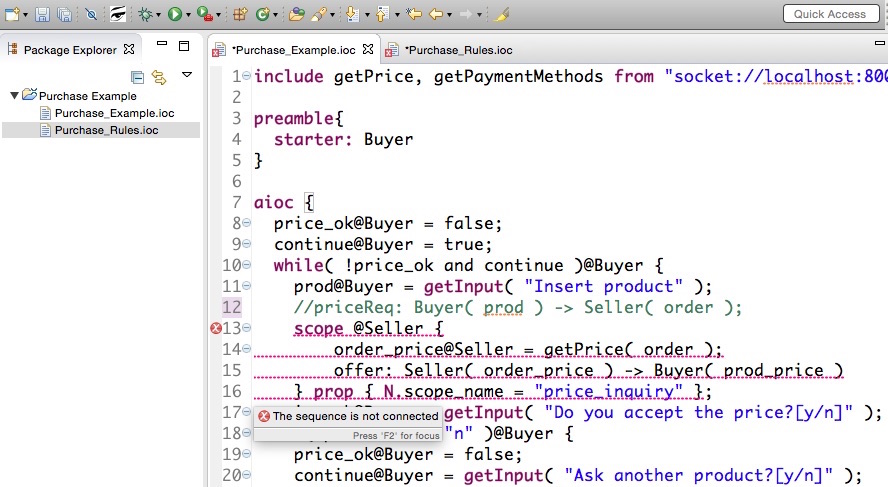}
	\end{center}
	\caption{Check for connectedness.}
	\label{fig:con_seq}
\end{figure}

As already mentioned, we chose Jolie as target language of the compilation of
\AIOCJ because its semantics and language constructs naturally lend themselves
to translate our theoretical results into practice.
Indeed, Jolie supports architectural primitives like dynamic embedding,
aggregation, and redirection, which ease the compilation of AIOCs. 

Each scope at the AIOC level is projected into
a specific sub-service for each role.  The roles run the projected
sub-services by embedding them and access them via redirection. In this way,
we implement adaptation by disabling the default sub-service and by
redirecting the execution to a new one, obtained from the \textsf{Adaptation
server}.

When at runtime the coordinator of the update reaches the beginning of that
scope, it queries the \textsf{Adaptation Manager} for adaptation rules to
apply. The \textsf{Adaptation Manager} queries each \textsf{Adaptation Server}
sequentially, based on their order of registration. On its turn, each
\textsf{Adaptation Server} checks the applicability condition of each of its
rules. The first rule whose applicability condition holds is applied. The
adaptation manager sends to the coordinator the updates that are then
distributed to the other involved roles. In each role, the new code replaces
the old one.

To improve the performance, differently from the theory, in AIOCJ adaptation
rules are compiled statically, and not by the coordinator of the update when
the update is applied.

\subsection{AIOCJ Practice}\label{sub:guidelines_for_programming_in_aiocj}
AIOCJ provides scopes as a way to specify which parts of the application
can be updated. As a consequence, deciding which parts of the code
to enclose in scopes is a relevant and non trivial decision. Roughly,
one should enclose into scopes parts of the code that may need to
be updated in the future.  This includes for instance parts of the
code which model business rules, hence may change according to
business needs, parts of the code which are location dependent, hence
may be updated to configure the software for running in different
locations, parts of the code which are relevant for performance or
security reasons, hence may be updated to improve the performance or
the security level.

One may think that a simple shortcut would be to have either a big
scope enclosing the whole code of the application or many small scopes
covering all the code, but both these solutions have relevant
drawbacks. Indeed, if one replaces a big scope in order to change a
small piece of it, (s)he has to provide again all the code, while one
would like to provide only the part that needs to change. Furthermore,
update for the big scope is no more available as soon as its execution
starts, which may happen quite earlier than when the activity to be
updated starts. Using small scopes also is not a solution, since an
update may need to cover more than one such scope, but there is no way
to ensure that multiple scopes are updated in a coordinated way. Also,
both the approaches have performance issues: in the first case large
updates have to be managed, in the second case many checks for the
updates are needed, causing a large overhead.

In order to write updates for a given scope, one needs to have some knowledge
about the scope that the update is going to replace, since the update will run
in the context where the scope is. This includes, for instance, knowing which
variables are available, their types and their meaning, and in which variables
results are expected. This information can be either tracked in the
documentation of the application, or added to the scope properties. In this
way, the adaptation rule embodying the update can check this information
before applying the update (see for instance the approach
of~\cite{JORBApaper}). The fact that this information is needed for all scopes
is another point against the approach of using many small scopes.

\section{Related work and discussion}\label{sec:related} 

Recently, languages such as Rust~\cite{rust} or SCOOP~\cite{scoop} have been
defined to provide high-level primitives to program concurrent applications
avoiding by construction some of the risks of concurrent programming. For
instance, Rust ensures that there is exactly one binding to any given resource
using ownership and borrowing mechanisms to achieve memory
safety. SCOOP instead provides the \texttt{separate} keyword to enable
programming patterns that enforce data-race freedom.
In industry, Rust has started to be used for the development of complex
concurrent/distributed systems (e.g., the file storage system of
Dropbox~\cite{dropbox_rust} or the parallel rendering engine of
Mozilla~\cite{servo_rust}). 
However, industry has not yet widely adopted any language-based solution
trading expressive power of the language for safety guarantees.

Inspired by the languages above, we have presented high-level primitives for
the dynamic update of distributed applications. We guarantee the absence of
communication deadlocks and races by construction for the running distributed
application, even in presence of updates that were unknown when the
application was started. More generally, the compilation of a \AIOC{}
specification produces a set of low-level \APOC{}s whose behaviour is
compliant with the behaviour of the originating \AIOC{}.

As already remarked, our theoretical model is very general. Indeed, whether to
update a scope or not and which update to apply if many are available is
completely non deterministic. Choosing a particular policy for decreasing the
non-determinism of updates is orthogonal with respect to our results. Our
properties are preserved by any such policy, provided that the same policy is
applied both at the \AIOC{} and at the \APOC{} level.
AIOCJ presents a possible instantiation of our general approach with a
concrete mechanism for choosing updates according to the state of the
application and of its running environment.

Our work is on the research thread of choreographies and multiparty session
types~\cite{hondaESOPext,hondaPOPL,poplmontesi,Castagna,BasuBO12,SalaunBR12}.
One can see~\cite{BETTYsurvey} for a description of the field. 
Like choreographies and multiparty session types, \AIOC{}s target
communication-centred concurrent and distributed applications, avoiding
deadlocks and communication races. Thanks to the particular structure of the
language, \AIOC{}s provide the same guarantees without the need of types.
Below we just discuss the approaches in this research thread closest to ours.
The two most related approaches we are aware of are based on multiparty
session types, and deal with dynamic software updates~\cite{DSUtypes} and with
monitoring of self-adaptive systems~\cite{coppo_multiparty_adaptation}.
The main difference between~\cite{DSUtypes} and our approach is that
\cite{DSUtypes} targets concurrent applications which are not
distributed.  Indeed, it relies on a check on the global state of the
application to ensure that the update is safe. Such a check cannot be
done by a single role, thus it is impractical in a distributed setting.
Furthermore, the language in~\cite{DSUtypes} is much more constrained than
ours, e.g., requiring each pair of participants to interact on a dedicated
pair of channels, and assuming that all the roles that are not the sender or
the receiver within a choice behave the same in the two branches.
The approach in \cite{coppo_multiparty_adaptation} is very different from
ours, too. In particular, in \cite{coppo_multiparty_adaptation}, all the
possible behaviours are available since the very beginning, both at the level
of types and of processes, and a fixed adaptation function is used to switch
between them. This difference derives from the distinction between
self-adaptive applications, as they discuss, and applications updated from the
outside, as in our case.

We also recall \cite{DiGiustoP13}, which uses types to ensure safe
adaptation. However, \cite{DiGiustoP13} allows updates only when no
session is active, while we change the behaviour of running \AIOC{}s.

We highlight that, contrarily to our approach, none of the approaches
above has been implemented.

Various tools~\cite{dsol-ghezzi,pistore_correct_and_adaptable,onthefly}
support adaptation exploiting automatic planning techniques in order to
elaborate, at runtime, the best sequence of activities to achieve a given
goal. These techniques are more declarative than ours, but, to the best of our
knowledge, they are not guaranteed to always find a plan to adapt the
application.

\cite{NeamtiuH09} presents an approach to safe dynamic software
updates for concurrent systems. According to~\cite{NeamtiuH09}, safe means
that all the traces are version consistent, namely that for each trace there
is a version of the application able to produce it. In our case, we want the
effects of the update to be visible during the computation, hence we do not
want this property to hold. Instead, for us safe means that the \AIOC{} and
the projected \APOC{} are weak system bisimilar, and that they are both free
from deadlocks and races.

Among the non-adaptive languages, Chor~\cite{poplmontesi} is the closest to
ours. Indeed, like ours, Chor is a choreographic language that compiles to
Jolie. Actually, AIOCJ shares part of the Chor code base.
Our work shares with~\cite{MontesiCompChor} the interest in
choreographies composition.  However, \cite{MontesiCompChor} uses
multiparty session types and only allows static parallel composition,
while we replace a term inside an arbitrary context at runtime.

Extensions of multiparty session types with error
handling~\cite{carboneEXC,giachinoescape} share with us the difficulties in
coordinating the transition from the expected pattern to an alternative one,
but in their case the error recovery pattern is known since the very
beginning, thus considerably simplifying the analysis.

We briefly compare to some works that exploit choreographic descriptions for
adaptation, but with very different aims. For instance, \cite{ICSOC07} defines
rules for adapting the specification of the initial requirements for a
choreography, thus keeping the requirements up-to-date in presence of run-time
changes. Our approach is in the opposite direction: we are not interested in
updating the system specification tracking system updates, but in programming
and ensuring correctness of adaptation itself.
Other formal approaches to adaptation represent choreographies as annotated
finite state automata. In \cite{DYCHOR06} choreographies are used to propagate
protocol changes to the other peers, while \cite{SCC09} presents a test to
check whether a set of peers obtained from a choreography can be reconfigured
to match a second one. Differently from ours, these works only provide change
recommendations for adding and removing message sequences.

An approach close to ours is multi-tier
programming~\cite{multi_tier,hiphop,multi_tier_web}, where the programmer
writes a single program that is later automatically distributed to different
tiers. This is similar to what we do by projecting a \AIOC{} on the different
roles.
In particular, \cite{multi_tier} considers functional programs with role
annotations and splits them into different sequential programs to be run in
parallel. Multi-tier approaches such as HipHop \cite{hiphop} or Link
\cite{multi_tier_web} have been implemented and used to develop three tier web
applications. 
Differently from ours, these works abstract away communication information as
much as possible, while we emphasise this kind of information. Furthermore,
they do not take into account code update.
It would be interesting to look for cross-fertilisation results between our
approach and multi-tier programming. For instance, one could export our
techniques to deal with runtime updates into multi-tier programming. In the
other direction, we could abstract away communications and let the compiler
manage them according to the programmed flow of data, as done in multi-tier
programming. For instance, in Figure~\ref{fig:purchase_scenario} the
communications at Lines 22, 24, and 26 could be replaced by local assignments,
since the information on the chosen branch is carried by the auxiliary
communications. In general, communications seem not strictly needed from a
purely technical point of view, however they are relevant for the systems we
currently consider. Furthermore, the automatic introduction of communications
in an optimised way is not always trivial since the problem is NP-hard even
for languages that do non support threads~\cite{multi_tier_add_com}.

On a broader perspective, our theory can be used to inject guarantees of
freedom from deadlocks and races into many existing approaches to adaptation,
e.g., the ones in the surveys~\cite{adapt-survey12,ADAPTGhezzi}. However, this
task is cumbersome, due to the huge number and heterogeneity of those
approaches. For each of them the integration with our techniques is far from
trivial. Nevertheless, we already started it. Indeed, \AIOCJ follows the
approach to adaptation described in~\cite{JORBApaper}. However, applications
in~\cite{JORBApaper} are not distributed and there are no guarantees on the
correctness of the application after adaptation.
Furthermore, in the website of the AIOCJ project~\cite{AIOCJ}, we give
examples of how to integrate our approach with
distributed~\cite{distributedAOP} and dynamic~\cite{YangCSSSM02}
Aspect-Oriented Programming (AOP) and with Context-Oriented Programming (COP)
\cite{cop}. In general, we can deal with cross-cutting concerns like logging
and authentication, typical of AOP, viewing pointcuts as empty scopes and
advices as updates. Layers, typical of COP, can instead be defined by updates
which can fire according to contextual conditions.
As future work, besides applying the approach to other adaptation mechanisms,
we also plan to extend our techniques to deal with multiparty session
types~\cite{hondaESOPext,hondaPOPL,poplmontesi,Castagna}. The main challenge here is
to deal with multiple interleaved sessions whilst each of our choreographies corresponds
to a single session. An initial analysis of the problem
is presented in~\cite{WS-SEFM}.

Also the study of more refined policies for rule selection,
e.g., based on priorities, is a topic for future work.

\bibliographystyle{abbrv}
\bibliography{biblio}

\newpage
\appendix

\section{Proof of Theorem~\ref{teo:compl}}
\label{app:complexity}

In order to prove the bound on the complexity of the connectedness
check we use the lemma below, showing that the checks to verify the
connectedness 
for a single sequence operator can be performed in linear time on the size of
the sets generated by $\transI$ and $\transF$.

\begin{lemma}\label{lemma:limit_connectedness}
Given $S, S'$ sets of multisets of two elements, checking if $\forall\ s \in
S \ . \ \forall s'\in S' \ . \ s \cap s' \neq \emptyset$ can be done in
$O(n)$ steps, where $n$ is the maximum of $|S|$ and $|S'|$.
\end{lemma}

\begin{proof}
Without loss of generality, we can assume that $|S| \leq |S'|$. If $|S| \leq
9$ then the check can be performed in $O(n)$ by comparing all the elements in
$S$ with all the elements in $S'$. If $|S| > 9$ then at least 4 distinct
elements appear in the multisets in $S$ since the maximum number of multisets
with cardinality 2 obtained by $3$ distinct elements is $9$. In this case the
following cases cover all the possibilities:

\begin{itemize}

 \item there exist distinct elements $a,b,c,d$ s.t.~$\{a,b\},\{a,c\}$, and
	$\{a,d\}$ belong to $S$. In this case for the check to succeed all the
	multisets in $S'$ must contain $a$, otherwise the intersection of the
	multiset not containing $a$ with one among the multisets $\{a,b\},\{a,c\}$,
	and $\{a,d\}$ is empty. Similarly, since $|S'| > 9$, for the check to
	succeed all the multisets in $S$ must contain $a$. Hence, if
	$\{a,b\},\{a,c\}$, and $\{a,d\}$ belong to $S$ then the check succeeds iff
	$a$ belongs to all the multisets in $S$ and in $S'$.

 \item there exist distinct elements $a,b,c,d$ s.t.~$\{a,b\}$ and $\{c,d\}$
belong to $S$. In this case the check succeeds only if $S'$ is a subset of
$\{\{a,c\},\{a,d\},\{b,c\},\{b,d\} \}$. Since $|S'| > 9$ the check can never
succeed.

 \item there exist distinct elements $a,b,c$ s.t.~$\{a,a\}$ and $\{b,c\}$
	belong to $S$. In this case the check succeeds only if $S'$ is a subset of
	$\{\{a,b\},\{a,c\} \}$. Since $|S'| > 9$ the check can never succeed.

 \item there exist distinct elements $a,b$ s.t.~$\{a,a\}$ and $\{b,b\}$
	belong to $S$. In this case the check succeeds only if $S'$ is a subset of
	$\{\{a,b\} \}$. Since $|S'| > 9$ the check can never succeed.

\end{itemize}

\noindent Summarising, if $|S| > 9$ the check can succeed iff all the multisets in $S$
and in $S'$ share a common element. The existence of such an element can be
verified in time $O(n)$.
\end{proof}

\compl*

\begin{proof}
To check the connectedness of ${\mathcal I}$ we first compute the values of
the functions $\transI$ and $\transF$
for each node of the abstract syntax tree (AST). We then check for each
sequence operator whether connectedness 
holds.

The functions $\transI$ and $\transF$ associate to each node a set of pairs
of roles.  Assuming an implementation of the data set structure based on
balanced trees (with pointers), $\transI$ and $\transF$ can be computed in
constant time for interactions, assignments, $\one$, $\zero$, and sequence
constructs. For while and scope constructs computing $\transF({\mathcal I}')$
requires the creation of balanced trees having an element for every role of
${\mathcal I}'$. Since the roles are $O(n)$, $\transF({\mathcal I}')$ can be
computed in $O(n\log(n))$.  For parallel and if constructs a union of sets is
needed. The union costs $O(n\log(n))$ since each set generated by $\transI$
and $\transF$ contains at maximum $n$ elements.

Since the AST contains $n$ nodes, the computation of the sets generated by
$\transI$ and $\transF$
can be performed in $O(n^2\log(n))$.

To check connectedness 
we have to verify that for each node
${\mathcal I}'\seqOp{\mathcal I}''$ of the AST $\forall
\rolesFuncPair{\role R_1}{\role R_2} \in \transF({\mathcal I}'), \forall
\rolesFuncPair{\role S_1}{\role S_2} \in \transI({\mathcal I}'') \; . \;
\{\role R_1,\role R_2\} \cap \{\role S_1,\role S_2\} \neq \emptyset$.
Since $\transF({\mathcal I}')$ and $\transI({\mathcal I}'')$ have $O(n)$
elements, thanks to Lemma \ref{lemma:limit_connectedness}, checking if
${\mathcal I}'\seqOp{\mathcal I}''$ is connected 
costs $O(n)$.
Since in the AST there are less than $n$ sequence operators, checking the
connectedness 
on the whole AST costs $O(n^2)$.

The complexity of checking the connectedness of the entire AST is therefore
limited by the cost of computing functions $\transI$ and $\transF$
and of checking 
the connectedness. 
All these activities have a complexity of $O(n^2\log(n))$.
\end{proof}

\section{Proofs of Section~\ref{sec:detailed_proof}}\label{appendix:proof}

\lemmaDistinct*

\begin{proof} The proof is by induction on the number $n$ of transitions, using
	a stronger inductive hypothesis: indexes are distinct but, possibly, inside
	\AIOC{} subterms of the form $\mathcal{I};\whileKey{b \at \role
	R}{\mathcal{I'}}{i}$. In this last case, the same index can occur both in
	$\mathcal I$ and in $\mathcal I'$, attached to constructs with different
	global indexes. The statement of the Lemma follows directly: first, distinct
	indexes imply distinct global indexes; second, global indexes of $\mathcal
	I$ and of $\mathcal I'$ are distinct, since $\mathcal I'$ is inside the while
	loop whilst $\mathcal I$ is not.

	In the base case ($n=0$), thanks to well annotatedness, indexes are always
	distinct. The inductive case follows directly by induction for transitions
	with label $\tick$. Otherwise, we have a case analysis on the only axiom which
	derives a transition with label different from $\tick$.

	The only difficult cases are $\did{\AIOC}{While-unfold}$ and
	$\did{\AIOC}{Up}$.

	In the case of Rule $\did{\AIOC}{While-unfold}$, note that the while is
	enabled, hence it cannot be part of a term of the form
	$\mathcal{I};\whileKey{b \at \role R}{\mathcal{I'}}{i}$. Hence, indexes of the
	body of the while loop do not occur elsewhere. As a consequence, after the
	transition no clashes are possible with indexes in the context. Note also that
	indexes of the body of the loop are duplicated, but the resulting term has the
	form $\mathcal{I};\whileKey{b \at \role R}{\mathcal{I'}}{i}$, thus global
	indexes are distinct by construction.

 	The case $\did{\AIOC}{Up}$ follows thanks to the condition
 	$\func{freshIndexes}(\mathcal{I}')$.
\end{proof}

\lemmaEv*

\begin{proof}\hfill
\begin{enumerate}
\item By definition of projection.
\item Let $\ev_1 \leqaioc \ev_2$. We have a case analysis on the condition
  used to derive the dependency.
\begin{description}
\item [Sequentiality] Consider $\mathcal{I} = \mathcal{I}'\seqOp
\mathcal{I}''$. If  events are in the same role the implication follows from
the sequentiality of the $\leqapoc$.

Let us show that there exists an event $\ev''$ in an initial interaction of
$\mathcal{I}''$ such that either $\ev'' \leqapoc \ev_2$ or $\ev'' \leqapoc
\overline{\ev_2}$. The proof is by induction on the structure of
$\mathcal{I}''$. The only difficult case is sequential composition. Assume
$\mathcal{I}'' = \mathcal{I}_1; \mathcal{I}_2$. If $\ev_2 \in
\devent(\mathcal{I}_1)$ the thesis follows from inductive hypothesis. If
$\ev_2 \in \devent(\mathcal{I}_2)$ then by induction there exists an event
$\ev_3$ in an initial interaction of $\mathcal{I}_2$ such that $\ev_3 \leqapoc
\ev_2$ or $\ev_3 \leqapoc \overline{\ev_2}$. By synchronisation
(Definition~\ref{def:causalapoc}) we have that $\overline{\ev_3} \leqapoc
\ev_2$ or $\overline{\ev_3} \leqapoc \overline{\ev_2}$. By connectedness 
we have that $\ev_3$ or $\overline{\ev_3}$ are in the same role of
an event $\ev_4$ in $\mathcal{I}'$. By sequentiality (Definition
\ref{def:causalapoc}) we have that $\ev_4 \leqapoc \ev_3$ or $\ev_4 \leqapoc
\overline{\ev_3}$. By synchronisation we have that $\overline{\ev_4} \leqapoc
\ev_3$ or $\overline{\ev_4} \leqapoc \overline{\ev_3}$. The thesis follows
from the inductive hypothesis on $\ev_4$ and by transitivity of $\leqapoc$.

Let us also show that there exists a final event $\ev''' \in
\devent(\mathcal{I}')$ such that $\ev_1 \leqapoc \ev'''$ or $\ev_1 \leqapoc
\overline{\ev'''}$. The proof is by induction on the structure of
$\mathcal{I}'$. The only difficult case is sequential composition. Assume
$\mathcal{I}' = \mathcal{I}_1; \mathcal{I}_2$. If $\ev_1 \in
\devent(\mathcal{I}_2)$ the thesis follows from inductive hypothesis. If
$\ev_1 \in \devent(\mathcal{I}_1)$ then the proof is similar to the one
above, finding a final event in $\mathcal{I}_1$ and applying sequentiality,
synchronisation, and transitivity.

The thesis follows from the two results above again by sequentiality,
synchronisation, and transitivity.
\item [Scope] it means that either (\emph{a}) $\ev_1 = \uparrow_{\idxSign \xi}$
and $\ev_2$ is an event in the scope or (\emph{b}) $\ev_1 = \uparrow_{\idxSign
\xi}$  and $\ev_2 = \downarrow_{\idxSign \xi}$, or (\emph{c}) $\ev_1$ is an
event in the scope and $\ev_2 = \downarrow_{\idxSign \xi}$. We consider case
(\emph{a}) since case (\emph{c}) is analogous and case (\emph{b})  follows by
transitivity. If $\ev_2$ is in the coordinator then the thesis follows easily.
Otherwise it follows thanks to the auxiliary synchronisations with a reasoning
similar to the one for sequentiality.
\item [Synchronisation] it means that $\ev_1$ is a sending event and $\ev_2$
is the corresponding receiving event, namely $\ev_1 = \overline{\ev_2}$ .
Thus, since $\ev_2 \leqapoc \ev_2$ then $\overline{\ev_2} \leqapoc \ev_2$.
\item [If] it means that $\ev_1$ is the evaluation of the guard and $\ev_2$
is an event in one of the two branches. Thus, if $\ev_2$ is in the coordinator
then the thesis follows easily. Otherwise it follows thanks to the auxiliary
synchronisations with a reasoning similar to the one for sequentiality.
\item [While] it means that $\ev_1$ is the evaluation of the guard and
$\ev_2$ is in the body of the while loop. Thus, if $\ev_2$ is in the
coordinator then the thesis follows easily. Otherwise it follows thanks to the
auxiliary synchronisations with a reasoning similar to the one for
sequentiality.\qedhere
\end{description}
\end{enumerate}
\end{proof}

\lemmaIOCwell*

\begin{proof} We have to prove that $\proj(\mathcal{I},\Sigma)$ satisfies the
conditions of Definition~\ref{defin:synchwa} of well-annotated \APOC{}:
\begin{description}
\item[{\sc C1}] For each global index $\xi$ there are at most two
communication events on programmer-specified operations with global index
$\xi$ and, in this case, they are matching events.
The condition follows by the definition of the projection function, observing
that in well-annotated \AIOC{}s, each interaction has its own index, and
different indexes are mapped to different global indexes.
\item[{\sc C2}] Only events which are minimal according to $\leqapoc$ may
correspond to enabled transitions. This condition follows from
Lemma~\ref{lemma:minimalevent}.
\item[{\sc C3}] For each pair of non-conflicting sending events $[f_\xi]_{\role
R}$ and $[f_{\xi'}]_{\role R}$ on the same operation $o^?$ and with the same
target $\role R'$ such that $\xi \neq \xi'$ we have $[f_\xi]_{\role R} \leqapoc
[f_{\xi'}]_{\role R}$ or $[f_{\xi'}]_{\role{R}} \leqapoc [f_{\xi}]_{\role R}$.
Note that the two events are in the same role $\role R$, thus without loss of
generality we can assume that there exist two processes $P,P'$ such that
$[f_\xi]_{\role R} \in \event(P)$ and $[f_{\xi'}]_{\role R} \in \event(P')$ and
there is a subprocess of $\net$ of one of the following forms:

\begin{itemize}
	\item $P\seqOp P'$: the thesis follows by sequentiality
	(Definition~\ref{def:causalapoc});
	\item $P \parOpP P'$: this case can never happen for the reasons below. 
	For events on program\-mer-specified operations this follow by the definition
	of projection, since the prefixes of the names of operations are different.
	For events on auxiliary operations originated by the same construct this
	follows since all the targets are different. For events on auxiliary
	operations originated by different constructs this follows since the
	prefixes of the names of the operations are different.
	\item $\ifthen{b}{P}{P'}$: this case can never happen since the events are
	non-conflicting (Definition \ref{def:conflict}).
\end{itemize}

\item[{\sc C4}] Similar to the previous case, with receiving events instead of
sending events.

\item[{\sc C5}] If $\ev$ is an event inside a scope with global index $\xi$ then
its matching events $\overline{\ev}$ (if they exist) are inside a scope with the
same global index. This case holds by definition of the projection function.

\item[{\sc C6}] If  two events have the same index but different global indexes
then one of them, let us call it $\ev_1$, is inside the body of a while loop
with global index $\xi_1$ and the other, $\ev_2$, is not. Furthermore, $\ev_2
\leqapoc \ev_{\xi_1}$ where $\ev_{\xi_1}$ is the guarding while-event of the
while loop with global index $\xi_1$. By definition of well-annotated
\AIOC{} and of projection the only case where there are two events with the same
index but different global indexes is for the auxiliary communications in the
projection of the while construct, where the conditions hold by construction.\qedhere
\end{description}
\end{proof}

\end{document}